\documentclass[3p]{elsarticle}

\journal{Journal of Logical and Algebraic Methods in Programming}
\usepackage[utf8]{inputenc}

\DeclareMathSymbol{:}{\mathpunct}{operators}{"3A}

\usepackage{breakurl}
\usepackage{eucal}

\usepackage{centernot}
\usepackage{mdframed}

\usepackage{amsfonts} 
\usepackage{amssymb}
\usepackage{amsthm}

\usepackage{pifont}
\newcommand{\cmark}{\ding{51}}%
\newcommand{\xmark}{\ding{55}}%
\usepackage{multirow}

\usepackage{paralist}
\usepackage{mathtools}
\usepackage{url}
\usepackage{extarrows}

\usepackage{etoolbox}
\usepackage{xspace} 

\usepackage{stmaryrd} 

\usepackage[inline,shortlabels]{enumitem} 
\newlist{inlinelist}{enumerate*}{1}
\setlist*[inlinelist,1]{%
  label=(\roman*),
}

\usepackage{xifthen}  
\usepackage{ifthen}

\usepackage{nicefrac}
\usepackage{caption,subcaption}

\usepackage{color}
\usepackage[usenames,dvipsnames]{xcolor}

\usepackage{orcidlink}

\usepackage{tikz} 
\usepackage{tikz-cd} 
\usepackage{float}

\usepackage{hyperref}
\usepackage{cleveref}

\usepackage[final,nomargin,inline,index]{fixme} 
\fxusetheme{color}

\definecolor{Black}{HTML}{000000}
\definecolor{Gray}{HTML}{808080}
\definecolor{Magenta}{HTML}{FF00FF}
\definecolor{RubineRed}{HTML}{ED017D}
\definecolor{ForestGreen}{HTML}{028A0F}
\definecolor{OliveGreen}{HTML}{808000}
\definecolor{MidnightBlue}{HTML}{006795}
\definecolor{Plum}{HTML}{92268F}

\FXRegisterAuthor{bart}{anbart}{\color{magenta} {\underline{bart}}}
\FXRegisterAuthor{zun}{anzun}{\color{OliveGreen} {\underline{zun}}}
\FXRegisterAuthor{ric}{anric}{\color{red} {\underline{ric}}}

\hypersetup{
  breaklinks   = true,
  colorlinks   = true, 
  urlcolor     = blue, 
  linkcolor    = blue, 
  citecolor    = red   
}
\hypersetup{final} 

\usepackage{listings}
\definecolor{listingBG}{HTML}{FFFFCB}%
\definecolor{listingFrame}{HTML}{BBBB98}%
\definecolor{listingLineno}{rgb}{0.5,0.5,1.0}%
\crefname{lstlisting}{Listing}{Listings}  

\definecolor{LightGrey}{rgb}{0.985,0.985,0.985}
\definecolor{LstBackground}{rgb}{0.975,0.975,0.975}

\definecolor{Black}{HTML}{000000}
\definecolor{Gray}{HTML}{808080}
\definecolor{LighterGray}{HTML}{FBFBFB}
\definecolor{Magenta}{HTML}{FF00FF}
\definecolor{RubineRed}{HTML}{ED017D}
\definecolor{ForestGreen}{HTML}{028A0F}
\definecolor{Plum}{HTML}{92268F}

\lstdefinelanguage{txscript}{
	commentstyle=\color{Gray},
	morecomment=[l]{//},
	morecomment=[s]{/*}{*/},
	classoffset=0,
        escapeinside={(*}{*)},
	morekeywords={if,then,else,for,in,receive,send,pay,pays,contract,function,public,returns,skip,require,mapping,int,uint,address,token,abort,return,msg,sender,value,view},
	keywordstyle=\color{Plum},
	classoffset=1,
	morekeywords={},
	keywordstyle=\addrColor,
	classoffset=2,        
	morekeywords={origin},
	keywordstyle=\pmvColor,
	classoffset=3,        
        morekeywords={constructor},
	keywordstyle=\color{MidnightBlue}\bfseries,
	basicstyle=\fontseries{m}\footnotesize\ttfamily
	\lst@ifdisplaystyle\footnotesize\fi,
}

\lstdefinelanguage{solidity}{
, basicstyle=\ttfamily\linespread{1.15}\normalsize\lst@ifdisplaystyle\fontsize{8}{9}\selectfont\fi
, commentstyle=\color{Gray}
, morecomment=[l]{//}
, morecomment=[s]{/*}{*/}
, escapechar=\$
, escapeinside={(*}{*)}
, classoffset=0,
, keywordstyle=\color{NavyBlue}\bfseries
, morekeywords={assert,require,if,then,else,for,break,call,delegatecall,transfer,send,approve,receive,balanceOf,case, catch,continue,do,while,emit, new, return, revert, selfdestruct, try, with, throw, switch, suicide}
, classoffset=1
, keywordstyle=\color{YellowGreen}\bfseries
, morekeywords={external, implements, import, interface, internal, library, payable, pragma, private, protected, public, pure, returns, super, using, view, immutable, memory}
, classoffset=2
, keywordstyle=\color{blue}
, morekeywords={function, constructor, contract, constant, struct, address, bool, byte, bytes, bytes1, bytes2, bytes3, bytes4, bytes5, bytes6, bytes7, bytes8, bytes9, bytes10, bytes11, bytes12, bytes13, bytes14, bytes15, bytes16, bytes17, bytes18, bytes19, bytes20, bytes21, bytes22, bytes23, bytes24, bytes25, bytes26, bytes27, bytes28, bytes29, bytes30, bytes31, bytes32, enum, int, int8, int16, int24, int32, int40, int48, int56, int64, int72, int80, int88, int96, int104, int112, int120, int128, int136, int144, int152, int160, int168, int176, int184, int192, int200, int208, int216, int224, int232, int240, int248, int256, mapping, string, uint, uint8, uint16, uint24, uint32, uint40, uint48, uint56, uint64, uint72, uint80, uint88, uint96, uint104, uint112, uint120, uint128, uint136, uint144, uint152, uint160, uint168, uint176, uint184, uint192, uint200, uint208, uint216, uint224, uint232, uint240, uint248, uint256, var, void, ether, finney, szabo, wei, days, hours, minutes, seconds, weeks, years, token}
, classoffset=3
, keywordstyle=\color{Plum}\bfseries
, morekeywords={balance, block, blockhash, instanceof, coinbase, difficulty, gaslimit, number, timestamp, msg, data, gas, sender, value, sig, value, now, tx, gasprice, origin}
}

\newcommand{\ifempty}[3]{%
  \ifthenelse{\isempty{#1}}{#2}{#3}%
}

\newcommand{\ifdots}[3]{%
  \ifthenelse{\equal{#1}{...}}{#2}{#3}%
}

\newtoggle{hidden}
\toggletrue{hidden}
\newcommand{\hidden}[1]{\iftoggle{hidden}{}{{\color{red}#1}}}

\makeatletter
\newcommand*{\itemequation}[3][]{%
  \item
  \begingroup
    \refstepcounter{equation}%
    \ifx\\#1\\%
    \else  
      \label{#1}%
    \fi
    \sbox0{#2}%
    \sbox2{$\displaystyle#3\m@th$}%
    \sbox4{\@eqnnum}%
    \dimen@=.5\dimexpr\linewidth-\wd2\relax
    \ifcase
        \ifdim\wd0>\dimen@
          \z@
        \else
          \ifdim\wd4>\dimen@
            \z@
          \else 
            \@ne
          \fi 
        \fi
      \@latex@warning{Equation is too large}%
    \fi
    \noindent   
    \rlap{\copy0}%
    \rlap{\hbox to \linewidth{\hfill\copy2\hfill}}%
    \hbox to \linewidth{\hfill\copy4}%
    \hspace{0pt}
  \endgroup
  \ignorespaces 
}
\makeatother

\newcommand{\Real}[1]{\mathrm{Real}}

\newcommand{\codefont}{\fontsize{10}{10}\selectfont}
\newcommand{\code}[1]{{\tt\codefont{#1}}}
\newcommand{\contract}[2][]{{\tt\codefont{\cmvColor{#2_{#1}}}}}
\newcommand{\interface}[2][]{{\tt\codefont{\intColor{#2_{#1}}}}}
\newcommand{\txcode}[1]{{\tt\codefont{\color{Black}{#1}}}}

\newcommand{\callee}[1]{{\cmvColor{\it callee}({#1})}}
\newcommand{\deps}[1]{{\cmvColor{\it deps}}({#1})}

\newcommand{\Eg}{E.g.\@\xspace}
\newcommand{\eg}{e.g.\@\xspace}
\newcommand{\ie}{i.e.\@\xspace}
\newcommand{\wrt}{w.r.t.\@\xspace}

\newcommand{\Wlog}{W.l.o.g.\@\xspace}

\renewcommand{\epsilon}{\varepsilon}
\newcommand{\emptyseq}{\epsilon}

\theoremstyle{plain}
\newtheorem{theorem}{Theorem}
\newtheorem{lemma}[theorem]{Lemma}
\newtheorem{proposition}[theorem]{Proposition}

\theoremstyle{definition}
\newtheorem{definition}{Definition}
\newtheorem{example}{Example}

\newenvironment{proofof}[2][]{%
  \ifempty{#1}
  {\subsection*{Proof of~\Cref{#2}}}
  {\subsection*{Proof of~\Cref{#2} ({#1})}}
  \label{#2-proof}%
  }%
  {}

\def\addrColor{\color{magenta}}
\newcommand{\addrFmt}[1]{{\addrColor{\tt #1}}}
\newcommand{\addr}[2][]{\addrFmt{#2}_{\addrColor{#1}}\xspace}
\newcommand{\AddrA}[1][]{{\addrColor{\mathcal{A}_{#1}}}} 
\newcommand{\AddrB}[1][]{{\addrColor{\mathcal{B}_{#1}}}} 

\newcommand{\AddrU}[1][]{\addrFmt{\mathbb{A}}_{\addrColor{#1}}} 

\newcommand{\cmvOfcst}[1]{\dag{#1}}
\newcommand{\strip}[2]{{#1}\!\upharpoonright_{#2}}

\def\pmvColor{\color{red}}
\newcommand{\pmvFmt}[1]{{\pmvColor{\tt #1}}}
\newcommand{\pmv}[2][]{\pmvFmt{#2}_{\pmvColor{#1}}\xspace}
\newcommand{\pmvA}[1][]{\pmv[{#1}]{A}} 
\newcommand{\pmvB}[1][]{\pmv[{#1}]{B}}

\newcommand{\pmvM}[1][]{\pmv{M}_{\pmvFmt{#1}}} 

\newcommand{\PmvM}[1][]{\pmvFmt{\mathcal{M}}_{\pmvColor{#1}}} 

\newcommand{\Adv}{\PmvM} 
\newcommand{\PmvU}{\pmvFmt{\mathbb{A}}_{\pmvColor{u}}} 

\def\cmvColor{\color{blue}}
\def\intColor{\color{YellowGreen}} 
\newcommand{\cmvFmt}[1]{{\cmvColor{\code{#1}}}}
\newcommand{\cmv}[2][]{\cmvFmt{#2}_{\cmvColor{#1}}}
\newcommand{\cmvC}[1][]{\cmv[#1]{C}} 
\newcommand{\cmvCi}[1][]{\cmvFmt{C'_{\cmvColor{{\rm #1}}}}}

\newcommand{\cmvD}[1][]{\cmv[#1]{D}} 

\newcommand{\CmvC}[1][]{\cmv[#1]{\mathcal{C}}} 
\newcommand{\CmvCi}[1][]{\cmvFmt{\mathcal{C}'_{\cmvColor{{\rm #1}}}}}
\newcommand{\CmvD}[1][]{\cmv[#1]{\mathcal{D}}} 
\newcommand{\CmvDi}[1][]{\cmvFmt{\mathcal{D}'_{#1}}} 

\newcommand{\CmvU}[1][]{\cmvFmt{\mathbb{A}}_{\cmvColor{c}}} 

\def\cstColor{\color{MidnightBlue}}
\newcommand{\cstFmt}[1]{{\cstColor{#1}}}
\newcommand{\cst}[2][]{\cstFmt{#2}_{\cstColor{#1}}}
\newcommand{\cstC}[1][]{\cst[#1]{\Gamma}} 
\newcommand{\cstCi}[1][]{\cstFmt{\Gamma'_{\cstColor{{\rm #1}}}}}

\newcommand{\cstD}[1][]{\cst[#1]{\Delta}} 
\newcommand{\cstDi}[1][]{\cstFmt{\Delta'_{\cstColor{{\rm #1}}}}}

\newcommand{\cstS}[1][]{\cst[#1]{\Sigma}} 

\def\txColor{\color{MidnightBlue}}

\newcommand{\txFmt}[1]{{\txColor{\sf #1}}}

\newcommand{\tx}[2][]{\txFmt{#2}_{\txColor{#1}}}
\newcommand{\txT}[1][]{\tx[#1]{X}} 
\newcommand{\txY}[1][]{\tx[#1]{Y}} 
\newcommand{\txTi}[1][]{\txFmt{X'_{\txColor{{\it #1}}}}}

\newcommand{\TxTS}[1][]{\vec{\tx[#1]{\mathcal{X}}}} 
\newcommand{\TxYS}[1][]{\vec{\tx[#1]{\mathcal{Y}}}} 

\DeclareMathAlphabet{\mathbfsf}{\encodingdefault}{\sfdefault}{bx}{n}

\newcommand{\dom}[1]{\operatorname{dom} {#1}}

\newcommand{\Nat}{\mathbb{N}}

\newcommand{\setenum}[1]{\{#1\}}
\newcommand{\setcomp}[2]{\left\{{#1} \,\middle|\, {#2}\right\}}

\newcommand{\wmvA}[1][]{w_{#1}}

\newcommand{\WmvA}[1][]{W_{#1}}
\newcommand{\WmvAi}[1][]{W'_{#1}}
\newcommand{\WmvAii}[1][]{W''_{#1}}

\newcommand{\wal}[2]{{\mathit{\omega}_{#1}({#2})}}

\newcommand{\waltok}[2]{#1:#2}
\newcommand{\walenum}[1]{[#1]}
\newcommand{\walpmv}[2]{{#1}\walenum{#2}}
\newcommand{\walu}[3]{\walpmv{#1}{\waltok{#2}{#3}}}

\newcommand{\waldistrarrow}[1]{\approx_{\$}}

\newcommand{\wealth}[2]{\$_{#1}({#2})}
\newcommand{\idxfun}[1]{\mathbf{1}_{#1}}
\newcommand{\price}[1]{\$\idxfun{#1}}

\newcommand{\gain}[3]{\mathit{\gamma}_{#1}\ifempty{#2}{}{({#2},{#3})}}

\newcommand{\mall}[2]{{\kappa_{#1}\ifempty{#2}{}{({#2})}}}

\newcommand{\mev}[3]{{\mathrm{MEV\!}_{#1}\ifempty{#1#2#3}{}{({#2\ifempty{#3}{}{,#3}})}}}
\newcommand{\lmev}[3]{{\mathrm{MEV\!}_{#1}\ifempty{#1#2#3}{}{({#2\ifempty{#3}{}{,#3}})}}}
\newcommand{\rlmev}[3]{{\mathrm{MEV\!}^{\,\infty}_{#1}\ifempty{#1#2#3}{}{({#2\ifempty{#3}{}{,#3}})}}}

\newcommand{\nonintrel}{\mathrel{\not\rightsquigarrow}}
\newcommand{\nonint}[2]{\ifempty{#1}{\nonintrel}{{#1} \nonintrel {#2}}}
\newcommand{\negnonint}[2]{\ifempty{#1}{\rightsquigarrow}{{#1} \rightsquigarrow {#2}}}
\newcommand{\richnonintrel}{\mathrel{\not\rightsquigarrow^{\infty}}}
\newcommand{\richnonint}[2]{\ifempty{#1}{\richnonintrel}{{#1} \richnonintrel {#2}}}
\newcommand{\negrichnonint}[2]{\ifempty{#1}{\rightsquigarrow^{\infty}}{{#1} \rightsquigarrow^{\infty} {#2}}}

\newcommand{\qedex}{\ensuremath{\diamond}}

\definecolor{LightGrey}{rgb}{0.95,0.95,0.95}
\definecolor{keyword}{HTML}{7F0055}

\def\tokColor{\color{ForestGreen}}
\newcommand{\tokFmt}[1]{{\tokColor{\tt #1}}}
\newcommand{\ETH}{\tokFmt{ETH}}

\newcommand{\tok}[2][]{\tokFmt{#2}_{\tokColor{#1}}\xspace}

\newcommand{\tokT}[1][]{\tok[{#1}]{T}}    
\newcommand{\tokTi}[1][]{\tok[{#1}]{T'}}
\newcommand{\TokU}{\tokFmt{\mathbb{T}}} 

\newlength\replength
\newcommand\repfrac{.1}

\newcommand\rulewidth{.6pt}
\newcommand\tdashfill[1][\repfrac]{\cleaders\hbox to \replength{%
  \smash{\rule[\arraystretch\ht\strutbox]{\repfrac\replength}{\rulewidth}}}\hfill}

\newcommand\tdotfill[1][\repfrac]{\cleaders\hbox to \replength{%
  \smash{\raisebox{\arraystretch\dimexpr\ht\strutbox-.1ex\relax}{.}}}\hfill}

\def\sysColor{\color{Black}}

\newcommand{\sysFmt}[1]{{\sysColor{#1}}}
\newcommand{\sysS}[1][]{\mathord{\sysFmt{S}_{\sysColor{#1}}}}

\newcommand{\sysSi}[1][]{\mathord{\sysColor{\sysS'_{#1}}}}

\newtoggle{arxiv}
\toggletrue{arxiv}

\begin{document}

\hyphenation{block-chain}
\hyphenation{Block-chain}
\hyphenation{block-chains}

\begin{frontmatter}

\title{A formal framework for the economic security of DeFi compositions}

\author[1]{Massimo Bartoletti\orcidlink{0000-0003-3796-9774}} 
\author[2]{Riccardo Marchesin}
\author[2]{Roberto Zunino\orcidlink{0000-0002-9630-429X}}

\cortext[bart]{\emph{Corresponding author}. Dipartimento di Matematica e Informatica, Universit\`a degli Studi di Cagliari, via Ospedale 72, 09124 Cagliari (Italy), e-mail: \texttt{bart@unica.it}}

\affiliation[1]{organization={University of Cagliari},
            city={Cagliari},
            country={Italy}}

\affiliation[2]{organization={University of Trento},
            city={Trento},
            country={Italy}}




\begin{abstract}
Decentralized Finance (DeFi) services are usually constructed by composing a variety of smart contracts.
While composability is a key driver of the success of DeFi, it also creates security risks: adversaries may exploit interactions between newly deployed contracts and the pre-existing ones to inflict economic losses.
We introduce \emph{MEV non-interference}, a formal security notion for DeFi composability requiring that the maximal extractable value from a set of newly deployed contracts is not increased by interactions with the existing blockchain state. 
To support this notion, we define \emph{local MEV}, a novel measure of economic attacks that focusses on the loss of a given set of victim contracts. 
We study two adversarial models, with bounded and unbounded wealth, and establish sufficient conditions and locality principles that enable modular reasoning about secure composability. 
We apply the framework to representative DeFi compositions, including exchanges, AMMs, options, lending pools, routers, and arbitrage contracts, showing how it distinguishes secure compositions from vulnerable ones. 
Our results provide a formal foundation for reasoning about the economic security of DeFi compositions.
\end{abstract}

\end{frontmatter}

\section{Introduction}
\label{sec:intro}

Decentralized Finance (DeFi) is often described as ``money Lego'' because it enables the construction of complex financial services through the composition of simpler building blocks~\cite{defipulse,Werner22aft}.
Recent empirical analyses of the DeFi ecosystem confirm that
DeFi protocols are deeply interconnected in practice,
giving rise to highly intricate dependency structures and interactions~\cite{Kitzler22fc,Kitzler23tweb}.
While such composability is a key driver of the success of DeFi, it also introduces significant security challenges:
adversaries can exploit unintended interaction among protocols
to obtain an economic profit to the detriment of honest
users~\cite{Daian20flash,Gudgeon2020cvcbt}.
These risks are further exacerbated by emerging platforms that allow users to seamlessly create arbitrary compositions of
DeFi protocols~\cite{furucombo}.
This raises a fundamental and still largely unanswered question: \emph{under which conditions is a DeFi composition secure?}

Despite the clear practical relevance of this question,
research on DeFi composability remains surprisingly limited.
The first notion of secure DeFi composition
in the scientific literature is the one introduced by 
Babel, Daian, Kelkar and Juels in their seminal 
``Clockwork finance'' paper~\cite{Babel23clockwork}.
The focus of this work is on attacks in which adversaries exploit
newly deployed contracts to increase their profit opportunities.
Accordingly, their notion of composability requires that augmenting a blockchain state $\sysS$ with a set of contracts $\cstD$ should not allow adversaries a to extract a substantially larger
Maximal Extractable Value (MEV).
In formulae, letting $\mev{}{\sysS}{}$ denote the 
maximal value that adversaries can extract from $\sysS$,
and writing $\sysS \mid \cstD$ for the blockchain state $\sysS$
extended with the contracts $\cstD$, 
their criterion is expressed as:
\begin{equation}
  \label{eq:babel-composability}
  \text{$\cstD$ is $\epsilon$-composable in $\sysS$}
  \quad \text{iff} \quad
  \mev{}{\sysS \mid \cstD}{} \leq (1 + \epsilon) \ \mev{}{\sysS}{}
\end{equation}
where $\epsilon$ parameterises the ``not substantially larger'' condition above.
For example, let $\cstD$ be a contract that allows users
to bet on the future price of a token,
using an Automated Market Maker (AMM) in $\sysS$ as a price oracle.
The key issue is that the AMM prices are not fixed: they are determined by the current token reserves held by the AMM and therefore fluctuate in response to trades. 
In particular, large trades can induce substantial price fluctuations.
If the adversary has sufficient capital in $\sysS$,
criterion~\eqref{eq:babel-composability} would correctly classify
$\cstD$ as \emph{not} $0$-composable in $\sysS$. 
Indeed, the adversary can strategically manipulate the AMM price by executing large trades, thereby creating artificial price fluctuations that allow them to systematically win the bet. 
As a result, the adversary can extract strictly more MEV from $\sysS \mid \cstD$ than from $\sysS$ alone.

\paragraph{Limitations of $\epsilon$-composability}

We argue that the $\epsilon$-composability has some limitations.

First, using the $\mev{}{}{}$ of the \emph{entire} blockchain state $\sysS$ as a baseline for comparison makes the security guarantee provided by $\epsilon$-composability difficult to interpret.
One might naturally conclude that deploying a contract $\cstD$ in $\sysS$ is safe after establishing that $\cstD$ is $0$-composable with $\sysS$. 
However, $0$-composability only guarantees that the total MEV of the extended state $\sysS \mid \cstD$ does not exceed that of $\sysS$; it does not constrain \emph{where} this MEV is extracted from. 
To illustrate, suppose that the total value of tokens in $\cstD$ is equal to the MEV of $\sysS$, and that extracting value from $\cstD$ eliminates the MEV opportunities previously available in $\sysS$. Then $\cstD$ is still $0$-composable with $\sysS$, even though an attack on $\sysS \mid \cstD$ may drain the entire value of $\cstD$
(see~\Cref{ex:babel-composability-not-implies-mev-nonintererence} for a concrete scenario). 
From the perspective of the designer or users of $\cstD$, such a composition would hardly be considered secure.

A second drawback of using the \emph{global} MEV as a baseline
is that~\eqref{eq:babel-composability} classifies
as \emph{not} composable contracts whose MEV is intentional and does not arise from any interaction with the rest of the system.
Consider, for instance, an airdrop contract $\cstD$ that allows anyone to withdraw its entire balance.
Despite being completely independent from the rest of the system,
$\cstD$ is not $0$-composable with any $\sysS$,
since deploying it necessarily increases the total MEV of the system, \ie, $\mev{}{\sysS \mid \cstD}{} > \mev{}{\sysS}{}$.
One could attempt to recover composability by choosing a sufficiently large value of $\epsilon$, but this is both cumbersome and conceptually unsatisfactory. In the extreme case where $\mev{}{\sysS}{} = 0$, the contract is not $\epsilon$-composable for any $\epsilon \geq 0$, regardless of its complete independence from the rest of the system.

Finally, the reliance on global MEV poses both algorithmic and usability concerns.
From a computational standpoint, evaluating~\eqref{eq:babel-composability} requires reasoning about a property of the entire blockchain state, which may be prohibitively large in practice. 
From a mitigation standpoint, when $\epsilon$-composability fails because additional MEV can be extracted from contracts already present in $\sysS$, it is unclear how the designer of $\cstD$ should react. The vulnerable contracts reside in the pre-existing state and are therefore outside the control of the deployer. Consequently, the definition may identify a composability problem without providing actionable guidance on how to address it.


\paragraph{Secure composability as non-interference}

The limitations of $\epsilon$-composability suggest the need for composability notions that do not rely on the global MEV of the blockchain state.

Our starting point is the notion of \emph{non-interference}, which has been extensively studied in information-flow security since the 1980s~\cite{GoguenMeseguer82sp,Ryan99csfw,Ryan01csfw,Backes02esorics,Giacobazzi04popl}. 
In its classical form, non-interference requires that 
adversaries interacting with a software system cannot observe private data.
More precisely, the property holds when
public outputs (that can be observed by adversaries)
are not affected by private inputs.

We adapt this principle to the setting of smart contracts composition, ending up with a notion that we call \emph{MEV non-interference}.
Intuitively, the contracts $\cstD$ being deployed should not become more vulnerable merely because they are composed with a larger blockchain state. 
In other words, the maximal economic loss that adversaries can inflict on $\cstD$ should only depend on the interactions of the adversary with~$\cstD$, and not with the rest of the system. 
Under this interpretation, the interactions between the adversary and the surrounding system $\sysS$ play the role of the private inputs in classic non-interference, while the economic loss suffered by $\cstD$ plays the role of the public outputs. 
MEV non-interference therefore requires that interacting with contracts in $\sysS$ does not enable adversaries to cause more loss on $\cstD$ than they could achieve by interacting with $\cstD$ alone.

We argue that MEV non-interference captures the security property that protocol designers are naturally interested in when deploying new contracts. By focusing on the \emph{local} MEV extractable from $\cstD$, 
our notion overcomes the above-mentioned limitations of $\epsilon$-composability.
First, contracts whose MEV is intentional are not necessarily classified as insecure. 
For example, deploying the above-mentioned airdrop contract in any $\sysS$ is considered secure by MEV non-interference. 
%
Second, MEV non-interference admits local reasoning: under suitable conditions, one can establish the composability of $\cstD$ by reasoning only about $\cstD$ and its dependencies, without analysing the entire blockchain state. 




\paragraph{Contributions}

This paper introduces a formal framework for reasoning about the secure composability of smart contracts. We summarise our main contributions as follows:

\begin{itemize}

\item \textbf{MEV non-interference.}
We adapt the classical notion of non-interference from information-flow security to the smart contracts setting, obtaining a novel security criterion for protocol composition. 
To the best of our knowledge, this is the first application of non-interference principles to the study of smart contracts composability.

\item \textbf{Local MEV.}
We introduce \emph{local MEV}, a novel measure of economic attacks that quantifies the maximal loss that adversaries can inflict on a given set of victim contracts. In contrast to classical MEV~\cite{Babel23clockwork,BZ25fc}, which is defined for an entire blockchain state, local MEV focuses on the contracts whose security is under analysis. This shift from a global to a local perspective is the key ingredient enabling our notion of MEV non-interference.

\item \textbf{Two adversarial models.}
We define two variants of local MEV. The first, $\lmev{}{\sysS}{\CmvC}$, models adversaries whose wealth is bounded by their capital available in the blockchain state. 
The second, $\rlmev{}{\sysS}{\CmvC}$, models adversaries with unbounded capital, capturing attacks enabled by mechanisms such as flash loans~\cite{Qin21fc}.

\item \textbf{MEV non-interference.}
Building on local MEV, we introduce two notions of secure composability: $\nonint{\sysS}{\cstD}$ for bounded-wealth adversaries and $\richnonint{\cstC}{\cstD}$ for unbounded-wealth adversaries (where $\cstC$ denotes a system of contracts obtained by removing users' wallets from the blockchain state). 
Intuitively, these properties require that composing $\cstD$ with an existing blockchain state does not create additional opportunities to extract value from $\cstD$.

\item \textbf{Local reasoning principles.}
We establish locality results showing that both local MEV and MEV non-interference can be analysed by considering only a suitable fragment of the blockchain state. 
For local MEV, we identify conditions under which the maximal loss suffered by the victim contracts depends only on those contracts and their dependencies (\Cref{th:mev:callable-narrowing:deps,th:rich-mev:callable-narrowing:deps} for bounded- and unbounded-wealth adversaries, respectively). 
Building on these results, we prove that MEV non-interference is preserved for suitable state narrowings (\Cref{th:nonint:state-narrowing,th:rich-nonint:state-narrowing}).
Intuitively, we can remove from the state the non-dependencies of $\cstD$, provided that the contracts in $\cstD$ are sender-agnostic, \ie their behaviour is uniform with respect to their callers.
For bounded-wealth adversaries, we additionally require that the removed contracts do not induce token flows to $\cstD$ or its dependencies. 
These locality principles provide both conceptual and algorithmic advantages over approaches based on global MEV.

\item \textbf{Sufficient conditions for secure composability.}
We establish practical sufficient conditions for MEV non-interference (\Cref{th:nonint:sufficient-conditions,th:rich-nonint:sufficient-conditions}). 
A key ingredient is \emph{observational invariance}, a property ensuring that adversarial interactions with the existing blockchain state cannot influence what the newly deployed contracts observe. Combined with suitable conditions on token flows, observational invariance guarantees that composing a contract with an existing system does not create additional opportunities for adversarial value extraction. 

\item \textbf{Analysis of representative DeFi compositions.}
We validate our framework on a collection of representative DeFi compositions, including exchanges, AMMs, binary options, lending pools, swap routers, and arbitrage protocols (\Cref{tab:defi-compositions}). 
Using the sufficient conditions and locality principles developed in the paper, we establish secure composability results for several classes of protocol compositions and identify concrete scenarios in which composability fails because of exploitable dependencies, token flows, or manipulable price oracles. The resulting classification highlights the practical usefulness of MEV non-interference as a tool for analysing real-world DeFi ecosystems.

\end{itemize}
\section{Blockchain model}
\label{sec:blockchain}

\begin{table}[t]
\caption{Summary of notation.}
\label{tab:notation}
\centering
\begin{tabular}{p{50pt}p{150pt}p{50pt}p{150pt}}
\hline
$\pmvA,\pmvB$ & User accounts 
& $\AddrA,\AddrB$ & Sets of [user$\mid$contract] accounts
\\
$\cmvC,\cmvD$ & Contract accounts 
& $\addr{a}, \addr{b}$ & User$\mid$contract accounts
\\
$\CmvC,\CmvD$ & Sets of contract accounts
& $\deps{\sysS}$ & Dependencies of $\sysS$
\\
$\tokT,\tokTi$ & Token types 
& $\price{\tokT}$ & Price of $\tokT$
\\
$\txT,\txTi$ & Transactions 
& $\TxTS$ & Sequence of transactions
\\
$\sysS,\sysSi$ & Blockchain states
& $\cmvOfcst{\sysS}$ & Contract accounts in $\sysS$
\\
$\WmvA,\WmvAi$ & User accounts states 
& $\wal{\AddrA}{\sysS}$ & Wallets of $\AddrA$ in $\sysS$
\\
$\cstS,\cstC,\cstD$ & Contract accounts states 
& $\wealth{\AddrA}{\sysS}$ & Wealth of $\AddrA$ in $\sysS$
\\
$\Adv$ & Set of adversaries
& $\gain{\AddrA}{\sysS}{\TxTS}$ & $\AddrA$'s gain upon firing $\TxTS$ in $\sysS$ 
\\
\hline
\end{tabular}
\end{table}

We fix a model of account-based blockchains and smart contracts inspired by Ethereum.
To keep our theory general, we abstract from some of Ethereum's peculiar features, such as its underlying consensus protocol, gas mechanism, and EVM bytecode, just focussing on the high-level behaviour of smart contracts.
We discuss the abstractions of our model in~\Cref{sec:conclusions}.

\paragraph{Tokens}

We assume a set $\TokU$ of \emph{token types} ($\tokT, \tokTi, \ldots$), which represent the types of exchanged (fungible) assets.
While Ethereum handles its native crypto-currency (the $\ETH$) differently from the custom assets (smart contracts implementing the ERC-20 interface), we provide a uniform treatment of assets, handling both $\ETH$ and other token types with the same set of operators.

\paragraph{Accounts}

We assume a countably infinite set $\AddrU$ of \emph{accounts}, partitioned into:
\begin{itemize}

\item \emph{user accounts} $\pmvA, \pmvB, \ldots \in \PmvU$, which represent the active entities who can initiate \emph{transactions} 
(the so-called ``externally owned'' accounts in Ethereum);

\item \emph{contract accounts} $\cmvC, \cmvD, \ldots \in \CmvU$, which represent the passive entities that can react to transactions initiated by user accounts.

\end{itemize}

We use $\addr{a}, \addr{b}, \ldots$ as meta-variables ranging over both user accounts and contract accounts.
We designate a subset $\Adv$ of user accounts as \emph{adversaries}, \ie, users with the ability to influence the selection and ordering of transactions in the blockchain.
In the Ethereum ecosystem, this subset includes, \eg, block proposers and MEV searchers~\cite{flashbots}.
Although the analysis of a specific contract may depend on the particular choice of $\Adv$, there exist natural and robust choices --- such as a finite set of fresh user accounts not already appearing in the contract code or state --- for which the resulting analysis is insensitive to this choice.

\paragraph{Wallets}

A \emph{wallet} \mbox{$\wmvA \in \TokU \to \Nat$} is a finite map from token types to non-negative integers, quantifying the amount of tokens held by an account.
We represent a wallet as a comma-separated list of associations between a token type $\tokT$ and its amount $n$, written as $n:\tokT$.
We denote by $+$ the pointwise summation of wallets.

\paragraph{Blockchain states}

A blockchain state $\sysS$ is a finite map from accounts to their respective states.
The state of a user account is a wallet, while that of
a contract account is a pair $(\wmvA, \sigma)$,
where $\wmvA$ is the contract's wallet and $\sigma$ is its persistent storage, represented as a key–value store.
We denote by:
\begin{itemize}

\item $\cmvOfcst{\sysS}$ the set of contract accounts
in $\sysS$, \ie $\cmvOfcst{\sysS} = (\dom{\sysS}) \cap \CmvU$;

\item $\wal{\AddrA}{\sysS}$ the summation of the wallets of accounts $\AddrA$ in $\sysS$.

\end{itemize}

Since the primary focus of this paper is the composability of contracts, it is often necessary to manipulate blockchain states so as to add or remove accounts. 
This, in turn, requires a suitable notation for representing the individual components of a state, as well as for constructing and deconstructing states in a systematic manner.
For this purpose, we adopt the following notation:
\begin{itemize}

\item We denote by $\walpmv{\pmvA}{\wmvA}$ the state of a user account $\pmvA$ whose wallet is $\wmvA$.

\item We denote by $\walpmv{\contract{C}}{\wmvA,\sigma}$ the state of a contract account $\contract{C}$ whose wallet is $\wmvA$ and storage is $\sigma$.
A storage is represented as a comma-separated list of associations between keys and values, each written as $\code{k} = v$.

\item We denote by \mbox{$\sysS \mid \sysSi$} the composition of two blockchain states $\sysS$ and $\sysSi$. 
Such composition is defined whenever $\dom{\sysS}$ and $\dom{\sysSi}$ are disjoint --- namely, a blockchain state cannot contain multiple states for the same account. 

\item We use $\WmvA, \WmvAi, \ldots$ to denote states comprising only user accounts.

\item We use $\cstS, \cstC, \cstD, \ldots$ to denote states comprising only contract accounts.

\end{itemize}

\paragraph{Contracts}

For generality, rather than assuming a specific syntax for the contract language, we model contracts abstractly as public APIs whose methods have the following behaviour:
\begin{enumerate}

\item receive parameters from the caller. In particular, the parameters include the immediate caller and the transaction initiator (these accounts may coincide);

\item perform a finite sequence of \emph{effects} drawn from the following kinds:
\begin{itemize}
    \item abort execution when specified conditions are not met;
    
    \item read and update the contract storage;
    
    \item send tokens from the contract to another account.
      Its behaviour depends on the whether the recipient is a user or a contract account:
      \begin{itemize}
      \item If the recipient is a user account and the sender's wallet contains a sufficient amount of tokens, then the tokens are immediately transferred to the recipient, possibly  creating the account if it does not yet exist in the blockchain state.
      \item If the recipient is a contract account, then the operation is interpreted as a mere \emph{approval} of the transfer, while the actual transfer is triggered by an explicit \emph{receive} action of the recipient.  
      \end{itemize}



    \item receive tokens from the caller, provided that it has previously approved the transfer;

    \item invoke methods of the same or of other contracts;

\end{itemize}
\item return a value to the caller.
\end{enumerate}



Note that the aforementioned contract capabilities do not include the ability to directly inspect the state of external accounts, including their wallets. 
This, however, is not a practical limitation. 
Dependence on a user's wallet can be enforced by requiring the caller of a method to provide tokens to the contract, while dependence on the state of another contract can be realised by querying that state through its public getter methods.
For simplicity, we also forbid direct token transfers between user accounts.
This restriction entails no loss of generality, since such transfers can always be mediated through suitable contracts.


\paragraph{Contract dependencies}

We say that a contract $\cmvC$ is \emph{called by} $\cmvD$ if some method of $\cmvD$ invokes a method of $\cmvC$. 
We denote by $\sqsubseteq$ the reflexive and transitive closure of this relation.
Given a set of contracts $\CmvC$, we denote by $\deps{\CmvC}$ the set of \emph{dependencies} of $\CmvC$, namely, the set of contracts called by some contract in $\CmvC$. 
Formally:
\[
\deps{\CmvC} 
\; = \;
\setcomp{\cmvCi}{\exists \cmvC \in \CmvC.\ \cmvCi \sqsubseteq \cmvC}
\]
We write $\deps{\sysS}$ as a shorthand for $\deps{\cmvOfcst{\sysS}}$, \ie the dependencies of the contracts occurring in $\sysS$.

\paragraph{Well-formedness}
We say that a blockchain state
\(
\sysS \; = \; \addr[1]{a}[\cdots] \mid \cdots \mid \addr[n]{a}[\cdots]
\)
is \emph{well-formed} if
\begin{inlinelist}
\item \label{eq:well-formedness:partial-order} 
$\sqsubseteq$ is a partial order and
\item \label{eq:well-formedness:closed-downwards}
\(
  \forall i,\cmvD.\ 
  \cmvD \sqsubseteq \addr[i]{a} 
  \implies
  \exists j \leq i.\ \cmvD = \addr[j]{a}
\).
\end{inlinelist}
Condition~\ref{eq:well-formedness:partial-order} requires that contracts can only call previously deployed contracts or themselves. 
Condition~\ref{eq:well-formedness:closed-downwards} means that all the dependencies of the contracts in $\sysS$ are included in $\sysS$, and if a contract $\cmvC=\addr[i]{a}$ calls $\cmvD$, then $\cmvD$ appears on the left of $\cmvC$ in the enumeration, \ie $j \leq i$.
Note that if $\sysS$ is well-formed, then:
\begin{itemize}

\item $\sysS$ is closed under dependencies, \ie $\deps{\sysS} \subseteq \cmvOfcst{\sysS}$;
\item mutual calls between contracts are forbidden;
\item when deconstructing $\sysS$ into $\sysS[0] \mid \sysS[1]$, the leftmost component $\sysS[0]$ is guaranteed to be well-formed, but
the rightmost component $\sysS[1]$ is \emph{not}, since some of the dependencies of $\sysS[1]$ may be in $\sysS[0]$.

\end{itemize}

Hereafter, we consider well-formed states up-to permutations of their individual components, provided that the well-formedness condition is respected.
Unless otherwise stated, in our statements we implicitly assume that the blockchain states mentioned (in hypothesis or thesis) are well-formed.

\paragraph{Transactions}
We model blockchains as reactive systems whose behaviour is specified as a 
deterministic transition relation $\xrightarrow{}$ between blockchain states,
and transitions are labelled by
\emph{transactions} $\txT, \txTi, \ldots$ sent by users.
We denote by:
\[
\pmvA : \contract{C}.\txcode{f}(\code{args})
\]
a transaction sent by $\pmvA$, where
$\contract{C}$ is the called contract,
$\txcode{f}$ is the called method,
and $\code{args}$ is the list of actual parameters.
%
Given $\txT = \pmvA:\contract{C}.\txcode{f}(\code{args})$,
we write $\callee{\txT}$ for the called contract~$\contract{C}$.
Note that, although a transaction must originate from a user account,
it can trigger additional method calls where the caller is a contract account.
In general, a transaction can involve a sequence of calls of the form: 
\begin{align*}
& 
\pmvA[0]:\contract[0]{D}.\txcode{f_{0}}(\code{args_{0}})
\;\;
\contract[1]{C}:\contract[1]{D}.\txcode{f_{1}}(\code{args_{1}})
\;\;
\contract[2]{C}:\contract[2]{D}.\txcode{f_{2}}(\code{args_{2}})
\;\; \cdots \;\;
\contract[n]{C}:\contract[n]{D}.\txcode{f_{n}}(\code{args_{n}})
\end{align*}
We refer to the method calls
$\contract[i]{C}:\contract[i]{D}.\txcode{f_{i}}(\code{args_{i}})$
as \emph{internal} calls, to their callers $\contract[i]{C}$ as \emph{senders}, and to $\pmv[0]{A}$ as the \emph{origin}.
Each method $\txcode{f_i}$ in this sequence can refer to the origin and to the sender; note that these two accounts coincide only for $i=0$, while the senders for $i>0$ are contract accounts.

We assume that invalid transactions are \emph{reverted}, namely: 
\[
\sysS \xrightarrow{\;\txT\;} \sysS
\tag*{\text{if $\txT$ is not valid in $\sysS$}}
\]
There are several possible causes of invalidity, such as attempting to call a non-existent method, triggering an abort, performing an undefined operation (\eg, division by zero), or attempting a token transfer without sufficient balance.

\begin{figure}[t!]
  \begin{lstlisting}[language=solidity
  ,classoffset=5,morekeywords={setPrice,getRate,swap},keywordstyle=\color{Black}
  ,frame=single
  ,caption={A price oracle contract.}
  ,label={lst:price-oracle}]
interface IPriceOracle {
  // exchange rate to sell ti and buy to (multiplied by 1e6)
  function getRate(token ti, token to) external returns (uint);
}

contract PriceOracle is IPriceOracle {
  mapping (token => uint) prices;
  mapping (token => bool) isPriced;
  address owner;
  uint public immutable fee;
  
  constructor(uint fee_) { owner=msg.sender; fee=fee_; }

  function setPrice(token t, uint p) public {
    require (msg.sender==owner); // only the owner can set the price
    isPriced[t]=true;
    prices[t]=p; 
  }
  
  function getRate(token ti, token to) public returns(uint) {
    require (isPriced[ti] && isPriced[to]);
    ETH.receive(fee); // transfer fee:ETH from msg.sender to this
    return (prices[ti] * 1_000_000)/prices[to]; 
  } 
}
\end{lstlisting}
\end{figure}

\begin{figure}[t!]
  \begin{lstlisting}[language=solidity
  ,classoffset=5,morekeywords={setPrice,getRate,swap},keywordstyle=\color{Black}
  ,frame=single
  ,caption={A simple exchange contract handling two token types.}
  ,label={lst:exchange}]
interface IExchange is IPriceOracle {
  // list of swappable tokens
  function getTokens() external view returns (token[] memory);
  // receive x:ti & send >=ymin:to; revert if pair is not swappable
  function swap(token ti, uint x, token to, uint ymin) external returns (uint);
}

contract Exchange is IExchange {
  IPriceOracle p; token immutable t0; token immutable t1;
  
  constructor(IPriceOracle p_, token t0_, token t1_) { p=p_; t0=t0_; t1=t1_; }

  function getTokens() public view returns (token[] memory) {
    token[] memory ts = new token[](2); ts[0]=t0; ts[1]=t1; return ts;
  }

  function getRate(token ti, token to) public view returns (uint) { 
    return p.getRate(ti,to); 
  }
  
  function swap(token ti, uint x, token to, uint ymin) public returns (uint) {
    require ((ti==t0 && to==t1) || (ti==t1 && to==t0)); 
    ti.receive(x); // transfer x:ti from msg.sender to this
    ETH.receive(p.fee());  // transfer fee from sender to this
    ETH.send(p.fee(), p);  // approve fee transfer to oracle 
    uint y = (x * p.getRate(ti,to))/1_000_000;
    require (y>=ymin);
    to.send(msg.sender,y); // transfer y:to to msg.sender
    return y;
  }
}
  \end{lstlisting}
\end{figure}

\paragraph{Contract language}

In our examples, we specify contracts in a Solidity-like language with primitive token types.
Namely, rather than implementing tokens as ERC-20 contracts, our language extends Solidity with a base type \lstinline[language=solidity]{token} to represent token types.
Each token type \lstinline[language=solidity]{T} supports the following operations:
\begin{itemize}

\item \lstinline[language=solidity]{T.send(rcv,n)} has different behaviours depending on whether the recipient \code{rcv} is a user or a contract account.
If \code{rcv} is a user account, then the operation transfers \code{n:T} from the contract to \code{rcv}.
Instead, if \code{rcv} is a contract account, the operation authorizes the contract \code{rcv} to receive up to \code{n:T} from the caller's wallet.
We assume that a subsequent \lstinline[language=solidity]{send} with the same token and receiver overwrites the previous one, effectively cancelling it;

\item \lstinline[language=solidity]{T.receive(n)} transfers \code{n:T} from the immediate caller to the contract. 
The operation reverts if the caller has not previously approved the transfer, or if its wallet does not contain a sufficient balance of tokens \code{T}.

\end{itemize}

Except for the \lstinline[language=solidity]{token} type and its operations, our code snippets adhere to Solidity.%
\footnote{The ERC-20 interface uses different names for token operations, \ie\ \code{approve} to authorize token transfers, and \code{transferFrom} to receive previously authorized tokens.}
We will be a bit sloppier only in the technical toy examples, where for brevity we will omit, \eg, the function visibility and mutability modifiers.

\begin{example}
%
The contract $\contract{PriceOracle}$ in~\Cref{lst:price-oracle} has two methods: $\txcode{setPrice}$ allows the owner to set the price of a given token type (taking as unit of measure some reference token, \eg the \ETH), and $\txcode{getRate}$ allows anyone to read the exchange rate between two token types, provided that a fee is paid to the contract.  
The contract $\contract{Exchange}$ in \Cref{lst:exchange} allows anyone to pay tokens of type \code{t0} in exchange of tokens of type \code{t1} (and vice versa). 
The exchange rate is obtained through a price oracle, fixed when $\contract{Exchange}$ is deployed. 

Consider a transaction $\txT = \pmvA:\contract{Exchange}.\txcode{swap}(\tokT[1],10,\tokT[2],0)$ and the state:
\begin{align*}
\sysS 
\; = \;
& \walpmv{\pmvA}{10:\tokT[1],0:\tokT[2],1:\ETH}
\mid
\walpmv{\pmvB}{0:\ETH}
\\
\mid \;
& \walpmv{\contract{PriceOracle}}{\waltok{0}{\ETH}, \code{owner}=\pmvB, \code{fee}=1, \code{prices}=\setenum{\tokT[1]=3, \tokT[2]=5}, \ldots}
\\
\mid \; 
& \walpmv{\contract{Exchange}}{\waltok{0}{\tokT[1]}, \waltok{100}{\tokT[2]}, \code{p}=\contract{PriceOracle}, \code{t0}=\tokT[1], \code{t1}=\tokT[2]} 
\end{align*}
Since transactions are signed by their senders, we do not require them to explicitly approve token transfers (as in the Move language~\cite{Blackshear2019MoveAL}). 
Firing $\txT$ in $\sysS$ triggers an internal call
\(
\contract{Exchange}:\contract{PriceOracle}.\txcode{getRate}(\tokT[1],\tokT[2])
\),
overall causing a state transition from $\sysS$ to the state:
\begin{align*}
\sysSi 
\; = \;
& \walpmv{\pmvA}{0:\tokT[1],6:\tokT[2],0:\ETH}
\mid
\walpmv{\pmvB}{0:\ETH}
\\
\mid \;
& \walpmv{\contract{PriceOracle}}{\waltok{1}{\ETH}, \code{owner}=\pmvB, \code{fee}=1,  \code{prices}=\setenum{\tokT[1]=3, \tokT[2]=5}, \ldots}
\\
\mid \; 
& \walpmv{\contract{Exchange}}{\waltok{10}{\tokT[1]}, \waltok{94}{\tokT[2]}, \code{p}=\contract{PriceOracle}, \code{t0}=\tokT[1], \code{t1}=\tokT[2]} 
\end{align*}
Note that firing $\txT$ in $\sysSi$ would revert, since the sender $\pmvA$ does not have enough units of $\tokT[1]$ (and of $\ETH$) to transfer to the $\contract{Exchange}$.
\hfill\qedex
\end{example}

\paragraph{Sender-agnostic methods}

We say that a contract method is \emph{sender-agnostic} when 
the effect of an internal call to the method does not depend on the identity of the caller.
More precisely, consider two transactions $\txT[0]$ and $\txT[1]$ with the same initiator, whose executions trigger internal calls to the method $\txcode{f}$ of contract $\contract{C}$: 
\[
\contract[0]{D}:\contract{C}.\txcode{f}(\code{args})
\qquad
\contract[1]{D}:\contract{C}.\txcode{f}(\code{args})
\]
in the same state, with the same arguments, 
but with different senders $\contract[0]{D}$ and $\contract[1]{D}$, both of which have approved the token transfers needed  to fund the respective calls.
We then say that $\txcode{f}$ is sender-agnostic whenever the two internal calls either both revert, or both succeed producing the same effects, up-to the transfer of tokens to the sender.
More specifically, the two calls update the contract state in the same way, transfer the same tokens to the same user accounts, perform the same internal calls to other methods, and transfer the same tokens to their respective senders $\contract[0]{D}$ and $\contract[1]{D}$.

\begin{example}
\label{ex:PriceOracle-Exchange:sender-agnostic}
In the $\contract{PriceOracle}$ contract in~\Cref{lst:price-oracle},
the method $\txcode{setPrice}$ is \emph{not} sender-agnostic whenever \code{owner} is a contract account with a method that can call $\txcode{setPrice}$.
In that case, the effect of an internal call to $\txcode{setPrice}$ depends on the sender: if it is the owner, the call succeeds, otherwise it reverts.
Instead, the method $\txcode{getRate}$ is sender-agnostic, as its effect does not depend on the identity of the caller.
For the same reason, the method $\txcode{swap}$ of the $\contract{Exchange}$ contract in~\Cref{lst:exchange} is sender-agnostic as well.
\hfill\qedex
\end{example}



\paragraph{Token flows}
Let $\sysS$ be a blockchain state and let $\CmvC, \CmvD$ be sets of contract accounts.
We say that there is a \emph{token flow} from $\CmvC$ to $\CmvD$ in $\sysS$ if there exists a sequence of valid transactions in $\sysS$ containing as effects an output of some token type $\tokT$ from $\CmvC$ and an input of $\tokT$ from $\CmvD$.
As a special case, we say that there is a \emph{direct token flow} from $\CmvC$ to $\CmvD$ in $\sysS$ if there exists a sequence of valid transactions from $\sysS$ where the effects of one of the transactions includes the transfer of tokens from some contract in $\CmvC$ to some contract in~$\CmvD$.


\begin{figure}[t!]
  \begin{lstlisting}[language=solidity
  % ,classoffset=5,morekeywords={setPrice,getRate,swap},keywordstyle=\color{Black}
  ,frame=single
  ,caption={An airdrop contract and a fee-less exchange contract.}
  ,label={lst:airdrop}]
contract Airdrop {
  token t;
  constructor(token t_, uint n) { t=t_; t.receive(n); }
  function drop(uint n) public { t.send(msg.sender,n); }
}

contract FreeExchange is IExchange {
  IExchange ex;
  
  constructor(IExchange ex_) { ex=ex_; }

  function getTokens() public view returns (token[] memory) { return ex.getTokens(); }

  function getRate(token ti, token to) public view returns (uint) { 
    return ex.getRate(ti,to); 
  }
  
  function swap(token ti, uint x, token to, uint ymin) public returns (uint) {
    ti.receive(x);
    uint fee = PriceOracle.fee();
    Airdrop.drop(fee);
    ti.approve(ex,x);
    ETH.approve(ex,fee);
    ex.swap(ti,x,to,ymin); // revert if pair (ti,to) not swappable
    uint y = T2.balanceOf(this);
    require (y>=ymin);
    to.send(msg.sender,y);
    return y;
  }
}
\end{lstlisting}
\end{figure}

\begin{example}
\label{ex:PriceOracle-Exchange:token-flows}
Let $\CmvC = \setenum{\contract{Exchange}}$ and let 
$\CmvD = \setenum{\contract{PriceOracle}}$ for the contracts of~\Cref{lst:price-oracle,lst:exchange},
and let $\sysS$ as in~\Cref{ex:PriceOracle-Exchange:sender-agnostic}.
As shown before, firing $\pmvA:\contract{Exchange}.\txcode{swap}(\tokT[1],10,\tokT[2],0)$ in $\sysS$
triggers an internal call:
\[
\contract{Exchange}:\contract{PriceOracle}.\txcode{getRate}(\tokT[1],\tokT[2])
\]
whose effects include a transfer of $\ETH$ from the caller $\contract{Exchange}$ to the callee $\contract{PriceOracle}$.
Hence, there is a \emph{direct} token flow from $\CmvC$ to $\CmvD$ in $\sysS$.

To illustrate \emph{in}direct token flows, consider the contracts in~\Cref{lst:airdrop}.
The contract $\contract{Airdrop}$ allows anyone to withdraw any amount of a given token deposited in the contract. 
This feature is leveraged by the contract $\contract{FreeExchange}$ to enable fee-less token swaps. 
To do that, the $\contract{FreeExchange}$ first withdraws the needed fee from the $\contract{Airdrop}$, then invokes the $\contract{Exchange}$ from~\Cref{lst:exchange} to perform the token swap, and finally transfers back the tokens to the sender.

Let $\CmvC = \setenum{\contract{Airdrop}}$, 
let $\CmvD = \setenum{\contract{PriceOracle}}$, and let
\[
\sysS[1] 
= 
\sysS
\mid
\walpmv{\contract{Airdrop}}{\waltok{1000}{\ETH}, \code{t}=\ETH}
\mid
\walpmv{\contract{FreeExchange}}{\code{ex}=\contract{Exchange}}
\]
The transaction $\pmvA:\contract{FreeExchange}.\txcode{swap}(10)$ causes an $\ETH$ flow from $\CmvC$ to $\CmvD$ in~$\sysS[1]$.
This flow is indirect, since none of the internal calls includes a transfer of $\ETH$ from $\CmvC$ to $\CmvD$ as an effect. 
However, the effect of the internal call to $\contract{Airdrop}$ is an output of $\ETH$ from $\CmvC$ (to the $\contract{FreeExchange}$), and the effect of the internal call to $\contract{PriceOracle}$ includes an input of $\ETH$ (from $\contract{Exchange}$) to $\CmvD$.
\hfill\qedex
\end{example}







\section{Maximal Extractable Value}
\label{sec:mev}
\label{sec:lmev}

\renewcommand{\strip}[2]{{#1} \cap \deps{#2}}

In this~\namecref{sec:mev} we introduce a notion of MEV tailored to the study of smart contract compositions.
In~\Cref{sec:mev:def}, we formally define MEV and compare our notion with the seminal definition introduced in~\cite{Babel23clockwork}. 
In~\Cref{sec:mev:basic} we establish a collection of basic properties of MEV that are used throughout the paper.
In~\Cref{sec:mev:frontrunning-resistance,sec:mev:observable-invariance}, we prove two key results that provide sufficient conditions under which the adversary's capabilities can be restricted while preserving MEV. 


\subsection{Definition}
\label{sec:mev:def}

We measure the effect of an attack as the variation in the \emph{wealth} of a set of victim accounts before and after the attack.
Given a blockchain state $\sysS$ and set of accounts $\AddrA$ (possibly comprising both user and contract accounts), we define their wealth $\wealth{\AddrA}{\sysS}$ as the weighted sum of the tokens held in $\AddrA$'s wallets, where the weights are the token prices.
We denote by $\price{\tokT}$ the \emph{price} of a token type $\tokT$ and assume that $\price{\tokT} > 0$ for all $\tokT$.
We assume these prices to be fixed throughout an attack (note instead that the prices given by a $\contract{PriceOracle}$ contracts can vary, since they depend on the blockchain state).

\begin{definition}[Wealth]
\label{def:wealth}
The wealth of a set of accounts $\AddrA \subseteq \AddrU$ in $\sysS$ is given by:
\begin{equation*}
    \wealth{\AddrA}{\sysS}
    \; = \;
    \sum_{\tokT \in \TokU}
    \wal{\AddrA}{\sysS}(\tokT) \cdot \price{\tokT}
\end{equation*}
\end{definition}

Note that since both wallets and blockchain states have finite domain,
states cannot have an infinite amount of tokens,
and so the wealth is always finite.


\begin{definition}[Gain]
  \label{def:gain}
  The gain of $\AddrA \subseteq \AddrU$ upon firing $\TxTS$ in $\sysS$ is given by:
  \[
  \gain{\AddrA}{\sysS}{\TxTS}
  \; = \;
  \wealth{\AddrA}{\sysSi} - \wealth{\AddrA}{\sysS}
  \qquad
  \text{if $\sysS \xrightarrow{\TxTS} \sysSi$}
  \]
\end{definition}

MEV was originally defined in~\cite{Babel23clockwork} as the maximum gain that an adversary can achieve by firing a sequence of transactions in a blockchain state $\sysS$.
If we denote by $\mall{}{\Adv}$ the set of transactions craftable by $\Adv$ by using their private knowledge,
and by $\mall{}{\Adv}^*$ their finite sequences,
the notion of MEV in~\cite{Babel23clockwork} can be translitterated into our notation as:
\begin{equation}
  \label{eq:mev}
  \mev{}{\sysS}{}
  = \max \setcomp
  {\gain{\Adv}{\sysS}{\TxTS}}
  {\TxTS \in \mall{}{\Adv}^*}
\end{equation}

While~\eqref{eq:mev} measures an adversary's success in terms of its \emph{gain}, to analyse the economic security for smart contract compositions we focus instead on the \emph{loss} an adversary can inflict, on a given set of victim contracts~$\CmvC$, by manipulating transaction ordering and exploiting cross-contract dependencies.
To this purpose, our notion of MEV distinguishes from~\eqref{eq:mev} in the following aspects:
\begin{enumerate}

\item We measure the value that adversaries can \emph{drain} from a given set of victim contracts~$\CmvC$.
  This means that only the tokens extracted from the contracts in~$\CmvC$ contribute to the MEV. 
  Unlike in~\eqref{eq:mev}, the tokens grabbed from other contracts do not contribute to MEV.

\item We count as MEV \emph{all} the tokens that $\Adv$ can remove from $\CmvC$, regardless of whether $\Adv$ can transfer them to their wallets.
  Namely, while~\eqref{eq:mev} maximises $\Adv$'s gain $\gain{\Adv}{}{}$, our notion maximises $\CmvC$'s \emph{loss} $-\gain{\CmvC}{}{}$.
  Doing so, our notion of secure composability also encompasses \emph{irrational adversaries} who try to damage the victim contracts without necessarily making a profit.

\item We parameterise our MEV \wrt the set $\CmvD$ of contracts \emph{callable} by~$\Adv$.
  This enables us to reinterpret the separation between private and public information in language-based non-interference. 
  In that setting, the main insight is that observable (public) outputs remain unaffected by variations in confidential (private) inputs --- equivalently, restricting inputs to the private ones must preserve the public outputs.
  In our context, we rephrase this by requiring that the loss inflictable to $\CmvC$ (the public output) is not affected when restricting the set of callable contracts to a subset $\CmvD$ (the private inputs).

\end{enumerate}

\begin{definition}[MEV]
  \label{def:lmev}
  Let
  \(
  \mall{\CmvD}{\Adv}
  = 
  \setcomp{\txT \in \mall{}{\Adv}}{\callee{\txT} \in \CmvD}
  \)
  be the set of transactions craftable by $\Adv$ 
  restricted to calling contracts in $\CmvD$.
  We define:
  \begin{equation}
    \label{eq:lmev}
    \lmev{\CmvD}{\sysS}{\CmvC}
    = \max \setcomp
    {-\gain{\CmvC}{\sysS}{\TxTS}}
    {\TxTS \in \mall{\CmvD}{\Adv}^*}
  \end{equation}
  Hereafter, we refer to the sets $\CmvC$ and $\CmvD$ as the \emph{victim} and \emph{callable} contracts, respectively, 
  and we abbreviate $\lmev{\CmvU}{\sysS}{\CmvC}$ as
  $\lmev{}{\sysS}{\CmvC}$.
\end{definition}

\begin{example}
Consider the $\contract{Airdrop}$ contract from~\Cref{lst:airdrop}, and let:
\[
\sysS \; = \sysSi \mid \walpmv{\contract{Airdrop}}{\waltok{n}{\tokT}, \code{t}=\tokT}
\]
where $\sysSi$ is arbitrary.
Let $\CmvC = \setenum{\contract{Airdrop}}$, 
and let $\CmvD$ be any set of contract accounts including $\CmvC$.
Any user can fire a transaction in $\sysS$ to withdraw the entire airdrop balance. 
Therefore, $\lmev{\CmvD}{\sysS}{\CmvC} = n \cdot \price{\tokT}$.
\hfill\qedex
\end{example}

\Cref{lem:mev:mev-vs-babel} relates our notion of MEV with that in~\eqref{eq:mev}.
In particular, when the set of victim contracts is the universe $\CmvU$, our MEV over-approximates~\eqref{eq:mev}.
Note that the inequality $\leq$ does not hold in general, since some of the tokens drained from contracts in $\CmvU$ could be transferred to some account not in $\Adv$, and therefore they would contribute to the loss of $\CmvU$ but not to the gain of $\Adv$.

\begin{proposition}
    \label{lem:mev:mev-vs-babel}
    $\lmev{\CmvU}{\sysS}{\CmvU} \geq \mev{}{\sysS}{}$.
\end{proposition}


\subsection{Basic properties of MEV}
\label{sec:mev:basic}

\Cref{lem:mev:basic} establishes some basic properties of MEV:
\begin{itemize}

\item \Cref{lem:mev:basic:zero} studies some border cases when the MEV is zero.

\item \Cref{lem:mev:basic:leq-wealth} states that the MEV extractable from $\CmvC$ is non-negative
and bounded by the wealth of $\CmvC$.

\item \Cref{lem:mev:basic:garbage} states that the MEV of $\sysS$ is preserved when restricting the set of callable/victim contracts to the set $\cmvOfcst{\sysS}$ of contracts occurring in $\sysS$. 

\item \Cref{lem:mev:basic:L-leq-H,lem:mev:basic:monotonicity}
state that widening the callable contracts or the contract state may potentially increase MEV.
\Cref{ex:lmev:not-monotonic-on-observed-contracts} shows, instead, that widening the set of victim contracts does not necessarily increase the MEV: this is because maximising the loss of a contract may increase the gain of another one.

\end{itemize}

\begin{lemma}
  \label{lem:mev:basic}
  For all $\sysS$, $\CmvC,\CmvD,\CmvDi$ and $\cstD$:
  \begin{enumerate}

  \item \label{lem:mev:basic:zero}
    $\lmev{\CmvD}{\sysS}{\emptyset} = \lmev{\emptyset}{\sysS}{\CmvC} = 0$

  \item \label{lem:mev:basic:leq-wealth}
    $0 \leq \lmev{\CmvD}{\sysS}{\CmvC} \leq \wealth{\CmvC}{\sysS}$

  \item \label{lem:mev:basic:garbage}
    $\lmev{\CmvD}{\sysS}{\CmvC} = \lmev{\CmvD}{\sysS}{\CmvC \cap \cmvOfcst{\sysS}} = \lmev{\CmvD \cap \cmvOfcst{\sysS}}{\sysS}{\CmvC}$  

  \item \label{lem:mev:basic:L-leq-H}
    if $\CmvD \subseteq \CmvDi$, then
    $\lmev{\CmvD}{\sysS}{\CmvC} \leq \lmev{\CmvDi}{\sysS}{\CmvC}$
  
  \item \label{lem:mev:basic:monotonicity}
    $\lmev{\CmvD}{\sysS}{\CmvC} \leq \lmev{\CmvD}{\sysS \mid \cstD}{\CmvC}$

  \end{enumerate}
  Additionally, \cref{lem:mev:basic:zero,lem:mev:basic:L-leq-H,lem:mev:basic:garbage,lem:mev:basic:leq-wealth} also hold for non well-formed states $\sysS$.
\end{lemma}

\begin{example}
\label{ex:lmev:not-monotonic-on-observed-contracts}
We show that
$\CmvC \subseteq \CmvCi$ does not imply
$\lmev{\CmvD}{\sysS}{\CmvC} \leq \lmev{\CmvD}{\sysS}{\CmvCi}$.
Consider the following contracts, and assume that tokens have unit prices:
\begin{lstlisting}[language=solidity
  % ,morekeywords={f},classoffset=4,morekeywords={a,A,Oracle},keywordstyle=\pmvColor,classoffset=5,morekeywords={t,T0,T1,T2,ETH},keywordstyle=\tokColor,classoffset=6,morekeywords={C0,C1,C2},keywordstyle=\cmvColor
]
contract C0 { function f() { T0.send(msg.sender,5); } }
contract C1 { function f() { T0.receive(5); T1.send(msg.sender,1); } }
contract C2 { function f() { T1.receive(1); T2.send(msg.sender,100); } }
\end{lstlisting}

\noindent
Let $\Adv = \setenum{\pmvM}$, and let
$\TxTS =
\pmvM:\contract{C0}.\txcode{f}() \ \pmvM:\contract{C1}.\txcode{f}() \ \pmvM:\contract{C2}.\txcode{f}()$.
We have that:
\begin{align*}
\sysS = & \;
\walu{\pmvM}{0}{\tokT[2]} 
\mid
\walpmv{\contract{C0}}{\waltok{5}{\tokT[0]}} \mid
\walpmv{\contract{C1}}{\waltok{1}{\tokT[1]}} \mid
\walpmv{\contract{C2}}{\waltok{100}{\tokT[2]}}
\; \xrightarrow{\TxTS} \;
\walu{\pmvM}{100}{\tokT[2]} \mid
\walpmv{\contract{C0}}{\waltok{0}{\tokT[0]}} \mid
\walpmv{\contract{C1}}{\waltok{5}{\tokT[0]}} \mid
\walpmv{\contract{C2}}{\waltok{1}{\tokT[1]}}
\end{align*}
Let $\CmvC = \setenum{\contract{C2}} \subseteq \setenum{\contract{C1},\contract{C2}} = \CmvCi$.
Note that $\TxTS$ maximizes the loss of both $\CmvC$ and $\CmvCi$.
Indeed, $\contract{C2}$ is completely drained, except for the $1:\tokT[1]$ that is required to call its $\txcode{f}$.
For $\CmvCi$, the only alternative action would be to call $\contract{C1}.\txcode{f}$ without calling $\contract{C2}.\txcode{f}$, but this would produce a gain instead of a loss. 
Therefore:
\begin{align*}
    & \lmev{}{\sysS}{\setenum{\contract{C2}}}
    = -\gain{\setenum{\contract{C2}}}{\sysS}{\TxTS}
    = 100-1
    = 99
    \; \not\leq \;
    \lmev{}{\sysS}{\setenum{\contract{C1},\contract{C2}}}
    = -\gain{\setenum{\contract{C1},\contract{C2}}}{\sysS}{\TxTS}
    = 101-6 = 95
    \tag*{\qedex}
\end{align*}
\end{example}

\begin{example}
\label{cex:mev:basic:monotonicity}
\Cref{lem:mev:basic}\eqref{lem:mev:basic:monotonicity} ceases to hold if $\sysS \mid \cstD$ is not well-formed and contracts are allowed to test the existence of other contracts.
In such a setting, adding contracts to the system may strictly reduce  MEV, violating monotonicity.
\Eg, consider a contract $\contract{C}$ exposing a method $\txcode{f}$ with the following behaviour: $\txcode{f}$ performs a call to another contract $\contract{D}$; if the call fails (\eg because $\contract{D}$ does not exist), then $\txcode{f}$ transfers $\contract{C}$'s $\tokT$ balance to the caller; otherwise, it reverts. 
Let $\sysS$ be a state containing $\walpmv{\contract{C}}{\waltok{1}{\tokT}}$ but not $\contract{D}$, and let $\cstD$ be a state containing $\contract{D}$. 
If $\CmvC = \CmvD = \setenum{\contract{C}}$, then 
$\lmev{\CmvD}{\sysS}{\CmvC} = 1 \cdot \price{\tokT} > 0$, while
$\lmev{\CmvD}{\sysS \mid \cstD}{\CmvC} = 0$. 
\hfill\qedex
\end{example}

\Cref{lem:mev:state-narrowing} allows us to narrow the blockchain state while preserving MEV. 
Specifically, if $\CmvD$ is the set of callable contracts, we can remove from the state all the contracts that are neither in $\CmvD$ nor in its dependencies.
Recall that the states mentioned in our statements are implicitly assumed to be well-formed: in particular, in~\Cref{lem:mev:state-narrowing} the well-formedness of $\sysS \mid \cstD$ implies that the contracts in $\sysS$ cannot call those in $\cstD$.

\begin{lemma}[State narrowing]
\label{lem:mev:state-narrowing}
$\lmev{\CmvD}{\sysS \mid \cstD}{\CmvC} = \lmev{\CmvD}{\sysS}{\CmvC}$
if $\CmvD \subseteq \cmvOfcst{\sysS}$.
\end{lemma}


The following~\namecref{cex:lem:mev:state-narrowing} shows that if the condition $\CmvD \subseteq \cmvOfcst{\sysS}$ in \Cref{lem:mev:state-narrowing} does not hold, then MEV preservation is not guaranteed.

\begin{example}
\label{cex:lem:mev:state-narrowing}
Consider the following state containing instances of the $\contract{PriceOracle}$ and $\contract{Exchange}$ contracts in~\Cref{lst:exchange},
which overall allow anyone to swap $1:\tokT[1]$ for $2:\tokT[2]$:
\begin{align*}
\sysS 
\; = \;
\walu{\pmvM}{1}{\tokT[2]}
\mid \;
& \walpmv{\contract{PriceOracle1}}{\code{fee}=0, \code{prices}=\setenum{\tokT[1]=2, \tokT[2]=1}, \ldots}
\\
\mid \; 
& \walpmv{\contract{Exchange1}}{\waltok{0}{\tokT[1]}, \waltok{10}{\tokT[2]}, \code{p}=\contract{PriceOracle1}, \code{t0}=\tokT[1], \code{t1}=\tokT[2]} 
\end{align*}
Assume that the tokens are equally priced, \eg $\price{\tokT[1]} = \price{\tokT[2]} = 1$.
The adversary $\pmvM$ cannot inflict any loss to $\contract{Exchange1}$:
indeed, $\pmvM$ only has $1:\tokT[2]$, and so the only exchange she could attempt is to sell $\tokT[2]$ to buy $\tokT[1]$.
However, $\contract{Exchange1}$ has not units of $\tokT[1]$, so no interaction with $\pmvM$ is possible.

Assume now to extend $\sysS$ with the following $\cstD$:
\begin{align*}
\cstD 
\; = \;
& \walpmv{\contract{PriceOracle2}}{\code{fee}=0, \code{prices}=\setenum{\tokT[1]=1, \tokT[2]=1}, \ldots}
\\
\mid \; 
& \walpmv{\contract{Exchange2}}{\waltok{1}{\tokT[1]}, \waltok{0}{\tokT[2]}, \code{p}=\contract{PriceOracle2}, \code{t0}=\tokT[2], \code{t1}=\tokT[1]} 
\end{align*}
which allows anyone to swap tokens $\tokT[2]$ with an equal amount of tokens $\tokT[1]$.
In particular, $\pmvM$ can use $\contract{Exchange2}$ to obtain $1:\tokT[1]$ for $1:\tokT[2]$, and then exchange it for $2:\tokT[2]$ through $\contract{Exchange1}$.
Let $\CmvC = \setenum{\contract{Exchange1}}$ and let $\CmvD = \CmvU$.
Since $\CmvD \not\subseteq \cmvOfcst{\sysS}$, \Cref{lem:mev:state-narrowing} does not apply and so MEV preservation is not guaranteed.
Indeed:
\begin{align*}
& \lmev{\CmvD}{\sysS}{\setenum{\contract{Exchange1}}} = 0
\qquad \neq \qquad
\lmev{\CmvD}{\sysS \mid \cstD}{\setenum{\contract{Exchange1}}} = 2 \cdot \price{\tokT[2]} - 1 \cdot \price{\tokT[1]} = 1
\tag*{\qedex}
\end{align*}
\end{example}


\subsection{Front-running resistance}
\label{sec:mev:frontrunning-resistance}

Consider a state $\sysS$ where a user $\pmvA$ is deploying new contracts $\cstD$, with the expectation that the loss inflictable to $\cstD$ is bounded by $\lmev{}{\sysS \mid \cstD}{\cmvOfcst{\cstD}}$.
This expectation is generally false, since an adversary could front-run $\pmvA$'s deployment transaction with the deployment of new contracts $\cstC$, and in this way increase the loss inflictable to $\cstD$.
\Cref{th:mev:frontrunning-resistance} gives sufficient conditions under which MEV is preserved by such front-running attacks.
This result will be the basis to establish front-running resistance for MEV non-interference.

The proof of \Cref{th:mev:frontrunning-resistance} leverages a fundamental MEV preservation result, established by \Cref{th:mev:callable-narrowing:deps}:
\[
\lmev{\CmvD}{\sysS}{\CmvC} \; = \; \lmev{\CmvD \cap \deps{\CmvC}}{\sysS}{\CmvC}
\tag*{(under given conditions)}
\]

The key intuition for proving this result is to take a sequence of transactions $\TxTS$ that maximizes $\lmev{\CmvD}{\sysS}{\CmvC}$, and to translate it into a sequence $\TxYS$ where:
\begin{inlinelist}
\item we preserve the transactions with callee in $\deps{\CmvC}$, and
\item replace those with callee in $\CmvD \setminus \deps{\CmvC}$ with transactions where the adversary directly performs the first internal call in the \emph{boundary} between $\CmvC$ and $\CmvD$, \ie the contracts in $\deps{C}$ that are called from some contract in $\deps{\CmvD} \setminus \deps{\CmvC}$ (see~\Cref{def:boundary}).
\end{inlinelist}
In order for $\TxYS$ to be valid and have the same gain as $\TxTS$, a few conditions are sufficient:
\begin{itemize}

\item the contracts in the boundary are sender-agnostic, \ie they are not aware of the identity of the caller. This guarantees that replacing the internal call with the external call (with sender in $\Adv$) does not affect the semantics of the called functions;

\item the adversary has enough tokens to fund the boundary transactions. 
Technically, we ensure this by requiring that there are no token flows between the dependencies of $\CmvD$ and the boundary.

\end{itemize}

\newcommand{\boundary}[2]{\partial({#1},{#2})}

\begin{definition}[Boundary]
\label{def:boundary}
We define the \emph{boundary} between $\CmvC$ and $\CmvD$ as:
\[
\boundary{\CmvC}{\CmvD} 
\; = \; 
\setcomp{\contract{C} \in \deps{\CmvC}}{\exists \contract{D} \in \deps{\CmvD} \setminus \deps{\CmvC} \; : \; \contract{D} \text{ calls } \contract{C}}
\]
\end{definition}

\begin{lemma}[Callable contracts narrowing]
\label{th:mev:callable-narrowing:deps}
The equality:
\begin{equation*}
\lmev{\CmvD}{\sysS}{\CmvC} = \lmev{\CmvD \cap \deps{\CmvC}}{\sysS}{\CmvC}
\end{equation*}
holds if the following conditions are satisfied:
\begin{enumerate}

\item \label{th:mev:callable-narrowing:deps:1}
the contracts in $\boundary{\CmvC}{\CmvD}$ are sender-agnostic;
    
\item \label{th:mev:callable-narrowing:deps:2}
$\boundary{\CmvC}{\CmvD} \subseteq \CmvD$;

\item \label{th:mev:callable-narrowing:deps:3}
there is no token flow from $\CmvD \setminus \deps{\CmvC}$ to 
$\CmvD \cap \deps{\CmvC}$ in $\sysS$;

\item \label{th:mev:callable-narrowing:deps:4}
there is no direct token flow from $\deps{\CmvD} \setminus \deps{\CmvC}$ to 
$\boundary{\CmvC}{\CmvD}$ in $\sysS$.

\end{enumerate}
\end{lemma}

\Cref{ex:mev:callable-narrowing:deps} illustrates our proof 
by applying its construction to a toy example.
Technical \cref{cex:mev:callable-narrowing:sender-agnostic,cex:mev:callable-narrowing:boundary,cex:mev:callable-narrowing:no-token-flows} show that, when the above conditions are not met, then MEV preservation is not guaranteed to hold. 

\begin{example}
\label{ex:mev:callable-narrowing:deps}
To illustrate~\Cref{th:mev:callable-narrowing:deps}, consider the contracts: 

\vbox{
\begin{lstlisting}[
,language=solidity
%,basicstyle=\fontseries{m}\normalsize\ttfamily\lst@ifdisplaystyle\footnotesize\fi,
%,morekeywords={f,g,f0,f1,g0,g1},classoffset=4
%,morekeywords={a,A,M},keywordstyle=\pmvColor
%,classoffset=5,morekeywords={t,T,T0,T1,T2,ETH}
%,keywordstyle=\tokColor,classoffset=6
%,morekeywords={C0,C1,C2,C3},keywordstyle=\cmvColor
%,caption={Illustration of~\Cref{th:mev:callable-narrowing:deps}}
%,label={lst:cex:mev:callable-narrowing:deps}
]
contract C0 { function f() { T.receive(1); T.send(M,2); } }
contract C1 { function f() { T.send(C0,1); C0.f(); } }
contract C2 { function f() { require (msg.sender==C3); C1.f(); } }
contract C3 { function f() { C2.f(); C1.f(); } }
\end{lstlisting}}

\noindent
Let 
$\Adv = \setenum{\pmvM}$, 
$\CmvD= \setenum{\contract{C1},\contract{C3}}$,
$\CmvC = \setenum{\contract{C0,C1}}$, 
and let:
\[
\sysS = 
\walpmv{\pmvM}{0:\tokT} \mid
\walpmv{\contract{C0}}{\waltok{2}{\tokT}} \mid 
\walpmv{\contract{C1}}{\waltok{2}{\tokT}} \mid 
\walpmv{\contract{C2}}{\waltok{0}{\tokT}} \mid
\walpmv{\contract{C3}}{\waltok{0}{\tokT}}
\]
Let $\TxTS = \pmvM:\contract{C3}.\txcode{f()} \in \mall{\CmvD}{\pmvM}^*$.
By executing $\TxTS$ in $\sysS$, we have that:
\begin{align*}
    \sysS
    & \xrightarrow{\pmvM:\contract{C3}.\txcode{f()}}
    \walpmv{\pmvM}{4:\tokT} \mid
    \walpmv{\contract{C0}}{\waltok{0}{\tokT}} \mid 
    \walpmv{\contract{C1}}{\waltok{0}{\tokT}} \mid 
    \walpmv{\contract{C2}}{\waltok{0}{\tokT}} \mid
    \walpmv{\contract{C3}}{\waltok{0}{\tokT}}
\end{align*}
Since there are no tokens left in $\CmvC$, $\TxTS$ clearly maximises the loss of $\CmvC$, hence:
\[
\lmev{\CmvD}{\sysS}{\CmvC} = 4 \cdot \price{\tokT}
\]
We first check that the conditions of~\Cref{th:mev:callable-narrowing:deps}
are satisfied.
We have that:
\begin{align*}
\boundary{\CmvC}{\CmvD}
& = 
\setcomp{\contract{C} \in \deps{\CmvC}}{\exists \contract{D} \in \deps{\CmvD} \setminus \deps{\CmvC} \; : \; \contract{D} \text{ calls } \contract{C}}
\\
& =
\setcomp{\contract{C} \in \setenum{\contract{C0},\contract{C1}}}{\exists \contract{D} \in \setenum{\contract{C2},\contract{C3}} \; : \; \contract{D} \text{ calls } \contract{C}}
\\
& =
\setenum{\contract{C1}}
\end{align*}
The conditions of~\Cref{th:mev:callable-narrowing:deps} are then satisfied, since:
\begin{itemize}

\item[\eqref{th:mev:callable-narrowing:deps:1}]
the contract $\contract{C1} \in \boundary{\CmvC}{\CmvD}$ is sender-agnostic.
Note that $\contract{C2}$ is not sender-agnostic: however, this does not violate assumption~\ref{th:mev:callable-narrowing:deps:1}, since only the contracts in  $\boundary{\CmvC}{\CmvD}$ are required to be sender-agnostic.

\item[\eqref{th:mev:callable-narrowing:deps:2}]
$\boundary{\CmvC}{\CmvD} = \setenum{\contract{C1}} \subseteq \CmvD$.

\item[\eqref{th:mev:callable-narrowing:deps:3}]
There is no token flow from
$\CmvD \setminus \deps{\CmvC} = \setenum{\contract{C3}} $ to 
$\CmvD \cap \deps{\CmvC} = \setenum{\contract{C1}}$,
since $\contract{C3}$ has no token output.

\item[\eqref{th:mev:callable-narrowing:deps:4}]
There is no direct token flow from $\deps{\CmvD} \setminus \deps{\CmvC} = \setenum{\contract{C2,C3}}$ to 
$\boundary{\CmvC}{\CmvD} = \setenum{\contract{C1}}$ in $\sysS$, since
there is no token transfer towards $\contract{C1}$.

\end{itemize}

\noindent
We now construct a sequence of transactions 
$\TxYS \in \mall{\CmvD \cap \deps{\CmvC}}{\pmvM}^* = \mall{\setenum{\contract{C1}}}{\pmvM}^*$ that inflicts on $\CmvC$ the same loss as $\TxTS$,
following the proof of \Cref{th:mev:callable-narrowing:deps}.
Let $\vec{x}$ be the sequence of calls triggered by $\TxTS$.
We obtain a sequence $\TxYS$ from $\vec{x}$ as follows:
\begin{enumerate}

\item external calls 
$\pmvM:\contract{C}.\txcode{f}(\code{args})$ in $\vec{x}$ where $\contract{C} \in \deps{\CmvC}$ are preserved in $\TxYS$;

\item internal calls
$\contract{C}:\contract{D}.\txcode{f}(\code{args})$ in $\vec{x}$ where
$\contract{C} \in \boundary{\CmvC}{\CmvD}$ --- \ie, 
$\contract{C} \not\in \deps{\CmvC}$ and $\contract{D} \in \deps{\CmvC}$ ---
are added to $\TxYS$, replacing the sender $\contract{C}$ with the originator $\pmvM$;

\item all the other calls in $\vec{x}$ are discarded.

\end{enumerate}

The sequence $\vec{x}$ of (external/internal) calls induced by $\TxTS$
and the constructed sequence of transactions $\TxYS$ are the following:
\[
\begin{array}{lllllll}
\vec{x} \; =
\quad
& \pmvM:\contract{C3}.\txcode{f()}
\quad
& \contract{C3}:\contract{C2}.\txcode{f()}
\quad
& \contract{C2}:\contract{C1}.\txcode{f()}
\quad
& \contract{C1}:\contract{C0}.\txcode{f()}
\quad
& \contract{C3}:\contract{C1}.\txcode{f()}
\quad
& \contract{C1}:\contract{C0}.\txcode{f()}
\\
\TxYS  =
& 
& 
& \pmvM:\contract{C1}.\txcode{f()}
\quad
&
& \pmvM:\contract{C1}.\txcode{f()}
&
\end{array}
\]
Note that $\TxYS \in \mall{\CmvD \cap \deps{\CmvC}}{\pmvM}^* = \mall{\setenum{\contract{C1}}}{\pmvM}^*$. 
By executing $\TxYS$ in $\sysS$, we obtain:
\begin{align*}
    \sysS
    & \xrightarrow{\pmvM:\contract{C1}.\txcode{f()}}
    \walpmv{\pmvM}{2:\tokT} \mid
    \walpmv{\contract{C0}}{\waltok{1}{\tokT}} \mid 
    \walpmv{\contract{C1}}{\waltok{1}{\tokT}} \mid 
    \walpmv{\contract{C2}}{\waltok{0}{\tokT}} \mid
    \walpmv{\contract{C3}}{\waltok{0}{\tokT}}
    \\
    & \xrightarrow{\pmvM:\contract{C1}.\txcode{f()}}
    \walpmv{\pmvM}{4:\tokT} \mid
    \walpmv{\contract{C0}}{\waltok{0}{\tokT}} \mid 
    \walpmv{\contract{C1}}{\waltok{0}{\tokT}} \mid 
    \walpmv{\contract{C2}}{\waltok{0}{\tokT}} \mid
    \walpmv{\contract{C3}}{\waltok{0}{\tokT}}
\end{align*}
Hence, 
\(
\lmev{\CmvD \cap \deps{\CmvC}}{\sysS}{\CmvC} = 4 \cdot \price{\tokT}
\),
confirming the preservation of inflictable loss stated by~\Cref{th:mev:callable-narrowing:deps}.
\hfill\qedex
\end{example}

\begin{theorem}[MEV and front-running resistance]
\label{th:mev:frontrunning-resistance}
Let $\CmvC = \cmvOfcst{\cstD}$. Then:
\[
\lmev{\CmvD}{\sysS \mid \cstC \mid \cstD}{\CmvC} 
\; = \; 
\lmev{\CmvD \cap \deps{\CmvC}}{\sysS \mid \cstD}{\CmvC}
\]
holds if the following conditions are satisfied:
\begin{enumerate}

\item \label{condition:th:mev:frontrunning-resistance:sender-agnostic}
the contracts in $\deps{\cstD}$ are sender-agnostic

\item \label{condition:th:mev:frontrunning-resistance:boundary-in-callable}
$\deps{\cstD} \subseteq \CmvD$

\item \label{condition:th:mev:frontrunning-resistance:token-flows}
there are no token flows from 
$\deps{\CmvD} \setminus \deps{\cstD}$
to $\deps{\cstD}$
in $\sysS \mid \cstC \mid \cstD$.
\end{enumerate}
\end{theorem}

\begin{proof}
The conditions of~\Cref{th:mev:callable-narrowing:deps} are satisfied:
\begin{itemize}

\item[\eqref{th:mev:callable-narrowing:deps:1}] 
$\boundary{\CmvC}{\CmvD}$ are sender-agnostic, since $\boundary{\CmvC}{\CmvD} \subseteq \deps{\CmvC} = \deps{\cstD}$ and, by hypothesis~\ref{condition:th:mev:frontrunning-resistance:sender-agnostic}, the contracts in $\deps{\cstD}$ are sender-agnostic;

\item[\eqref{th:mev:callable-narrowing:deps:2}] 
$\boundary{\CmvC}{\CmvD} \subseteq \deps{\cstD} \subseteq \CmvD$ by hypothesis~\ref{condition:th:mev:frontrunning-resistance:boundary-in-callable};

\item[\eqref{th:mev:callable-narrowing:deps:3}]
No token flows from
$\CmvD \setminus \deps{\CmvC}$ to 
$\CmvD \cap \deps{\CmvC}$ in $\sysS \mid \cstC \mid \cstD$.
Indeed, we have
$\CmvD \setminus \deps{\CmvC} \subseteq \deps{\CmvD} \setminus \deps{\cstD}$
and 
$\CmvD \cap \deps{\CmvC} \subseteq \deps{\cstD}$,
hence we conclude by hypothesis~\ref{condition:th:mev:frontrunning-resistance:token-flows}.

\item[\eqref{th:mev:callable-narrowing:deps:4}]
No direct token flows from $\deps{\CmvD} \setminus \deps{\CmvC}$ to $\boundary{\CmvC}{\CmvD}$.
Indeed, we have $\deps{\CmvD} \setminus \deps{\CmvC} = 
\deps{\CmvD} \setminus \deps{\cstD}$ and $\boundary{\CmvC}{\CmvD} \subseteq \deps{\cstD}$, hence we conclude by hypothesis~\ref{condition:th:mev:frontrunning-resistance:token-flows}.

\end{itemize}
Therefore, by~\Cref{th:mev:callable-narrowing:deps} it follows that:
\[
\lmev{\CmvD}{\sysS \mid \cstC \mid \cstD}{\CmvC}
\; = \;
\lmev
{\CmvD \cap \deps{\CmvC}}
{\sysS \mid \cstC \mid \cstD}{\CmvC}
\]
Now, since $\sysS \mid \cstC$ and $\sysS \mid \cstD$ are well-formed, then $\deps{\cstC}$ is disjoint from $\cmvOfcst{\cstD}$ and $\deps{\cstD}$ is disjoint from $\cmvOfcst{\cstC}$, and so the states $\sysS \mid \cstC \mid \cstD$ and $\sysS \mid \cstD \mid \cstC$ are equivalent. Then:
\[
\lmev
{\CmvD \cap \deps{\CmvC}}
{\sysS \mid \cstC \mid \cstD}{\CmvC}
\; = \;
\lmev
{\CmvD \cap \deps{\CmvC}}
{\sysS \mid \cstD \mid \cstC}{\CmvC}
\]
Since $\sysS \mid \cstD$ is well-formed, then $\deps{\cstD} \subseteq \cmvOfcst{(\sysS \mid \cstD)}$.
Then, since $\CmvD \cap \deps{\CmvC} \subseteq \deps{\cstD} \subseteq \cmvOfcst{(\sysS \mid \cstD)}$, by~\Cref{lem:mev:state-narrowing} we have:
\[
\lmev
{\CmvD \cap \deps{\CmvC}}
{\sysS \mid \cstD \mid \cstC}{\CmvC}
\; = \;
\lmev
{\CmvD \cap \deps{\CmvC}}
{\sysS \mid \cstD}{\CmvC}
\]
Summing up, we obtain the thesis:
\[
\lmev{\CmvD}{\sysS \mid \cstC \mid \cstD}{\CmvC}
\; = \;
\lmev
{\CmvD \cap \deps{\CmvC}}
{\sysS \mid \cstD}{\CmvC}
\tag*{\qedhere}
\]
\end{proof}


\subsection{Observational invariance}
\label{sec:mev:observable-invariance}

We now provide another MEV preservation result, which will be the basis to establish sufficient conditions for MEV non-interference in~\Cref{sec:nonint}.

\Cref{th:mev:callable-narrowing:observational-invariance} states that the MEV inflictable to $\cstD$ in \mbox{$\sysS \mid \cstD$} is preserved by narrowing the callable contracts to $\cstD$ when $\cstD$ is \emph{observationally invariant} \wrt adversarial moves in~$\sysS$. 
Roughly, this means that any adversarial transaction calling a contract in $\sysS$ cannot affect what $\cstD$ observes by invoking the functions in $\sysS$.
The observables include the returned values, the transferred tokens and the aborts of functions that can be called from $\cstD$. 
In other words, even if the adversary interacts with $\sysS$, possibly changing its internal state, the contracts in $\cstD$ will have the same behaviour as if the adversary never interacted with $\sysS$.


\begin{definition}[Observational invariance]
\label{def:observational-invariance}
Let \mbox{$\sysS \mid \cstD$} be a blockchain state (not necessarily well-formed). 
We say that $\cstD$ is \emph{observationally invariant in $\sysS$} 
whenever the following condition holds.
Let $\sysSi$ be any state reachable through moves of $\Adv$ from $\sysS \mid \cstD$.
Then, we require that:
\begin{enumerate}

\item if $\txT \in \mall{}{\Adv}$ is such that $\callee{\txT} \in \cmvOfcst{\sysS}$, then executing $\txT$ in $\sysSi$ does not change the state of the contracts in $\cstD$ (neither their wallet nor storage);

\item if $\txT, \txY \in \mall{}{\Adv}$ are such that
$\callee{\txT} \in \cmvOfcst{\sysS}$,
$\callee{\txY} \in \cmvOfcst{\cstD}$, 
and the sequence $\txT\txY$ is valid in $\sysSi$,
then $\txY$ satisfies the following conditions:
\begin{itemize}
\item $\txY$ is valid in $\sysSi$ provided that the sender has enough tokens;
\item executing $\txY$ in $\sysSi$ triggers the same sequence of internal calls from the contracts in $\cmvOfcst{\cstD}$ to those in $\cmvOfcst{\sysS}$, with the same arguments and sent tokens, as the sequence of calls that $\txT\txY$ triggers in $\sysSi$; 
\item each call triggered by $\txY$ has the same effects (of the kinds: aborts, token transfers, method calls to $\cmvOfcst{\cstD}$, returns) of the corresponding call of $\txT\txY$.
\end{itemize}
\end{enumerate}
\end{definition}

\begin{example}
  \label{ex:observational-invariance:exchange-bet}
  Let $\cstD = \walpmv{\contract{Exchange}}{\cdots}$, where the exchange uses an instance of the $\contract{PriceOracle}$ contract occurring in $\sysS$.
  Let $\Adv$ be a set of accounts not including the oracle's owner.
  Then, $\cstD$ is observationally invariant \wrt $\Adv$'s moves in $\sysS$.
  To show that, let $\sysSi$ be any state reachable from $\sysS \mid \cstD$ through adversarial moves, and let $\txT$ and $\txY$ satisfy the hypotheses of~\Cref{def:observational-invariance}. 
  The transaction $\txT$ must call $\contract{PriceOracle}$; since $\txT$ is valid, it must call the method $\txcode{getRate}$, since $\txcode{setPrice}$ requires that the caller is the owner. 
  Since $\txY$ is valid after $\txT$ is executed and $\txT$ does not modify the $\contract{PriceOracle}$'s storage, then $\txY$ is also valid in $\sysSi$, and it produces exactly the same effects when executed in both states.
  \hfill\qedex
\end{example}

\begin{example}
  \label{ex:observational-invariance:bank}
  Let $\sysS$ consist of the $\contract{Bank}$ contract in~\Cref{lst:bank}. 
  We show that any $\cstD$ is observationally invariant \wrt adversarial moves in $\sysS$.
  Let $\sysSi$ be any state reachable from $\sysS \mid \cstD$ through adversarial moves, and let $\txT$, $\txY$ satisfy the hypotheses of~\Cref{def:observational-invariance}. 
  Then, $\txT$ must call one of the methods in $\contract{Bank}$.
  The effects of both methods $\txcode{deposit}$ and $\txcode{withdraw}$ is to update the \code{credits} entry of the adversary and its token balance.
  By assumption, $\txY$ calls a method of a contract $\contract{D}$ in $\cstD$, and of course $\contract{D} \not\in \Adv$.
  The calls from $\contract{D}$ to $\contract{Bank}$ are not affected by $\txT$, since they depend on the \code{credits} of $\contract{D}$.
  Then, since $\txY$ is valid after firing $\txT$, it is also valid in $\sysSi$ --- provided that the $\txY$'s sender has enough tokens to fund the call. 
  Finally, the effects of $\txY$ are not affected by the execution of $\txT$.
  %
  \hfill\qedex
\end{example}

\begin{figure}[t!]
\begin{lstlisting}[
,language=solidity
,caption={A bank contract.}
,label={lst:bank}
]
contract Bank {
  mapping (address user => uint credit) credits;

  // receive tokens T from the sender and store them in the contract
  // (inreasing the sender's credits of the corresponding amount)
  function deposit(uint amount) public {
    T.receive(amount);
    credits[msg.sender] += amount;
  }

  // transfers tokens T from the contract to the sender
  // (decreasing the sender's credits of the corresponding amount)
  function withdraw(uint amount) public {
    require(amount <= credits[msg.sender]);
    credits[msg.sender] -= amount;
    T.send(msg.sender,amount);
  }
}
\end{lstlisting}
\end{figure}




\begin{theorem}[MEV and observational invariance]
\label{th:mev:callable-narrowing:observational-invariance}
For any state $\sysS \mid \cstD$ (not necessarily well-formed),  
\[
    \lmev{\CmvD}{\sysS\mid \cstD}{\cmvOfcst{\cstD}}
    \; = \; 
    \lmev{\CmvD \cap \cmvOfcst{\cstD}}{\sysS\mid \cstD}{\cmvOfcst{\cstD}} 
\]
holds whenever both the following conditions are satisfied:
\begin{enumerate}

\item there are no token flows from $\sysS$ to $\cstD$ in \mbox{$\sysS \mid \cstD$};

\item $\cstD$ is observationally invariant \wrt adversarial moves in $\sysS$.

\end{enumerate}
\end{theorem}
\begin{proof}
The inequality $\geq$ is given by \Cref{lem:mev:basic}\eqref{lem:mev:basic:L-leq-H},
so we prove~$\leq$.
Let $\TxTS \in \mall{\CmvD}{\Adv}^*$ be a sequence of transactions that maximizes the loss of $\cmvOfcst{\cstD}$ when executed $\sysS\mid \cstD$.
W.l.o.g., assume that such transactions are valid. 
We must show that there is a sequence $\TxYS \in \mall{\CmvD \cap \cmvOfcst{\cstD}}{\Adv}^*$ that causes a loss of $\cmvOfcst{\cstD}$ greater than or equal to the one caused by $\TxTS$.
We choose $\TxYS$ to be the subsequence of $\TxTS$ comprising only direct calls to contracts in $\cstD$. Clearly, $\TxYS$ is in $\mall{\CmvD \cap \cmvOfcst{\cstD}}{\Adv}^*$.

We now show that $\TxYS$ is valid in $\sysS \mid \cstD$.
First, we note that any transaction of $\TxTS$ calling contracts in $\sysS$ (\ie the ones that have been removed to form $\TxYS$) does not change the behaviour of $\sysS$ as observed from $\cstD$.
This is due to the observational invariance of $\cstD$ (\Cref{def:observational-invariance}): the removed transactions can affect the state of $\sysS$ but this does not alter the effects of the contract calls from $\cstD$ to $\sysS$. Consequently, $\cstD$ can not distinguish between the state of $\sysS$ reached by $\TxTS$ and the one reached by $\TxYS$, hence 
$\cstD$ will exhibit the same behavior in both cases.

For this reason, $\TxTS$ and $\TxYS$ both cause the same state updates in $\cstD$, and in particular they perform the same token transfers to and from it.
To conclude that $\TxYS$ is valid, we note that the adversary must have enough tokens to fund all the calls in $\TxYS$. Indeed, the adversary has enough funds to execute $\TxTS$, and any token that they gain from the discarded transactions cannot be used by calls in $\TxYS$, since there is no token flow from $\sysS$ to $\cstD$.

Since $\TxYS$ is valid and updates the state of $\cstD$ in the same way as $\TxTS$, the loss caused to $\cstD$ must also be the same, so $\lmev{\CmvD}{\sysS \mid \cstD}{\cmvOfcst{\cstD}} \leq \lmev{\CmvD \cap \cmvOfcst{\cstD}}{\sysS\mid \cstD}{\cmvOfcst{\cstD}}$, which concludes the proof. 
\end{proof}

\section{MEV for wealthy adversaries}
\label{sec:rich-mev}

We now study MEV under a stronger adversarial model, in which the adversary is assumed to always possess enough tokens to inflict the maximal possible loss on the victim contracts. 
This analysis relies a basic property of the underlying blockchain model, which we call \emph{wallet monotonicity}: 
any transaction that is valid in $\sysS$ must produce the same effect if we enrich the users' wallets in $\sysS$ with additional tokens.

We formalise wallet monotonicity in \Cref{lem:wallet-monotonicity}, where $\WmvA + \WmvAi$ is the pointwise summation of wallet states 
(\eg, if $\WmvA = \setenum{\pmvA \mapsto \setenum{\tokT \mapsto 1}}$ and $\WmvAi = \setenum{\pmvA \mapsto \setenum{\tokT \mapsto 2, \tokTi \mapsto 5}}$,
then $\WmvA + \WmvAi = \setenum{\pmvA \mapsto \setenum{\tokT \mapsto 3, \tokTi \mapsto 5}}$).

\begin{lemma}[Wallet monotonicity]
  \label{lem:wallet-monotonicity}
  Let $\txT$ be valid in $\WmvA \mid \cstC$.
  For all $\WmvAii$:
  \[
  \WmvA \mid \cstC \xrightarrow{\txT} \WmvAi \mid \cstCi
  \quad\implies\quad
  (\WmvA + \WmvAii) \mid \cstC \xrightarrow{\txT} (\WmvAi + \WmvAii) \mid \cstCi
  \]
\end{lemma}
\begin{proof}
Recall that our contract model (\Cref{sec:blockchain}) does not allow contracts to read users' wallets.
Consequently, since $\txT$ is valid in $\WmvA \mid \cstC$, then $\txT$ is also valid in $(\WmvA + \WmvAii) \mid \cstC$ and it produces the same effect when executed in such state.
\end{proof}


We introduce in~\Cref{def:richer-state} a relation $\sysS \leq_{\$} \sysSi$ between blockchain states, which holds when $\sysSi$ can be obtained from $\sysS$ by enriching the users' wallets.

\begin{definition}
  \label{def:richer-state}
  We write $\sysS \leq_{\$} \sysSi$ when $\sysS = \WmvA \mid \cstC$
  and $\sysSi = (\WmvA + \WmvAi) \mid \cstC$,
  for some $\WmvA$, $\WmvAi$ and $\cstC$.
\end{definition}

\Cref{lem:mev:wallet} states that the only user wallets that
need to be taken into account to estimate the  MEV are those of the adversary
(\Cref{lem:mev:wallet:zero-user}).
This follows because:
\begin{inlinelist}
\item $\Adv$ has no way to force other users to spend their
tokens in the attack sequence, and
\item contract functions cannot read users' wallets.
\end{inlinelist}
Furthermore, wealthier adversaries may potentially inflict more loss 
(\Cref{lem:mev:wallet:monotonicity}).

\begin{lemma}[MEV and adversaries' wallets]
  \label{lem:mev:wallet}
  We have that:
  \begin{enumerate}

  \item \label{lem:mev:wallet:zero-user}
    if $\Adv \cap \dom{\WmvA} = \emptyset$, then
    \(
    \lmev{\CmvD}{\sysS \mid \WmvA}{\CmvC} = \lmev{\CmvD}{\sysS}{\CmvC}
    \)
    
  \item \label{lem:mev:wallet:monotonicity}
    if $\sysS \leq_{\$} \sysSi$, then
    $\lmev{\CmvD}{\sysS}{\CmvC} \leq \lmev{\CmvD}{\sysSi}{\CmvC}$
  \end{enumerate}
\end{lemma}

\Cref{lem:mev:stability} states that, for any contract state $\cstC$, there exists a wallet state $\WmvA$ yielding the same MEV as any of its enrichments $\WmvAi \geq_{\$} \WmvA$.
Together with~\Cref{lem:mev:wallet}, this ensures that an adversary $\Adv$ with a sufficiently rich wallet $\WmvA$ can always inflict the maximum loss to the contracts in $\cstC$.

\begin{lemma}[MEV stability]
  \label{lem:mev:stability}
  For all $\CmvC$, $\CmvD$, and $\cstC$: 
  \[
  \exists \WmvA . \;\;
  \forall \WmvAi \geq_{\$} \WmvA . \;\;
  \lmev{\CmvD}{\WmvA \mid \cstC}{\CmvC}  
  =
  \lmev{\CmvD}{\WmvAi \mid \cstC}{\CmvC}    
  \]
\end{lemma}

Note that MEV stability implies the existence of the maximum MEV over all possible wallet states, \ie $\max_{\WmvAii} \; \lmev{\CmvD}{\WmvAii \mid \cstC}{\CmvC}$. 
Such maximum is obtained taking $\WmvAii = \WmvA$ since
for any other $\WmvAii$ we have:
\begin{align*}  
\lmev{\CmvD}{\WmvAii \mid \cstC}{\CmvC}
& \leq
\lmev{\CmvD}{(\WmvA + \WmvAii) \mid \cstC}{\CmvC}
&& \text{by~\Cref{lem:mev:wallet}\eqref{lem:mev:wallet:monotonicity}}
\\
& =
\lmev{\CmvD}{\WmvA \mid \cstC}{\CmvC}
&& \text{by~\Cref{lem:mev:stability}}
\end{align*}

By taking such maximum MEV we can then define a notion of MEV 
that does \emph{not} depend on users' wallets (including $\Adv$'s).
The new notion captures, for example, the attacks where the adversary can obtain the needed tokens through a \emph{flash loan}, \ie a loan of an arbitrary amount of tokens that must be entirely repaid within the same transaction~\cite{Qin21fc,Babel23clockwork}.

\begin{definition}[MEV of wealthy adversaries]
  \label{def:rich-mev}
  For all $\CmvC,\CmvD,\cstC$, we define:
  \begin{equation*}
    \rlmev{\CmvD}{\cstC}{\CmvC}
    = \max_{\WmvA} \;
    \lmev{\CmvD}{\WmvA \mid \cstC}{\CmvC}
  \end{equation*}
\end{definition}


The following~\namecref{lem:rich-mev:basic} is the wealthy-adversaries counterpart of~\Cref{lem:mev:basic}.

\begin{lemma}[Basic properties of $\rlmev{}{}{}$]
  \label{lem:rich-mev:basic}
  For all $\cstC$, $\cstD$, and $\CmvC,\CmvD \subseteq \CmvU$:
  \begin{enumerate}

  \item \label{lem:rich-mev:basic:zero}
    $\rlmev{\CmvD}{\cstC}{\emptyset} = \rlmev{\emptyset}{\cstC}{\CmvC} = 0$

  \item \label{lem:rich-mev:basic:leq-wealth}
    $0 \leq \rlmev{\CmvD}{\cstC}{\CmvC} \leq \wealth{\CmvC}{\cstC}$

  \item \label{lem:rich-mev:basic:garbage}
    $\rlmev{\CmvD}{\cstC}{\CmvC} = \rlmev{\CmvD}{\cstC}{\CmvC \cap \cmvOfcst{\cstC}} = \rlmev{\CmvD\cap \cmvOfcst{\cstC}}{\cstC}{\CmvC}$

  \item \label{lem:rich-mev:basic:L-leq-H}
    if $\CmvD \subseteq \CmvDi$, then
    $\rlmev{\CmvD}{\cstC}{\CmvC} \leq \rlmev{\CmvDi}{\cstC}{\CmvC}$
	
  \item \label{lem:rich-mev:basic:monotonicity}
    $\rlmev{\CmvD}{\cstC}{\CmvC} \leq \rlmev{\CmvD}{\cstC \mid \cstD}{\CmvC}$
        
  \end{enumerate}
  Additionally, \cref{lem:rich-mev:basic:zero,lem:rich-mev:basic:L-leq-H,lem:rich-mev:basic:garbage,lem:rich-mev:basic:leq-wealth} also hold for non well-formed states $\cstC$.
\end{lemma}

The following~\namecref{lem:rich-mev:state-narrowing} is the wealthy-adversaries counterpart of~\Cref{lem:mev:state-narrowing}:
extending the state with new contracts does not affect the MEV
extractable from the old contracts.%

\begin{lemma}[State narrowing]
\label{lem:rich-mev:state-narrowing}
$\rlmev{\CmvD}{\cstS\!\mid\!\cstD}{\CmvC} = \rlmev{\CmvD}{\cstS}{\CmvC}$ if \mbox{$\CmvD \subseteq \cmvOfcst{\cstS}$}.
\end{lemma}
\begin{proof}
By~\Cref{def:rich-mev} and~\Cref{lem:mev:stability} (twice), there exist $\WmvA[0]$ and $\WmvA[1]$ such that:
\begin{align*}
  & \forall \WmvAi \geq_{\$} \WmvA[0] . \;
  \lmev{\CmvD}{\WmvA[0] \mid \cstS \mid \cstD}{\CmvC}  
  =
  \lmev{\CmvD}{\WmvAi \mid \cstS \mid \cstD}{\CmvC}
  = \rlmev{\CmvD}{\cstS \mid \cstD}{\CmvC}
  \\
  & \forall \WmvAi \geq_{\$} \WmvA[1] . \;
  \lmev{\CmvD}{\WmvA[1] \mid \cstS}{\CmvC}  
  =
  \lmev{\CmvD}{\WmvAi \mid \cstS}{\CmvC}    
  = \rlmev{\CmvD}{\cstS}{\CmvC}
\end{align*}
Let $\WmvA = \max \setenum{\WmvA[0],\WmvA[1]}$.
By~\Cref{lem:mev:state-narrowing}, since $\CmvD \subseteq \cmvOfcst{\cstS}$:
\[
  \lmev{\CmvD}{\WmvA \mid \cstS \mid \cstD}{\CmvC}    
  \; = \;
  \lmev{\CmvD}{\WmvA \mid \cstS}{\CmvC}    
\]
Since $\WmvA \geq_{\$} \WmvA[0]$ and $\WmvA \geq_{\$} \WmvA[1]$,
by the above equations we have the thesis.
\end{proof}




\Cref{th:rich-mev:callable-narrowing:deps} is the wealthy-adversaries counterpart of~\Cref{th:mev:callable-narrowing:deps}, as it 
gives sufficient conditions under which
we can remove from the contract state $\cstS$ all the non-dependencies of $\CmvC$,
preserving $\rlmev{\CmvD}{\cstS}{\CmvC}$.
Compared to~\Cref{th:mev:callable-narrowing:deps}, here the conditions do not include the absence of token flows, since wealthy adversaries always have enough tokens to perform the optimal attack.

\begin{lemma}[Callable contracts narrowing]
\label{th:rich-mev:callable-narrowing:deps}
The equality:
\begin{equation*}
\rlmev{\CmvD}{\cstS}{\CmvC} \; = \; \rlmev{\CmvD \cap \deps{\CmvC}}{\cstS}{\CmvC}
\end{equation*}
holds if the contracts in $\boundary{\CmvC}{\CmvD}$ are sender-agnostic and
$\boundary{\CmvC}{\CmvD} \subseteq \CmvD$.
\end{lemma}
\begin{proof}
By~\Cref{lem:mev:stability}, it suffices to prove that,
for a rich enough $\WmvA$:
\[
\lmev{\CmvD}{\WmvA \mid \cstS}{\CmvC} 
\; = \; 
\lmev{\CmvD \cap \deps{\CmvC}}{\WmvA \mid \cstS}{\CmvC}
\]
This equation has the same form as the thesis of~\Cref{th:mev:callable-narrowing:deps}, whose proof we can replay almost verbatim. 
The only points of such proof which are changed are those where we exploit its ``no tokens flows'' assumptions to justify that the adversaries have enough tokens to perform their attack. Here, we have no similar assumption to exploit, but on the other hand we only have to prove the statement for \emph{rich enough} $\WmvA$, so we can assume that $\WmvA$ has enough tokens for the adversaries.
\end{proof}

The following~\namecref{cex:rich-mev:callable-narrowing:deps} shows that if the  assumptions of~\Cref{th:rich-mev:callable-narrowing:deps} do not hold, then MEV preservation is not guaranteed. 

\begin{figure}[t]
\begin{lstlisting}[language=solidity
  ,classoffset=4,morekeywords={f,g}
  ,classoffset=5,morekeywords={t,T,T0,T1,T2,ETH},keywordstyle=\tokColor
  ,classoffset=6,morekeywords={C0,C1},keywordstyle=\cmvColor
  ,caption={Contracts for~\Cref{cex:rich-mev:callable-narrowing:deps}.}
  ,label={lst:rich-mev:callable-narrowing:deps}]
contract C0 {
  bool b;
  constructor(bool b_) { b=b_; }
  function f() { require (b||msg.sender==C1); T.send(msg.sender,5);} 
}
contract C1 { 
  function g() { C0.f(); T.send(msg.sender,5); } 
}
\end{lstlisting}
\end{figure}

\begin{example}
\label{cex:rich-mev:callable-narrowing:deps}
Consider the contracts in~\Cref{lst:rich-mev:callable-narrowing:deps}.
Let $\Adv = \setenum{\pmvM}$, let
$\CmvC = \setenum{\contract{C0}}$, 
$\CmvD= \setenum{\contract{C1}, \contract{C0}}$, 
and let \mbox{$\cstS = \walpmv{\contract{C0}}{\waltok{5}{\tokT},\code{b}=\code{false}} \mid \walpmv{\contract{C1}}{\waltok{0}{\tokT}}$}. 
We have that $\boundary{\CmvC}{\CmvD} = \setenum{\contract{C0}} \subseteq \CmvD$, 
but $\contract{C0}$ violates sender-agnosticism
since it enables token transfers only to $\contract{C1}$.  
Therefore, \Cref{th:rich-mev:callable-narrowing:deps} does not apply, and so MEV preservation is not guaranteed. 
Indeed, we have that:
\begin{align*}
& \lmev{\CmvD}{\cstS}{\CmvC} = 5 \cdot \price{\tokT} 
\\ 
& \lmev{\CmvD \cap \deps{\CmvC}}{\cstS}{\CmvC}
= \lmev{\setenum{\contract{C0}}}{\cstS}{\CmvC}
= 0
\end{align*}
The first equation is witnessed by the transaction 
\mbox{$\pmvM:\contract{C1}.\txcode{g}() \in \mall{\CmvD}{\Adv}$}, which extracts $5:\tokT$ from $\contract{C0}$.
For the second equation, the MEV is zero since the only possible transaction $\pmvM:\contract{C0}.\txcode{f}()$ reverts because of the \lstinline[language=solidity]{require} condition.

We now show that if the assumption $\boundary{\CmvC}{\CmvD} \subseteq \CmvD$ does not hold, then MEV preservation is not guaranteed either. 
Consider again the contracts in~\Cref{lst:rich-mev:callable-narrowing:deps}, and
let $\Adv = \setenum{\pmvM}$, $\CmvC = \setenum{\contract{C0}}$, $\CmvD= \setenum{\contract{C1}}$, and \mbox{$\cstS = \walpmv{\contract{C0}}{\waltok{5}{\tokT},\code{b}=\code{true}} \mid \walpmv{\contract{C1}}{\waltok{0}{\tokT}}$}.
This instance of $\contract{C0}$ is sender-agnostic, but $\CmvC \not\subseteq \CmvD$.
As before, we have that 
$\lmev{\CmvD}{\cstS}{\CmvC} = 5 \cdot \price{\tokT}$, 
but 
$\lmev{\CmvD \cap \deps{\CmvC}}{\cstS}{\CmvC}
= \lmev{\emptyset}{\cstS}{\CmvC}
= 0$.
\hfill\qedex
\end{example}

\Cref{th:rich-mev:callable-narrowing:observational-invariance} is the wealthy-adversaries counterpart of~\Cref{th:mev:callable-narrowing:observational-invariance}.
Compared to that, the assumption about the absence of indirect token flows is not required.

\begin{theorem}[$\rlmev{}{}{}$ and observational invariance]
\label{th:rich-mev:callable-narrowing:observational-invariance}
For any state $\cstS \mid \cstD$ (not necessarily well-formed),  
\[
    \rlmev{\CmvD}{\cstS \mid \cstD}{\cmvOfcst{\cstD}}
    \; = \; 
    \rlmev{\CmvD \cap \cmvOfcst{\cstD}}{\cstS \mid \cstD}{\cmvOfcst{\cstD}} 
\]
holds iff $\cstD$ is observationally invariant \wrt adversarial moves in $\WmvA \mid \cstS$, for all sufficiently large $\WmvA$.
\end{theorem}
\begin{proof}
    Assume that observational invariance holds for all $\WmvA \geq_{\$} \WmvA[oi]$.
    We show:
	\begin{equation}
		\label{eq:proof:th:richnonint:sufficient-conditions}
		\rlmev{\CmvD \cap \cmvOfcst{\cstD}}{\cstS \mid \cstD }{\cmvOfcst{\cstD}}
        \geq
        \rlmev{\CmvD}{\cstS \mid \cstD }{\cmvOfcst{\cstD}}
	\end{equation}
	since the other inequality is guaranteed by \Cref{lem:rich-mev:basic}\eqref{lem:rich-mev:basic:L-leq-H}.
	To do so, we consider a $\WmvA[max]$ that maximizes $\lmev{\CmvD}{\WmvA[max] \mid \cstS \mid \cstD }{\cmvOfcst{\cstD}}$.
    \Wlog, by MEV stability (\Cref{lem:mev:stability}) we can assume $\WmvA[max] \geq \WmvA[oi]$.
    
    Then, to prove~\eqref{eq:proof:th:richnonint:sufficient-conditions} it suffices to find a wallet $\WmvAi$ satisfying the following:
    \begin{align}
        \rlmev{\CmvD \cap \cmvOfcst{\cstD}}{\cstS \mid \cstD }{\cmvOfcst{\cstD}} 
        & \geq
        \lmev{\CmvD \cap \cmvOfcst{\cstD}}{\WmvAi \mid \cstS \mid \cstD }{\cmvOfcst{\cstD}} 
        \label{eq:proof:th:richnonint:max}
        \\
        & \geq
        \lmev{\CmvD}{\WmvA[max] \mid \cstS \mid \cstD}{\cmvOfcst{\cstD}} 
        \label{eq:proof:th:richnonint:defWi}
        \\
        & =
        \rlmev{\CmvD}{\cstS \mid \cstD}{\cmvOfcst{\cstD}}
        \label{eq:proof:th:richnonint:defMax}
    \end{align}
    where inequality \eqref{eq:proof:th:richnonint:max} follows from the definition of $\rlmev{}{}{}$, which takes the maximum over all possible wallets $\WmvAi$,
    and equality \eqref{eq:proof:th:richnonint:defMax} follows from our choice of $\WmvA[max]$ as the wallet state that gives $\rlmev{}{}{}$.

    For inequality \eqref{eq:proof:th:richnonint:defWi}, we proceed similarly to the proof of~\Cref{th:mev:callable-narrowing:observational-invariance}, where we show that for any MEV-extracting sequence of transactions in $\WmvA[max] \mid \cstS \mid \cstD$ there exists another sequence of transactions that extracts the same value but targeting only contracts in $\cstD$.
    The same steps of our previous proof, with one exception, can be replayed since observational invariance holds by assumption and $\WmvA[max] \geq \WmvA[oi]$.
    The exception is due to those steps that exploited the absence of token flows (an assumption of~\Cref{th:mev:callable-narrowing:observational-invariance}) to justify that the new sequence has enough funds. Here, we do not have that hypothesis.
    Nonetheless, we can still claim that the new sequence is well funded since we can simply take $\WmvAi$ to include all the tokens we need for that purpose. Indeed, wealthy adversaries can always provide funds for their attacks.
\end{proof}

\section{MEV non-interference}
\label{sec:non-interference}
\label{sec:nonint}

We formalise secure contract composability as a relation 
between blockchain states $\sysS$ and contract states $\cstD$:
intuitively, \mbox{$\nonint{\sysS}{\cstD}$} means that
the adversary cannot leverage $\sysS$ to inflict more MEV to $\cstD$
than it would be possible by interacting with $\cstD$ alone.
We will then say that
$\sysS$ is \emph{MEV non-interfering} with~$\cstD$.

\subsection{Definition}

Formally, MEV non-interference $\nonint{\sysS}{\cstD}$
holds when the MEV extractable from the contract accounts in $\cstD$
(\ie, $\cmvOfcst{\cstD}$) using \emph{any} contract in $\sysS \mid \cstD$
is exactly the same MEV that can be extracted using \emph{only}
the contracts in $\cstD$.
We write $\negnonint{\sysS}{\cstD}$ when $\nonint{\sysS}{\cstD}$ does not hold.

\begin{definition}[MEV non-interference]
  \label{def:non-interference}
  A state $\sysS$ is \emph{$\lmev{}{}{}$ non-interfering} with $\cstD$,
  in symbols $\nonint{\sysS}{\cstD}$, when
  \(
  \lmev{}{\sysS \mid \cstD}{\cmvOfcst{\cstD}} = \lmev{\cmvOfcst{\cstD}}{\sysS \mid \cstD}{\cmvOfcst{\cstD}}
  \).%
\end{definition}

Note that the inequality $\geq$ in~\Cref{def:non-interference} always holds by~\Cref{lem:mev:basic}\eqref{lem:mev:basic:L-leq-H}, so checking MEV non-interference only requires to check~$\leq$. 

The following example discriminates MEV non-interference
from Babel \emph{et al.}' $\epsilon$-composability,
showing that a contract with intended MEV but no interactions with the
context enjoys MEV non-interference, but it is not $\epsilon$-composable.


\begin{example}
  \label{ex:nonint:airdrop}
  Consider the airdrop contract in~\Cref{lst:airdrop},
  let $\sysS$ be any blockchain state, and let
  $\cstD = \walpmv{\contract{Airdrop}}{\waltok{n}{\tokT},\code{t}=\tokT}$.
  We have that $\nonint{\sysS}{\cstD}$, since:
  \[
  \lmev{}{\sysS \mid \cstD}{\setenum{\contract{Airdrop}}}
  \; = \;
  n \cdot \price{\tokT}
  \; = \;
  \lmev{\setenum{\contract{Airdrop}}}{\sysS \mid \cstD}{\setenum{\contract{Airdrop}}}
  \]
  Here, MEV non-interference correctly reflects the fact that $\Adv$ cannot leverage
  the contracts in $\sysS$ to damage the airdrop.
  \hfill\qedex
\end{example}

\begin{figure}[t]
\begin{lstlisting}[language=solidity
  ,classoffset=4,morekeywords={T,ETH},keywordstyle=\tokColor
  %,classoffset=5,morekeywords={C,D},keywordstyle=\cmvColor
  ,frame=single,
  ,caption={Contracts for~\Cref{ex:nonint:contract-dependency,cex:nonint-not-imply-richnonint}.}
  ,label={lst:nonint-not-imply-richnonint}
]
contract D {
  function f() { require (C.get()); T.send(msg.sender,1); }
}
contract C {
  bool x; uint fee;
  constructor(uint fee) { x=false; fee=fee_; } 
  function get() { return x; }  
  function set() { T.receive(fee); x=true; }
}
\end{lstlisting}
\end{figure}

We show some compositions that do not enjoy MEV non-interference, for different causes.
In~\Cref{ex:nonint:contract-dependency}, the interference is caused by contract dependencies, while in~\Cref{ex:nonint:token-dependency} it is caused by token flows.

\begin{example}
\label{ex:nonint:contract-dependency}
Consider the contracts in~\Cref{lst:nonint-not-imply-richnonint},
and let: 
\[
\sysS = \walu{\pmvM}{0}{\tokT}
\mid \walpmv{\contract{C}}{\waltok{0}{\tokT},\code{x}=\code{false},\code{fee}=0}
\qquad
\cstD = \walpmv{\contract{D}}{\waltok{100}{\ETH}}
\]
We have that
$\lmev{}{\sysS \mid \cstD}{\setenum{\cmvD}} = 100\cdot\price{\ETH}$,
while
$\lmev{\setenum{\cmvD}}{\sysS \mid \cstD}{\setenum{\cmvD}} = 0$,
since $\pmvM$ cannot call $\contract{C}.\txcode{set}()$.
Therefore, $\negnonint{\sysS}{\cstD}$.
\hfill\qedex
\end{example}

\begin{example}
  \label{ex:nonint:token-dependency}
  Consider the contracts in \Cref{lst:exchange,lst:airdrop}, and let:
  \begin{align*}
    \sysS & \; = \;
    \walu{\pmvM}{0}{\tokT}
    \mid     
    \walpmv{\contract{PriceOracle}}{\code{fee}=0, \code{prices}=\setenum{\tokT[1]=10, \tokT[2]=1}, \ldots}
    \\
    \cstC 
    & \; = \; \walpmv{\contract{Airdrop}}{\waltok{1}{\tokT[1]},\code{t}=\tokT[1]}
    \\
    \cstD 
    & \; = \;
    \walpmv{\contract{Exchange}}{\waltok{0}{\tokT[1]},\waltok{10}{\tokT[2]},\code{t0}=\tokT[1],\code{t1}=\tokT[2]}
  \end{align*}
  Assume that $\price{\tokT[1]} = \price{\tokT[2]}$, and let $\Adv = \setenum{\pmvM}$.
  We have that:   
  \[
  \lmev{}{\sysS \mid \cstC \mid \cstD}{\setenum{\contract{Exchange}}} = 10 \cdot \price{\tokT[2]} - 1 \cdot \price{\tokT[1]} = 9 \cdot \price{\tokT[2]} 
  \]
  since $\pmvM$ can first extract $1:\tokT[1]$
  from the airdrop, and then swap it for $10:\tokT[2]$ through the $\contract{Exchange}$.
  Instead, $\lmev{\setenum{\contract{Exchange}}}{\sysS \mid \cstC \mid \cstD}{\setenum{\contract{Exchange}}} = 0$,
  since $\pmvM$ cannot obtain the tokens $\tokT[1]$ needed to use the $\contract{Exchange}$.
  Then, $\negnonint{\sysS \mid \cstC}{\cstD}$.
  \hfill\qedex
\end{example}

\subsection{Sufficient conditions for MEV non-interference}

\Cref{th:nonint:sufficient-conditions} gives sufficient conditions
for MEV non-interference.
Condition~\ref{th:nonint:sufficient-conditions:zero-mev} states that
\mbox{$\nonint{\sysS}{}{\cstD}$} holds whenever the new contracts $\cstD$
are MEV-free:
a special case is when $\cstD$ has no tokens,
by \Cref{lem:mev:basic}\eqref{lem:mev:basic:leq-wealth}.
Condition~\ref{th:nonint:sufficient-conditions:adv-inert}
allows to establish \mbox{$\nonint{\sysS}{}{\cstD}$} even when $\cstD$ has MEV, provided that the token flows from $\sysS$ to $\cstD$ cannot be exploited by adversaries, and that $\cstD$ is observationally invariant \wrt adversarial moves in $\sysS$.
Recall that this property (see \Cref{def:observational-invariance}) ensures that any transaction sent by adversary to the contracts in $\sysS$ is not observable by the contracts in $\cstD$, and therefore the ability to call the contracts in $\sysS$ does not give the adversary any advantage compared to calling only $\cstD$.  
For instance, condition~\ref{th:nonint:sufficient-conditions:adv-inert} is applicable to establish MEV non-interference in~\Cref{ex:nonint:airdrop}, since 
the $\contract{Airdrop}$ contract is (trivially) observationally invariant \wrt any $\sysS$.

\begin{theorem}[Sufficient conditions for $\nonint{}{}$]
\label{th:nonint:sufficient-conditions}
Let $\sysS \mid \cstD$ be a blockchain state, not necessarily well-formed.
MEV non-interference $\nonint{\sysS}{\cstD}$ holds if any of the following conditions is satisfied: 
\begin{enumerate}

\item \label{th:nonint:sufficient-conditions:zero-mev}
$\lmev{}{\sysS \mid \cstD}{\cmvOfcst{\cstD}} = 0$


\item \label{th:nonint:sufficient-conditions:adv-inert}
There is no token flow from 
$\sysS$ to $\cstD$ in \mbox{$\sysS \mid \cstD$},
and $\cstD$ is observationally invariant in $\sysS$.


\end{enumerate}
\end{theorem}
\begin{proof}
For condition~\ref{th:nonint:sufficient-conditions:zero-mev} we have:
\begin{align*}
	0 &\leq 	\lmev{\cmvOfcst{\cstD}}{\sysS\mid \cstD}{\cmvOfcst{\cstD}} &&\text{by \Cref{lem:mev:basic}\eqref{lem:mev:basic:leq-wealth}}
	\\
	& \leq \lmev{}{\sysS\mid \cstD}{\cmvOfcst{\cstD}} &&\text{by~\Cref{lem:mev:basic}\eqref{lem:mev:basic:L-leq-H}}
	\\
	& = 0 &&\text{by hypothesis}
\end{align*}
which implies $\lmev{\cmvOfcst{\cstD}}{\sysS\mid \cstD}{\cmvOfcst{\cstD}} = 0$, and therefore $\nonint{\sysS}{\cstD}$.

\medskip\noindent
To prove that $\nonint{\sysS}{\cstD}$ under condition~\ref{th:nonint:sufficient-conditions:adv-inert}, observe that the equality:
\[
\lmev{\cmvOfcst{\cstD}}{\sysS\mid \cstD}{\cmvOfcst{\cstD}} = \lmev{}{\sysS\mid \cstD}{\cmvOfcst{\cstD}}
\]
follows directly from~\Cref{th:mev:callable-narrowing:observational-invariance}.
\end{proof}


We have already shown in~\Cref{ex:nonint:airdrop} a case where MEV non-interference does not imply condition~\ref{th:nonint:sufficient-conditions:zero-mev} of \Cref{th:nonint:sufficient-conditions}.
It is easy to show a case where MEV non-interference holds but condition~\ref{th:nonint:sufficient-conditions:adv-inert} does not hold.
For example, consider a contract $\contract{D}$ with a state variable that can be set to the value of some variable in a contract in $\sysS$, but does not otherwise affect the MEV extractable from $\contract{D}$. Then, observational invariance does not hold, but MEV non-interference holds.

We will see several applications of~\Cref{th:nonint:sufficient-conditions} laster in~\Cref{sec:defi}.


\subsection{Resistance to front-running attacks}

In general, \mbox{$\nonint{\sysS}{\cstD}$} does not imply \mbox{$\nonint{\sysS \mid \cstC}{\cstD}$}.
For instance, in~\Cref{ex:nonint:token-dependency} we showed that
$\negnonint{\sysS \mid \cstC}{\cstD}$, but if we remove the $\contract{Airdrop}$ contract then $\pmvM$ is no longer able to extract MEV from the $\contract{Exchange}$, since $\pmvM$'s wallet is empty: hence, $\nonint{\sysS}{\cstD}$.
This means that, in general, MEV non-interference is not robust against attacks
where adversaries manage to deploy contracts $\cstC$ \emph{before} $\cstD$.
Namely, assume that a user detects that $\nonint{\sysS}{\cstD}$,
and then sends a transaction to deploy the contracts $\cstD$.
If this transaction is front-run with another transaction
that deploys $\cstC$, then the new state $\sysS \mid \cstC$
could MEV-interfere with $\cstD$.

\Cref{th:nonint:state-narrowing} establishes sufficient conditions under which MEV non-interference $\nonint{\sysS}{\cstD}$ is preserved when the deployment of $\cstD$ is front-run with the deployment of $\cstC$.
Condition~\ref{condition:qnonint:preserving-interference:1} requires the contracts in $\cstD$ to be \emph{sender-agnostic}, \ie, their behaviour must not depend on the identity of the caller.
Condition~\ref{condition:qnonint:preserving-interference:2} requires that the deployment of $\cstC$ does not enable new token flows into $\cstD$.
If either condition is violated, contracts in $\cstC$ may enable forms of token extraction that were impossible when $\cstD$ interacted only with $\sysS$.
For instance, if a contract in $\cstD$ is not sender-agnostic, an access control check can succeed when the caller is a contract in $\cstC$
(see~\Cref{cex:th:nonint:state-narrowing:sender-agnostic}).
Similarly, token transfers from $\cstC$ to $\cstD$ may unlock MEV-extraction opportunities from $\cstD$
(see~\Cref{cex:th:nonint:state-narrowing:token-flows}).  


\begin{theorem}[Front-running resistance]
\label{th:nonint:state-narrowing}
\mbox{$\nonint{\sysS}{\cstD} \iff \nonint{\sysS \mid \cstC}{\cstD}$}
holds if
\begin{enumerate}

\item \label{condition:qnonint:preserving-interference:1}
the contracts in $\deps{\cstD}$ are sender-agnostic, and

\item \label{condition:qnonint:preserving-interference:2}
there are no token flows from 
$\cmvOfcst{(\sysS \mid \cstC)} \setminus \deps{\cstD}$
to $\deps{\cstD}$
in $\sysS \mid \cstC \mid \cstD$.
\end{enumerate}
\end{theorem}

\begin{proof}
The thesis can be written as:
\[
\begin{array}{rcl}
    \lmev{}{\sysS\mid\cstD}{\cmvOfcst{\cstD}} & = &
    \lmev{\cmvOfcst{\cstD}}{\sysS\mid\cstD}{\cmvOfcst{\cstD}}
    \\
    & \iff & \\
    \lmev{}{\sysS\mid \cstC\mid\cstD}{\cmvOfcst{\cstD}} & = &
    \lmev{\cmvOfcst{\cstD}}{\sysS\mid \cstC\mid\cstD}{\cmvOfcst{\cstD}}
\end{array}
\]
We prove the thesis by establishing the following two equalities:
\begin{align}
    \label{eq:qnonint:preserving-interference:1}
    \lmev{}{\sysS\mid\cstD}{\cmvOfcst{\cstD}} &= \lmev{}{\sysS \mid \cstC\mid\cstD}{\cmvOfcst{\cstD}}
    \\
    \label{eq:qnonint:preserving-interference:2}
    \lmev{\cmvOfcst{\cstD}}{\sysS\mid\cstD}{\cmvOfcst{\cstD}} &= \lmev{\cmvOfcst{\cstD}}{\sysS \mid \cstC\mid\cstD}{\cmvOfcst{\cstD}}
\end{align}

We start by observing that, since $\sysS \mid \cstC$ and $\sysS \mid \cstD$ are well-formed, then $\cstC$ has no dependencies in $\cstD$ and $\cstD$ has no dependencies in $\cstC$. 
Then, the states $\sysS \mid \cstC \mid \cstD$ and
$\sysS \mid \cstD \mid \cstC$ are equivalent:
\begin{equation}
\label{eq:qnonint:preserving-interference:state-equivalence}
\sysS \mid \cstC \mid \cstD 
\; = \;
\sysS \mid \cstD \mid \cstC
\end{equation}

For~\eqref{eq:qnonint:preserving-interference:1}, note that by~\Cref{lem:mev:basic}\eqref{lem:mev:basic:garbage}, 
if we define $\CmvD = \cmvOfcst{(\sysS \mid \cstC \mid \cstD)}$ we have:
\begin{equation}
\label{eq:nonint:state-narrowing:callable-narrowing}   
\lmev{}{\sysS \mid \cstC\mid\cstD}{\cmvOfcst{\cstD}}
=
\lmev{\CmvD}{\sysS \mid \cstC\mid\cstD}{\cmvOfcst{\cstD}}
\end{equation}
Using such $\CmvD$, the conditions of~\Cref{th:mev:frontrunning-resistance} are satisfied:
\begin{enumerate}

\item 
the contracts in $\deps{\cstD}$ are sender-agnostic
by hypothesis~\ref{condition:qnonint:preserving-interference:1};

\item 
$\deps{\cstD} \subseteq \CmvD$ holds trivially;

\item 
there are no token flows from 
$\deps{\CmvD} \setminus \deps{\cstD}$
to $\deps{\cstD}$
in $\sysS \mid \cstC \mid \cstD$,
by hypothesis~\ref{condition:qnonint:preserving-interference:2}.

\end{enumerate}
\noindent
Therefore, we obtain~\eqref{eq:qnonint:preserving-interference:1} from the chain of (in)equalities:
\begin{align*}
\lmev{}{\sysS \mid \cstC\mid\cstD}{\cmvOfcst{\cstD}}
& =
\lmev{\CmvD}{\sysS \mid \cstC\mid\cstD}{\cmvOfcst{\cstD}}
&& \text{by~\eqref{eq:nonint:state-narrowing:callable-narrowing}}
\\
& = \lmev{\CmvD \cap \deps{\cstD}}{\sysS \mid \cstD}{\cmvOfcst{\cstD}} 
&& \text{by~\Cref{th:mev:frontrunning-resistance}}
\\
& = \lmev{\deps{\cstD}}{\sysS \mid \cstD}{\cmvOfcst{\cstD}} 
&& \text{since $\deps{\cstD} \subseteq \CmvD$}
\\
& \leq 
\lmev{}{\sysS \mid \cstD}{\cmvOfcst{\cstD}} 
&& \text{by~\Cref{lem:mev:basic}\eqref{lem:mev:basic:L-leq-H}}
\\
& \leq
\lmev{}{\sysS \mid \cstD \mid \cstC}{\cmvOfcst{\cstD}}
&& \text{by~\Cref{lem:mev:basic}\eqref{lem:mev:basic:monotonicity}}
\\
& =
\lmev{}{\sysS \mid \cstC \mid \cstD}{\cmvOfcst{\cstD}}
&& \text{by~\eqref{eq:qnonint:preserving-interference:state-equivalence}}
\end{align*}

\noindent
For \eqref{eq:qnonint:preserving-interference:2}, we have that:
\begin{align*}
\lmev{\cmvOfcst{\cstD}}{\sysS\mid\cstD}{\cmvOfcst{\cstD}} 
& = 
\lmev{\cmvOfcst{\cstD}}{\sysS \mid \cstD \mid \cstC}{\cmvOfcst{\cstD}}
&& \text{by~\Cref{lem:mev:state-narrowing}, as $\cmvOfcst{\cstD} \subseteq \cmvOfcst{(\sysS \mid \cstD)}$}
\\
& =
\lmev{\cmvOfcst{\cstD}}{\sysS \mid \cstC \mid \cstD}{\cmvOfcst{\cstD}}
&& \text{by~\eqref{eq:qnonint:preserving-interference:state-equivalence}}
\tag*{\qedhere}
\end{align*}
\end{proof}

\begin{figure}[t]
\begin{lstlisting}[language=solidity
  ,classoffset=4,morekeywords={T,ETH},keywordstyle=\tokColor
  %,classoffset=5,morekeywords={C,D},keywordstyle=\cmvColor
  ,frame=single,
  ,caption={Contracts for~\Cref{cex:th:nonint:state-narrowing:sender-agnostic}.}
  ,label={lst:cex:th:nonint:state-narrowing:sender-agnostic}
]
contract D {
  function f() { require (Lock.b()); T.send(msg.sender,1); }
}
contract C {
  function g() { Lock.unlock(); }
}
contract Lock { // not sender-agnostic
  bool public b=false;
  function unlock() { require (msg.sender==C); b=true; }
}
\end{lstlisting}
\end{figure}

\begin{example}
\label{cex:th:nonint:state-narrowing:sender-agnostic}
Consider the contracts in~\Cref{lst:cex:th:nonint:state-narrowing:sender-agnostic}, let $\Adv = \setenum{\pmvM}$, and let:
\[
\sysS = \walu{\pmvM}{0}{\tokT} \mid \walpmv{\contract{Lock}}{\code{b}=\code{false}}
\qquad
\cstC = \walpmv{\contract{C}}{\waltok{0}{\tokT}}
\qquad
\cstD = \walpmv{\contract{D}}{\waltok{1}{\tokT}}
\]
In $\sysS \mid \cstD$, it is not possible to extract the $1:\tokT$ from $\contract{D}$, since $\pmvM$ cannot set \code{Lock.b}.
Therefore, $\lmev{}{\sysS \mid \cstD}{\setenum{\contract{D}}} = 0$, which implies $\nonint{\sysS}{\cstD}$.
Instead, in $\sysS \mid \cstC \mid \cstD$, if $\pmvM$ has unrestricted access to $\contract{C}$, then she can set \code{Lock.b} and then extract $1:\tokT$ from $\contract{D}$.
When the set of callable contracts is restricted to $\setenum{\contract{D}}$, instead, the MEV extractable from $\contract{D}$ is zero.
Therefore, $\negnonint{\sysS \mid \cstC}{\cstD}$.
Note that the contract $\contract{Lock} \in \deps{\cstD}$ is \emph{not} sender-agnostic, hence \Cref{th:nonint:state-narrowing} does not apply.
\hfill\qedex
\end{example}

\begin{figure}[t]
\begin{lstlisting}[language=solidity
  ,classoffset=4,morekeywords={T,ETH},keywordstyle=\tokColor
  %,classoffset=5,morekeywords={C,D},keywordstyle=\cmvColor
  ,frame=single,
  ,caption={Contracts for~\Cref{cex:th:nonint:state-narrowing:token-flows}.}
  ,label={lst:cex:th:nonint:state-narrowing:token-flows}
]
contract D {
  function f() { require (PaidLock.b()); T.send(msg.sender,1); }
}
contract C {
  function g() { T.approve(PaidLock,1); PaidLock.unlock(); }
}
contract PaidLock { // token flow from C
  bool public b=false;
  function unlock() { T.receive(1); b=true; }
}
\end{lstlisting}
\end{figure}
\begin{example}
\label{cex:th:nonint:state-narrowing:token-flows}
Consider the contracts in~\Cref{lst:cex:th:nonint:state-narrowing:token-flows}, let $\Adv = \setenum{\pmvM}$, and let:
\[
\sysS = \walu{\pmvM}{0}{\tokT} \mid \walpmv{\contract{PaidLock}}{\code{b}=\code{false}}
\qquad
\cstC = \walpmv{\contract{C}}{\waltok{1}{\tokT}}
\qquad
\cstD = \walpmv{\contract{D}}{\waltok{1}{\tokT}}
\]
The reasoning is similar to the in~\Cref{cex:th:nonint:state-narrowing:sender-agnostic}.
We have $\nonint{\sysS}{\cstD}$ because, without the token in $\contract{C}$, 
it is not possible to unlock the extraction of $1:\tokT$ from $\contract{D}$.
Instead, $\negnonint{\sysS \mid \cstC}{\cstD}$, since $\pmvM$ can call $\contract{C}$ to set \code{PaidLock.b} and then extract $1:\tokT$ from~$\contract{D}$.
Unlike~\Cref{cex:th:nonint:state-narrowing:sender-agnostic}, all contracts are sender-agnostic, but there is a token flow from $\contract{C}$ to $\contract{PaidLock} \in \deps{\cstD}$, so \Cref{th:nonint:state-narrowing} does not apply.
\hfill\qedex
\end{example}
\section{MEV non-interference for wealthy adversaries}
\label{sec:rich-nonint}

MEV non-interference $\nonint{\sysS}{\cstD}$ may or may not hold depending on the wealth of the adversary in $\sysS$,
as we have already observed in~\Cref{lem:mev:wallet}
that MEV depends on the adversary's wallets.
As noted before, assuming bounds on the capital available to the adversary
may be unsafe, and consequently we introduced in~\Cref{def:rich-mev} a notion of MEV that assumes an unbounded adversary's wealth.
We now introduce a wealthy-adversaries counterpart to MEV non-interference:
$\richnonint{\cstC}{\cstD}$ holds when contracts $\cstC$ do not interfere with the MEV extractable from $\cstD$ \emph{regardless} of the adversary's wealth.

\begin{definition}[MEV non-interference against wealthy adversaries]
  \label{def:rich-non-interference}
  A contract state $\cstC$ is \emph{$\rlmev{}{}{}$ non-interfering} with $\cstD$,
  in symbols $\richnonint{\cstC}{\cstD}$, when
  \[
  \rlmev{}{\cstC \mid \cstD}{\cmvOfcst{\cstD}}
  =
  \rlmev{\cmvOfcst{\cstD}}{\cstC \mid \cstD}{\cmvOfcst{\cstD}}
  \]
\end{definition}

\Cref{lem:rich-nonint-nonint} characterizes
$\richnonint{}{}$ in terms of $\nonint{}{}$:
intuitively, $\richnonint{\cstC}{\cstD}$ holds whenever
$\nonint{\WmvA \mid \cstC}{\cstD}$ holds for rich enough adversaries' wallets.

\begin{proposition} 
  \label{lem:rich-nonint-nonint}
  \mbox{$\richnonint{\cstC}{\cstD}$} if and only if
  \mbox{$\exists \WmvA[0] . \ \forall \WmvA \geq_{\$} \WmvA[0] . \ \nonint{\WmvA \mid \cstC}{\cstD}$}.
\end{proposition}


Of course, we cannot deduce \mbox{$\richnonint{\cstC}{\cstD}$} 
when \mbox{$\nonint{\WmvA \mid \cstC}{\cstD}$} holds for \emph{some}~$\WmvA$:
indeed, a ``poor'' adversary could not be able to call some MEV-triggering method in $\cstC$, while a wealthy one could:
then, we could observe a discrepancy between the restricted and
unrestricted $\rlmev{}{}{}$ which was not visible for $\lmev{}{}{}$.
The following~\namecref{cex:nonint-not-imply-richnonint} elaborates on this informal argument.

\begin{example}
  \label{cex:nonint-not-imply-richnonint}
  The implication
  ``if $\nonint{\WmvA \mid \cstC}{\cstD}$ then $\richnonint{\cstC}{\cstD}$''
  does not hold.
  To show that, consider again the contracts in~\Cref{lst:nonint-not-imply-richnonint},
  and let: 
  \[
  \WmvA = \walu{\pmvM}{0}{\tokT}
  \qquad
  \cstC = \walpmv{\contract{C}}{\waltok{0}{\tokT},\code{x}=\code{false},\code{fee}=1}
  \qquad
  \cstD = \walpmv{\contract{D}}{\waltok{100}{\ETH}}
  \]
  For $\nonint{}{}$, we have that
  $\lmev{}{\WmvA \mid \cstC \mid \cstD}{\setenum{\cmvD}} = 0$,
  since $\pmvM$ does not have the $1:\tokT$ needed to extract the $\ETH$ from $\cstD$.
  Therefore, $\nonint{\WmvA \mid \cstC}{\cstD}$.
  For $\richnonint{}{}$, we have that
  $\rlmev{}{\cstC \mid \cstD}{\setenum{\cmvD}} = 100\cdot\price{\ETH}$,
  while
  $\rlmev{\setenum{\cmvD}}{\cstC \mid \cstD}{\setenum{\cmvD}} = 0$,
  since $\pmvM$ cannot call the $\txcode{set}$ method of $\contract{C}$.
  Then, $\negrichnonint{\cstC}{\cstD}$.
  \hfill\qedex
\end{example}

The following~\namecref{th:rich-nonint:sufficient-conditions}
refines \Cref{th:nonint:sufficient-conditions},
giving sufficient conditions for $\richnonint{}{}$.
Note that the absence of token flows is no longer required.

\begin{theorem}[Sufficient conditions for $\richnonint{}{}$]
  \label{th:rich-nonint:sufficient-conditions}
  For any state $\cstC \mid \cstD$ (not necessarily well-formed),
  each of the following conditions implies
  \mbox{$\richnonint{\cstC}{\cstD}$}:
  \begin{enumerate}

  \item \label{th:rich-nonint:sufficient-conditions:zero-mev}
    $\rlmev{}{\cstC \mid \cstD}{\cmvOfcst{\cstD}} = 0$;

  \item \label{th:rich-nonint:sufficient-conditions:no-deps}
    $\deps{\cstD}$ and $\cmvOfcst{\cstC}$ are disjoint;

  \item \label{th:rich-nonint:sufficient-conditions:adv-inert}
    $\cstD$ is observationally invariant in $\WmvA \mid \cstC$, for all sufficiently large $\WmvA$.
    
  \end{enumerate}
\end{theorem}
\begin{proof}
For~\Cref{th:rich-nonint:sufficient-conditions:zero-mev}, $\richnonint{\cstC}{\cstD}$ is implied by: 
\begin{align*}
	0 &\leq 	\rlmev{\cmvOfcst{\cstD}}{\cstC\mid \cstD}{\cmvOfcst{\cstD}} &&\text{by~\Cref{lem:rich-mev:basic}\eqref{lem:rich-mev:basic:leq-wealth}}
	\\
	& \leq \rlmev{}{\cstC\mid \cstD}{\cmvOfcst{\cstD}} &&\text{by~\Cref{lem:rich-mev:basic}\eqref{lem:rich-mev:basic:L-leq-H}}
	\\
	& = 0 &&\text{by hypothesis}
\end{align*}

For~\Cref{th:rich-nonint:sufficient-conditions:no-deps}, assume that $\deps{\cstD}$ and $\cmvOfcst{\cstC}$ are disjoint.
We prove that $\cstD$ is observationally invariant in $\WmvA \mid \cstC$, for all $\WmvA$.
\Cref{def:observational-invariance} requires checking the effects of all transactions $\txY$ fired in \mbox{$\WmvA \mid \cstC \mid \cstD$} which trigger a call
$\contract{D}:\contract{C}.\txcode{f}(\code{args})$ with $\contract{D} \in \cmvOfcst{\cstD}$ and $\contract{C} \in \cmvOfcst{(\WmvA \mid \cstC)} = \cmvOfcst{\cstC}$.
This implies that $\contract{C} \in \deps{\cstD} \cap \cmvOfcst{\cstC}$, contradicting
the hypothesis in~\Cref{th:rich-nonint:sufficient-conditions:no-deps}.

Since the assumption of~\Cref{th:rich-nonint:sufficient-conditions:no-deps} implies that of~\Cref{th:rich-nonint:sufficient-conditions:adv-inert}, we only need to prove the theorem in the latter case.
Under this assumption, the thesis follows directly from~\Cref{th:rich-mev:callable-narrowing:observational-invariance}.
\end{proof}

\hidden{
Nota 1:
	Credo sia possibile indebolire le ipotesi del teorema 4, dimostrandolo in maniera molto simile alla dimostrazione del Lemma 3B. 
	Il nuovo teorema diventerebbe qualcosa del tipo: 
	\begin{itemize}
	\item $\richnonint{\cstC}{\cstD} \implies \richnonint{\strip{\cstC}{\cmvOfcst{\cstD}}}{\cstD} $ incondizionatamente
	\item  $\richnonint{\strip{\cstC}{\cmvOfcst{\cstD}}}{\cstD} \implies \richnonint{\cstC}{\cstD}$ sotto le condizioni di sender-agnostic che ci sono ora, quindi $\deps{\cstD} \cap \deps{\cmvOfcst{\cstC}\setminus\deps{\cstD}}$ devono essere sender-agnostic.
	\end{itemize}
}

\Cref{th:rich-nonint:state-narrowing} is the wealthy-adversaries counterpart of \Cref{th:nonint:state-narrowing}, establishing that MEV non-interference is robust against front-running the deployment of $\cstD$ with adversarial contracts $\cstC$ whenever the dependencies of $\cstD$ are sender-agnostic.
Intuitively, this holds because any MEV extraction exploiting $\cstC$ can also be performed by directly calling the dependencies of $\cstD$, since they are not influenced by the caller's identity.
\bartnote{per l'implicazione ``if'', non serve l'ipotesi sender-agnostic}

\begin{theorem}[Front-running resistance]
\label{th:rich-nonint:state-narrowing}
\mbox{$\richnonint{\cstS}{\cstD} \iff \richnonint{\cstS \mid \cstC}{\cstD}$}
holds if  
the contracts in $\deps{\cstD}$ are sender-agnostic.
\end{theorem}
\begin{proof}
We replay the proof of~\Cref{th:nonint:state-narrowing}, replacing $\nonint{}{}$ with its wealthy-adversary counterpart $\richnonint{}{}$, and exploiting the wealthy-adversary counterparts of the used lemmas about MEV.
More in detail, each use of~\Cref{th:mev:callable-narrowing:deps} within the proof of~\Cref{th:nonint:state-narrowing} is replaced here with~\Cref{th:rich-mev:callable-narrowing:deps}, which only requires the sender-agnostic assumption on $\deps{\cstD}$.
\end{proof}

\section{Application to DeFi compositions}
\label{sec:defi}

We now analyse MEV non-interference of DeFi compositions, considering as basic building blocks typical DeFi protocols.
We summarize the results of the analysis in~\Cref{tab:defi-compositions}, where
each line shows a compound contract $\cstD$, the context $\cstC$, 
and whether $\nonintrel$ and $\richnonintrel$ hold (\cmark) or not (\xmark).
In the table we omit the actual wallet and contract states, with the following interpretation: 
\begin{itemize}

\item \cmark\ in the $\nonint{\WmvA \mid \cstC}{\cstD}$ column means that non-interference holds for all $\WmvA$, 
and for all $\cstC$ and $\cstD$ of the given form; 
\xmark\ means that $\nonint{}{}$ fails to hold
for some $\WmvA$, and for some $\cstC$ and $\cstD$ of the given form;

\item \cmark\ in the $\richnonint{\cstC}{\cstD}$ column means that non-interference holds for all $\cstC$ and $\cstD$ of the given form; 
\xmark\ means that $\richnonint{}{}$ fails to hold
for some $\cstC$ and $\cstD$ of the given form.

\end{itemize}

Entries marked \cmark\ follow by~\Cref{th:nonint:sufficient-conditions,th:rich-nonint:sufficient-conditions}, and indicate which case of the respective theorem applies.
For entries marked \xmark\ we provide a counterexample in the respective section. 
Note that the \cmark\ and \xmark\ results may still hold when $\cstC$ is extended with further contracts,
whenever the conditions of~\Cref{th:nonint:state-narrowing,th:rich-nonint:state-narrowing} are satisfied.

\newcommand{\cstAny}{\cdots}

\begin{table}[t!]
  \caption{MEV non-interference of DeFi compositions.}
  \label{tab:defi-compositions}
  \setlength{\tabcolsep}{7pt} 
  \renewcommand{\arraystretch}{1.05} 
  \centering
  \begin{tabular}{|c|c|c|c|}
  \hline
  \textbf{Dependencies $\cstC$} 
  & \textbf{New contracts $\cstD$} 
  & \textbf{$\nonint{\WmvA \mid \cstC}{\cstD}$}
  & \textbf{$\richnonint{\cstC}{\cstD}$}
  \\
  \hline
  $\interface{IPriceOracle} \mid \cstAny$
  & $\contract{Exchange}$
  & \xmark
  & \xmark
  \\
  $\contract{PriceOracle} \mid \cstAny$
  & $\contract{Exchange}$
  & \xmark
  & \cmark${}^{\ref{th:rich-nonint:sufficient-conditions:adv-inert}}$
  \\
  $\cstAny$
  & $\contract{AMM}$
  & \xmark
  & \cmark${}^{\ref{th:rich-nonint:sufficient-conditions:no-deps}}$
  \\
  $\contract{AMM}$
  & $\contract{Option}$
  & \xmark
  & \xmark
  \\
  $\contract{PriceOracle}$
  & $\contract{Option}$
  & \cmark${}^{\ref{th:nonint:sufficient-conditions:adv-inert}}$
  & \cmark${}^{\ref{th:rich-nonint:sufficient-conditions:adv-inert}}$
  \\
  $\interface{IExchange} \mid \interface{IExchange}$
  & $\contract{SwapProxy}$
  & \cmark${}^{\ref{th:nonint:sufficient-conditions:zero-mev}}$
  & \cmark${}^{\ref{th:rich-nonint:sufficient-conditions:zero-mev}}$
  \\
  $\interface{IExchange} \mid \interface{IExchange}$
  & $\contract{SwapRouter}$
  & \cmark${}^{\ref{th:nonint:sufficient-conditions:zero-mev}}$
  & \cmark${}^{\ref{th:rich-nonint:sufficient-conditions:zero-mev}}$
  \\
  $\contract{PriceOracle} \mid \contract{LinearInterestRate}$
  & $\contract{LP}$
  & \cmark${}^{\ref{th:nonint:sufficient-conditions:adv-inert}}$
  & \cmark${}^{\ref{th:rich-nonint:sufficient-conditions:adv-inert}}$
  \\  
  $\contract{AMM} \mid \interface{IInterestRate}$
  & $\contract{LP}$
  & \xmark
  & \xmark
  \\  
  $\interface{IExchange} \mid \interface{IExchange} \mid \interface{ILoaner}$
  & $\contract{LPArbitrage}$
  & \cmark${}^{\ref{th:nonint:sufficient-conditions:zero-mev}}$
  & \cmark${}^{\ref{th:rich-nonint:sufficient-conditions:zero-mev}}$
  \\
  $\interface{IExchange} \mid \interface{IExchange} \mid \interface{IFlashLoaner}$
  & $\contract{FlashLoanArbitrage}$
  & \cmark${}^{\ref{th:nonint:sufficient-conditions:zero-mev}}$
  & \cmark${}^{\ref{th:rich-nonint:sufficient-conditions:zero-mev}}$
  \\
  \hline
  \end{tabular}
\end{table}


\subsection{Exchange}
\label{sec:defi:exchange}

The $\contract{Exchange}$ contract in~\Cref{lst:exchange} does not enjoy neither $\lmev{}{}{}$ nor $\rlmev{}{}{}$ non-interferences \wrt a generic context.
For example, $\lmev{}{}{}$ interference may occur when the context includes a contract from which the adversary can receive tokens needed to interact with the $\contract{Exchange}$, \eg an $\contract{Airdrop}$ or another $\contract{Exchange}$.
Instead, $\rlmev{}{}{}$ interference may occur when the $\interface{IPriceOracle}$ instance upon which the $\contract{Exchange}$ depends can be manipulated by the adversary.
This happens \eg when the price oracle is an $\contract{AMM}$.

When, instead, the $\contract{Exchange}$ depends on the $\contract{PriceOracle}$ contract in~\Cref{lst:exchange}, it is observationally invariant \wrt adversarial moves in the context.
Therefore, $\rlmev{}{}{}$ non-interference holds by~\Cref{th:rich-nonint:sufficient-conditions}\eqref{th:rich-nonint:sufficient-conditions:adv-inert}. 
Note that $\lmev{}{}{}$ non-interference may still fail because of token flows to the adversary.


\subsection{AMM}
\label{sec:defi:amm}

We specify in~\Cref{lst:amm} an Automated Market Maker (AMM) inspired by Uniswap v2 \cite{Angeris20aft,Angeris21analysis,Xu21sok}.
The AMM stores reserves of two token types $\tokT[0]$ and $\tokT[1]$; 
these types and their initial reserves are defined upon deployment.
The method $\txcode{swap}$ allows anyone to exchange units of $\tokT[0]$ for units of $\tokT[1]$ and vice-versa, provided that the product between their reserves remains constant.
More specifically, a transaction
$\pmvA:\txcode{swap}(\tokT[0],x,y_{\min})$
allows $\pmvA$ to send $\waltok{x}{\tokT[0]}$ to the AMM,
and receive at least \mbox{$\waltok{y_{\min}}{\tokT[1]}$} in exchange.
Symmetrically, $\pmvA:\txcode{swap}(\tokT[1],x,y_{\min})$
allows $\pmvA$ to exchange $\waltok{x}{\tokT[1]}$ for at least
\mbox{$\waltok{y_{\min}}{\tokT[0]}$}.
The method $\txcode{getRate}$ gives the exchange rate between the two token types, thereby making the AMM usable as a price oracle.

AMMs have a well-known source of MEV when the internal exchange rate given by $\txcode{getRate}$ is not aligned with the external rate (\ie, the ratio of the token prices given by the price function $\price{}$).
In this case, MEV can be extracted by firing a $\txcode{swap}$ transaction that realigns the internal with the external rate~\cite{BCL22amm}.

\begin{figure}[tb!]
\begin{lstlisting}[language=solidity
  % ,morekeywords={AMM,addLiquidity,swap,getTokens,getRate}
  % ,classoffset=4,morekeywords={T0,T1,t},keywordstyle=\tokColor
  ,frame=single
  ,caption={A constant-product AMM contract.}
  ,label={lst:amm}]
contract AMM is IPriceOracle, IExchange {
  token immutable public t0; token immutable public t1;
  
  constructor(token t0_, uint x0, token t1_, uint x1) { 
    t0=t0_; t1=t1_; require t0!=t1 && x0>0 && x1>0; 
    t0.receive(x0); t1.receive(x1);
  }
  // exchange rate to sell ti and buy to (multiplied by 1e6)
  function getRate(token ti, token to) public returns(uint) {
    require (ti==t0 && to==t1) || (ti==t1 && to==t0);
    uint ri = ti.balanceOf(this);
    uint ro = to.balanceOf(this);
    return (1_000_000 * ro)/ri; // note ro,ri flipped wrt Exchange
  }
  // list of tokens swappable through the AMM
  function getTokens() public view returns (token[] memory) {
    token[] memory ts = new token[](2); ts[0]=t0; ts[1]=t1; return ts;
  }
  // swaps x:ti for at least ymin:to
  function swap(token ti, uint x, token to, uint ymin) public returns (uint) {
    require (ti==t0 && to==t1) || (ti==t1 && to==t0);
    ti.receive(x);
    uint y = (x * getRate(ti,to))/1_000_000;
    require (y>=ymin);
    to.send(msg.sender,y); 
    return y;
  }
}
\end{lstlisting}
\end{figure}


To study MEV non-interference, consider two AMM instances for the pairs $(\tokT[0],\tokT[1])$ and $(\tokT[1],\tokT[2])$.
Let $\Adv = \setenum{\pmvM}$, 
let $\price{\tokT[0]} = \price{\tokT[1]} = \price{\tokT[2]} = 1$, and let:
\[
\sysS = 
\walpmv{\pmvM}{\waltok{3}{\tokT[0]}} \mid
\walpmv{\contract{AMM1}}{\waltok{6}{\tokT[0]},\waltok{6}{\tokT[1]}}
\qquad
\cstD =
\walpmv{\contract{AMM2}}{\waltok{4}{\tokT[1]},\waltok{9}{\tokT[2]}}
\]
The internal exchange rate for selling $\tokT[1]$ and buying $\tokT[2]$ in $\contract{AMM2}$ is $\nicefrac{4}{9}$, while the external rate is $1$:
therefore, it is possible to extract MEV from $\contract{AMM2}$.
Since $\pmvM$ does not have any units of $\tokT[1]$ that are needed to swap through $\contract{AMM2}$, extracting such MEV requires first obtaining the needed units through $\contract{AMM1}$.
The sequence of transactions that maximises the loss of $\contract{AMM2}$
is the following:
\begin{align*}
\sysS
& \xrightarrow{\pmvM:\contract{AMM1}.\txcode{swap}(\tokT[0],3,0)}
&& \hspace{-80pt}
\walu{\pmvM}{2}{\tokT[1]} \mid
\walpmv{\contract{AMM1}}{\waltok{9}{\tokT[0]},\waltok{4}{\tokT[1]}} \mid
\walpmv{\contract{AMM2}}{\waltok{4}{\tokT[1]},\waltok{9}{\tokT[2]}}  
\\
& \xrightarrow{\pmvM:\contract{AMM2}.\txcode{swap}(\tokT[1],2,0)}
&& \hspace{-80pt}
\walu{\pmvM}{3}{\tokT[2]} \mid
\walpmv{\contract{AMM1}}{\waltok{9}{\tokT[0]},\waltok{4}{\tokT[1]}} \mid
\walpmv{\contract{AMM2}}{\waltok{6}{\tokT[1]},\waltok{6}{\tokT[2]}}
\end{align*}
This causes $\contract{AMM2}$ to receive $2:\tokT[1]$ and send $3:\tokT[2]$, for an overall loss of $1$. 
Therefore,  
$\lmev{}{\sysS}{\setenum{\contract{AMM2}}} = 1$.
Instead, when $\pmvM$ is restricted to use
$\CmvD = \setenum{\contract{AMM2}}$,
it has no way to obtain the tokens $\tokT[1]$ that are needed to
extract value from $\contract{AMM2}$: therefore,
$\lmev{\setenum{\contract{AMM2}}}{\sysS}{\setenum{\contract{AMM2}}} = 0$.
Summing up, $\negnonint{\sysS}{\cstD}$.
Similar attacks are possible whenever the context allows the adversary to obtain tokens needed to interact with the AMM in~$\cstD$.

For wealthy adversaries, instead, AMMs are always $\lmev{}{}{}$ non-interfering with any context.
This holds because wealthy adversaries always have enough tokens to attack the new AMM in $\cstD$,
without the need of exploiting the context $\cstC$.
This follows from~\Cref{th:rich-nonint:sufficient-conditions}\eqref{th:rich-nonint:sufficient-conditions:no-deps}.


\subsection{Binary option}
\label{sec:defi:binaryoption}

\begin{figure}[!htbp]
  \begin{lstlisting}[language=solidity
  %,morekeywords={win,close,getTokens,getRate},classoffset=4,morekeywords={a,b,A},keywordstyle=\pmvColor,classoffset=5,morekeywords={t,ETH},keywordstyle=\tokColor,classoffset=6,morekeywords={oracle},keywordstyle=\cmvColor
  ,frame=single
  ,caption={A binary option contract relying on an external price oracle.}
  ,label={lst:option}
]
contract Option {
  address owner;   // option seller
  address buyer;   // option holder
  token t;         // underlying asset
  uint strike;     // strike price (multiplied by 1_000_000)
  uint expiry;     // expiration time
  uint fee;        // option price paid by buyer
  uint payout;     // fixed payoff if strike reached
  IPriceOracle public oracle;
  
  constructor(token t_,uint s,uint e,uint f,uint p,IPriceOracle o) {
    t=t_; strike=s; expiry=e; fee=f; payout=p; oracle=o; owner=msg.sender;
    ETH.receive(payout);
    require (expiry>block.number);
  }
  
  function buy() public {
    require (buyer==address(0)); // option has not been sold yet
    ETH.receive(fee);
    ETH.send(owner,fee);
    buyer=msg.sender;
  }
  
  function exercise() public {
    require (block.number<=expiry && buyer!=address(0) && msg.sender==buyer);
    uint price=oracle.getRate(t,ETH);
    require (price>=strike);
    ETH.send(buyer,payout);
  }

  function reclaim() public {
    require (block.number>expiry && msg.sender==owner);
    ETH.send(owner,payout);
  }
}
\end{lstlisting}
\end{figure}

The $\contract{Option}$ contract in~\Cref{lst:option} implements a minimal american binary call option settled via a price oracle. 
A seller deploys the contract by locking a payoff in $\ETH$, and choosing a target token $\tokT$ and a strike price it has to reach in order to claim the option.  
A buyer may then acquire the option by paying a fixed fee. 
Before expiration, if the price of the token $\tokT$ is greater than or equal to the strike price, the buyer can redeem the payoff. 
Otherwise, the payoff is returned to the seller after expiration. 

Consider an instance of $\contract{Option}$ using an $\contract{AMM}$
as price oracle, with payout of $10:\ETH$, fee of $1:\ETH$, owner $\pmvA$, target token $\tokT$, and strike price $2 \cdot 10^6$:
\begin{align*}
\sysS & = 
\walu{\pmvM}{301}{\ETH} \mid
\walpmv{\contract{AMM}}{\waltok{600}{\ETH},\waltok{600}{\tokT}} \mid
\code{block.number}=n-k \mid
\cdots
\\
\cstD & = \walpmv{\contract{Option}}{\waltok{10}{\ETH}, \code{owner}=\pmvA,\code{buyer}=0,\code{t}=\tokT,\code{strike}=\code{2e6},\code{expiry}=n,\code{fee}=1, \ldots}
\end{align*}
Note that the current exchange rate is $1 \cdot 10^6$, while exercising the option requires to make it exceed $2 \cdot 10^6$.
Since the option has not expired ($\code{block.num}=n-k < n$) and the adversary $\pmvM$ is has enough $\ETH$, she can fire the following sequence:
\begin{align*}
\sysS \mid \cstD
& 
\xrightarrow{\pmvM:\contract{Option}.\txcode{buy}()}
&& \hspace{-40pt} 
\walu{\pmvM}{300}{\ETH} \mid
\walpmv{\contract{AMM}}{\waltok{600}{\ETH},\waltok{600}{\tokT}} \mid
\walpmv{\contract{Option}}{\waltok{20}{\ETH},\cdots}
\\
& 
\xrightarrow{\pmvM:\contract{AMM}.\txcode{swap}(\ETH,300,0)}
&& \hspace{-40pt}
\walu{\pmvM}{200}{\tokT} \mid
\walpmv{\contract{AMM}}{\waltok{900}{\ETH},\waltok{400}{\tokT}} \mid
\walpmv{\contract{Option}}{\waltok{20}{\ETH},\cdots}
\\
&
\xrightarrow{\pmvM:\contract{Option}.\txcode{exercise}()}
&& \hspace{-40pt}
\walpmv{\pmvM}{\waltok{20}{\ETH},\waltok{200}{\tokT}} \mid
\walpmv{\contract{AMM}}{\waltok{900}{\ETH},\waltok{400}{\tokT}} \mid
\walpmv{\contract{Option}}{\waltok{0}{\ETH},\cdots}
\\
&
\xrightarrow{\pmvM:\contract{AMM}.\txcode{swap}(\tokT,200,0)}
&& \hspace{-40pt}
\walu{\pmvM}{320}{\ETH} \mid
\walpmv{\contract{AMM}}{\waltok{600}{\ETH},\waltok{600}{\tokT}} \mid
\walpmv{\contract{Oracle}}{\waltok{0}{\ETH},\cdots}
\end{align*}
where the third transition succeeds because $\contract{AMM}.\txcode{getRate}(\tokT,\ETH) = \nicefrac{9}{4} \cdot 10^6 > 2 \cdot 10^6$.
This shows that, when $\pmvM$ can interact with $\contract{Option}$'s dependency $\contract{AMM}$, the maximal loss inflictable by $\pmvM$ to the option contract is: 
\[
\lmev{}{\sysS \mid \cstD}{\setenum{\contract{Option}}} 
\; = \;
10 \cdot \price{\ETH}
\]
Instead, when $\pmvM$ is restricted to calling only $\contract{Option}$, then $\pmvM$ has no way to remove tokens from the contract in the state $\sysS$, and so:
\[
\lmev{\setenum{\contract{Option}}}{\sysS \mid \cstD}{\setenum{\contract{Option}}} 
\; = \; 
0
\]
Therefore, $\contract{Option}$ does not enjoy MEV non-interference with $\sysS$.
Note that a poorer $\pmvM$ may not have enough $\ETH$
to produce the short-term volatility in the AMM's exchange rate.

MEV non-interference holds, instead, when using the $\contract{PriceOracle}$ in~\Cref{lst:price-oracle}, since no adversary (except the oracle owner itself) can manipulate the exchange rate.
Formally, let $\sysS$ be a contract state including the $\contract{PriceOracle}$.
We have that:
\begin{inlinelist}
\item there are no indirect token flows from $\sysS$ to the option contract, and 
\item the option contract is observationally invariant in $\sysS$.
\end{inlinelist}
Therefore, $\nonint{S}{\cstD}$ follows from~\Cref{th:nonint:sufficient-conditions}\eqref{th:nonint:sufficient-conditions:adv-inert}.
By~\Cref{th:rich-nonint:sufficient-conditions}\eqref{th:rich-nonint:sufficient-conditions:adv-inert}, MEV non-interference also holds against wealthy adversaries.

Note that, while the composition $\contract{AMM}$/$\contract{Option}$ is correctly classified as \emph{non} $0$-composable by~\cite{Babel23clockwork},
the same happens for $\contract{PriceOracle}$/$\contract{Option}$
in the states when the adversary (honestly) wins the bet.
This seems intuitively incorrect, but technically $0$-composability does not hole because in~\cite{Babel23clockwork} an adversary winning the bet is indistinguishable from extracting MEV.


\subsection{Swap proxy}
\label{sec:defi:swapproxy}

\begin{figure}[!tb]
\begin{lstlisting}[
    ,language=solidity
    % ,morekeywords={SwapProxy,AMM,swap,getTokens,getRate},classoffset=4,morekeywords={a,b,A,Oracle},keywordstyle=\pmvColor,classoffset=5,morekeywords={ETH,t,t0,t1,tout},keywordstyle=\tokColor,classoffset=6,morekeywords={c0,c1,exch},keywordstyle=\cmvColor
    ,caption={A contract that composes two exchanges, swapping with the one with higher rate.}
    ,label={lst:swapproxy}
    ,frame=single]
contract SwapProxy is IExchange {
  IExchange immutable ex0; IExchange immutable ex1; 
  
  constructor(IExchange ex0_, IExchange ex1_) { ex0=ex0_; ex1=ex1_; }
  
  function getRate(token ti, token to) public view returns(uint) {
    uint r0 = ex0.getRate(ti,to);
    uint r1 = ex1.getRate(ti,to);
    return (r0 >= r1) ? r0 : r1;
  }
  // assumption: the two exchanges handle the same token types
  function getTokens() public view returns (token[] memory) {
    return ex0.getTokens();
  }

  function swap(token ti, uint x, token to, uint ymin) public returns (uint) {
    ti.receive(x);
    IExchange bex; // the exchange with the maximum exchange rate 
    bex = (ex0.getRate(ti,to) >= ex1.getRate(ti,to)) ? ex0 : ex1;
    ti.send(bex,x);
    uint y = bex.swap(ti,x,to,ymin); // reverts if bex outputs < ymin:to 
    to.receive(y);
    to.send(msg.sender,y);
    return y;
  }
}
\end{lstlisting}
\end{figure}

The $\contract{SwapProxy}$ contract in~\Cref{lst:swapproxy} composes two exchanges, providing a $\txcode{swap}$ method that forwards the call to one of the underlying exchanges. 
More precisely, when $\txcode{swap}$ is called, it queries the current rates of the two underlying exchanges, and selects the one with the maximum rate to perform the swap.
The composed exchanges must implement the $\interface{IExchange}$ interface in~\Cref{lst:exchange} (for example, possible instantiations are the $\contract{Exchange}$ and the $\contract{AMM}$ contracts), and must handle the same token types.
Note that, in general, $\contract{SwapProxy}$ does not guarantee the optimal swap that one could achieve leveraging the underlying exchanges, for various reasons.
First, the queried exchange rate is only the \emph{marginal} rate, \ie the rate obtained for infinitely small trades.
In some exchanges (\eg, in the $\contract{AMM}$), the actual rate also depends on the input amount: swapping more tokens causes the rate to decrease (diminishing returns).
Second, it may happen that none of the underlying exchanges has a sufficient amount of output tokens to cover the input: in that case, one could have to split the trade into two smaller trades.

Since the balance of $\contract{SwapProxy}$ is zero in any blockchain state, no tokens can be extracted, and so its MEV is zero.
Therefore, by~\Cref{th:nonint:sufficient-conditions}\eqref{th:nonint:sufficient-conditions:zero-mev}, it enjoys MEV non-interference against any blockchain state --- regardless of the specific instantiations of the $\interface{IExchange}$ implemented by the underlying exchanges.
MEV non-interference also holds against wealthy adversaries, by~\Cref{th:rich-nonint:sufficient-conditions}\eqref{th:rich-nonint:sufficient-conditions:zero-mev}.



\subsection{Swap router}
\label{sec:defi:swaprouter}

\begin{figure}[tbp!]
\begin{lstlisting}[language=solidity
    %,morekeywords={swap,getTokens,getRate},classoffset=4,morekeywords={a,b,A,Oracle},keywordstyle=\pmvColor,classoffset=5,morekeywords={ETH,t,t0,t1,t2,t1b},keywordstyle=\tokColor,classoffset=6,morekeywords={SwapRouter,c0,c1,exch},keywordstyle=\cmvColor
    ,caption={A contract to route swaps across two exchanges.}
    ,label={lst:swaprouter}
    ,frame=single]
contract SwapRouter {
  IExchange immutable ex0; IExchange immutable ex1; 

  constructor(IExchange ex0_,IExchange ex1_) { ex0=ex0_; ex1=ex1_; }

  function getTokens() public view returns (address[] memory) {
    address[] memory ts0 = ex0.getTokens();
    address[] memory ts1 = ex1.getTokens();
    address[] memory ts = new address[](ts0.length + ts1.length);
    uint k=0;
    for (uint i=0; i<ts0.length; i++) ts[k++]=ts0[i];
    for (uint i=0; i<ts1.length; i++) ts[k++]=ts1[i];
    return ts;
  }

  function swap(token[] memory path, uint x, uint ymin) public returns (uint) {
    require (path.length >= 2);
    path[0].receive(x);
    uint y = x;
    for (uint i=0; i<path.length-1; i++) {
      IExchange ex = _findExchange(path[i],path[i+1]);
      require (ex!=address(0)); // no exchange for hop
      path[i].send(ex, y);
      y = ex.swap(path[i], path[i+1], y);
      to.receive(y);
    }
    require (y>=ymin);
    path[path.length-1].send(msg.sender, y);
    return y;
  }
   
  function _findExchange(token a, token b) internal view returns (IExchange) {
    // Find the exchange supporting the swap which has the highest exchange rate
    if (_supports(ex0,a,b)) {
        if (_supports(ex1,a,b) && ex0.getRate(a,b) < ex1.getRate(a,b)) return ex1;
        return ex0;
    }
    if (_supports(ex1,a,b)) return ex1;
    return IExchange(address(0));
  }

  function _supports(IExchange ex, token a, token b) internal view returns (bool) {
    token[] memory ts = ex.getTokens();
    bool hasA = false;
    bool hasB = false;
    for (uint i=0; i<ts.length; i++) {
      if (ts[i]==a) hasA = true;
      if (ts[i]==b) hasB = true;
    }
    return hasA && hasB;
   }
}
\end{lstlisting}
\end{figure}

The $\contract{SwapRouter}$ contract in~\Cref{lst:swaprouter} composes two exchanges and exposes a single $\txcode{swap}$ method that enables users to perform token exchanges along a user-specified path.
Each underlying exchange implements the $\interface{IExchange}$ interface and may support swaps among multiple token types, as specified by its $\txcode{getTokens}$ method.
A swap is defined by a path $\langle \tokT[0], \tokT[1], \ldots, \tokT[n] \rangle$, indicating a sequence of intermediate token conversions.
The router receives an input amount of $\tokT[0]$ from the caller and executes a sequence of swaps $\tokT[i] \to \tokT[{i+1}]$ for $i = 0,\ldots,n-1$.
For each hop, the router dynamically selects one of the two exchanges that supports both tokens involved in the conversion.
If no such exchange exists for some hop, the transaction reverts.
The router itself does not expose any pricing; rather, it delegates each hop to the selected underlying exchange, approving the required token transfers.
After completing the final hop, the router transfers the resulting tokens to the caller.

Since the balance of $\contract{SwapRouter}$ is zero in any blockchain state, there is nothing to extract, and so its MEV is zero.
Therefore, by~\Cref{th:nonint:sufficient-conditions}\eqref{th:nonint:sufficient-conditions:zero-mev}, it enjoys MEV non-interference against any blockchain state --- regardless of the implementations of the $\interface{IExchange}$ interface.
MEV non-interference also holds against wealthy adversaries, by~\Cref{th:rich-nonint:sufficient-conditions}\eqref{th:rich-nonint:sufficient-conditions:zero-mev}.


\begin{figure}[!htbp]
\begin{lstlisting}[language=solidity
    ,frame=single
    ,caption={A Lending Pool contract.}
    ,label={lst:lp}
    ]
interface ILoaner {
  function deposit(uint x) external;
  function borrow(uint x) external;
  function repay(uint x) external;
}

contract LP is ILoaner {
  token immutable tB; token immutable tC; // borrowed & collateral
  struct Position { uint credit; uint debt; uint lastAccrual; }
  mapping(address => Position) public pos;
  uint constant BT = 150e4;  // borrowing threshold = 150%  
  uint constant LT = 120e4;  // liquidation threshold = 120%
  IPriceOracle oracle; IInterestRate irm;
  
  constructor(IPriceOracle o,IInterestRate i,token tB_,token tC_) {
    oracle=o; irm=i; tB=tB_; tC=tC_; 
    require (tB!=tC && o.getRate(tB,ETH)>0 && o.getRate(tC,ETH)>0);
  }
  function _accrue(address user) internal {
    Position storage p=pos[user];
    if (p.debt==0) { p.lastAccrual=block.number; return; }
    uint interest = irm.getInt(p.debt,p.lastAccrual);
    p.debt += interest; p.lastAccrual=block.number;
  }
  function _isColl(address a) internal view returns (bool) {
    uint vd = pos[a].debt * oracle.getRate(ETH,tB);
    if (vd==0) { return true; } // no debt => collateralized
    uint vc = pos[a].credit * oracle.getRate(ETH,tC);
    return (vc*1e6 >= vd*BT);
  }
  function deposit(uint x) public {
    tC.receive(x);
    pos[msg.sender].credit += x;
  }
  function borrow(uint x) public {
    _accrue(msg.sender);
    tB.send(msg.sender,x);
    pos[msg.sender].debt += x;
    require _isColl(msg.sender);
  }
  function repay(uint x) public {
    _accrue(msg.sender);
    require (pos[msg.sender].debt>=x);
    tB.receive(x);
    pos[msg.sender].debt -= x;
  }
  function liquidate(address borrower) public {
    _accrue(borrower);
    uint vd = (pos[user].debt   * prices.getPrice(tB))/1e6;
    uint vc = (pos[user].credit * prices.getPrice(tC))/1e6;
    require (vc*1e6 < vd*LT);
    tB.receive(pos[borrower].debt);
    tC.send(msg.sender,pos[borrower].credit);
    pos[borrower].debt = 0;
    pos[borrower].credit = 0; // simplified full liquidation
  }
}
\end{lstlisting}
\end{figure}

\subsection{Lending Pool}
\label{sec:defi:lp}

The contract $\contract{LP}$ in~\Cref{lst:lp} implements a simplified lending protocol that allows users to deposit a collateral token $\tokT[c]$ and borrow a distinct token $\tokT[b]$.
The protocol relies on two external contract dependencies:
\begin{itemize}
\item a \emph{price oracle}, used to compute the value of collateral and outstanding debt;
\item an \emph{interest rate model}, used to compute interests on loans (\Cref{lst:interest-rate}).
\end{itemize}

\begin{figure}[tb!]
\begin{lstlisting}[language=solidity
    ,frame=single
    ,caption={A simple Interest Rate Model.}
    ,label={lst:interest-rate}
    ]
interface IInterestRate {
  // Returns interest accrued since `last` block on `debt`
  function getInt(uint debt, uint last) external view returns(uint);
}

// Simple linear interest per block
contract LinearInterestRate is IInterestRate {
  uint public immutable ratePerBlock; // e.g., 1e16 = 1% per block

  constructor(uint r) { ratePerBlock = r; }

  function getInt(uint debt, uint last) public view returns (uint) {
    if (debt==0 || block.number<=last) return 0;
    uint blocks = block.number - last;
    return (debt * ratePerBlock * blocks) / 1e18;
  }
}
\end{lstlisting}
\end{figure}

Our $\contract{LP}$ contract features the following methods:
\begin{itemize}

\item $\txcode{deposit}$ allows anyone to deposit tokens $\tokT[c]$ in the pool. These tokens can be used as collateral to borrow tokens $\tokT[b]$. For simplicity, in~\Cref{lst:lp} we omit the method to redeem deposited $\tokT[c]$ tokens.

\item $\txcode{borrow}$ allows any user with sufficient collateral to withdraw tokens $\tokT[b]$ from the pool. The actual condition is defined in the method \lstinline{_isColl}: more precisely, borrowing is possible when the value of the user's collateral is at least \code{BT} times the value of their debt (after the borrow).
In our example, we choose a borrowing threshold $\code{BT} = 150\%$.
Note that, once tokens are borrowed, the value of the debt will automatically increase over time following the interest rate model.

\item $\txcode{repay}$ allows a borrower to return part of their debt.

\item $\txcode{liquidate}$ allows anyone to repay the debt of an under-collateralized borrower, thus gaining the ownership of the collateral. 
We consider a borrower to be under-collateralized when the value of their collateral is below $\code{LT}$ times the value of their debt.
In our example, we choose a liquidation threshold $\code{LT} = 120\%$.
Note that a borrower can become under-collateralized for various reasons, such as the accrual of interest on the debt and price fluctuations (\eg, the price of the collateral token $\tokT[c]$ drops).   
For simplicity, the $\contract{LP}$ contract in~\Cref{lst:lp} implements a \emph{full liquidation} mechanism, in which the liquidator repays the entire outstanding debt and receives the borrower's entire collateral balance.

\end{itemize}

Lending protocols of this kind are notoriously subject to MEV~\cite{BZ25fc}.
For example, an adversary with the power of causing sufficiently large price fluctuations can cause a borrower's collateral to fall below the threshold, and then  
trigger liquidation by repaying the borrower's debt and receiving its collateral in exchange.
Under suitable conditions on prices, this can cause a loss of value to the $\contract{LP}$ contract. 

We now study the composition of a lending pool contract $\contract{LP}$ with a $\contract{SwapRouter}$ serving as a price oracle.
The router is based on two AMMs as underlying exchanges.   
We will show that this composition does \emph{not} enjoy MEV non-interference.
Assume that the borrowed and collateral tokens are, respectively, $\tokT[b]$ and $\tokT[c]$, with external prices $\price{\tokT[b]} = 2$ and $\price{\tokT[c]} = 1$. 
Let $\Adv = \setenum{\pmvM}$, and 
consider the following state $\sysS$ and an $\contract{LP}$ in $\cstD$ where $\pmvA$ has deposited $200:\tokT[c]$ and borrowed $60:\tokT[b]$: 
\begin{align*}
\sysS & = 
\walpmv{\pmvA}{\waltok{0}{\tokT[c]},\waltok{60}{\tokT[b]}} \mid
\walpmv{\pmvM}{\waltok{100}{\tokT[c]},\waltok{60}{\tokT[b]}} \mid
\walpmv{\contract[1]{AMM}}{\waltok{50}{\tokT[b]},\waltok{100}{\ETH}} \mid
\walpmv{\contract[2]{AMM}}{\waltok{100}{\tokT[c]},\waltok{100}{\ETH}} \mid \walpmv{\contract{SwapRouter}}{\cdots} 
\\
\cstD & = \walpmv{\contract{LP}}{\waltok{40}{\tokT[b]},\waltok{200}{\tokT[c]},\code{oracle}=\contract{SwapRouter},\code{pos}=\setenum{\pmvA \mapsto (200,60,0)},\ldots}
\end{align*}
In the state $\sysS \mid \cstD$, we have that: 
\begin{itemize}
\item the $\ETH$/$\tokT[b]$ rate is $2$, and so the debt value of $\pmvA$ is $\code{vd} = 60 \cdot 2 = 120$;
\item the $\ETH$/$\tokT[c]$ rate is $1$, and so the collateral of $\pmvA$ is $\code{vc} = 200 \cdot 1 = 200$;
\item $\pmvA$ is not liquidatable, since $\code{vc} = 200 \geq \code{vd} \cdot 120\% = 144$. 
\end{itemize}

Therefore, $\pmvA$ is not liquidatable in $\sysS \mid \cstD$: 
this implies that when $\pmvM$ is restricted to call the contracts in $\cstD$, they cannot extract value from the $\contract{LP}$.
Instead, if $\pmvM$ can call any contract in $\sysS$, they can manipulate the prices in order to make $\pmvA$ under-collateralized, and then liquidate $\pmvA$ --- causing the extraction of value from the $\contract{LP}$.
The attack sequence $\TxTS$ is the following:
\begin{enumerate}

\item $\pmvM:\contract[2]{AMM}.\txcode{swap}(\tokT[c],100,\ETH,0)$.
This changes the AMM state to:
\[
\walpmv{\contract[2]{AMM}}{\waltok{200}{\tokT[c]},\waltok{50}{\ETH}}
\]
In this state, the price of $\tokT[c]$ has \emph{decreased} to $0.25$, while the price of $\tokT[b]$ and $\pmvA$'s debt value remain unchanged, so
the collateral value of $\pmvA$ is $200 \cdot 0.25 = 50$.
Therefore, $\pmvA$ is liquidatable, since $\code{vc} = 50 < \code{vd} \cdot 120\% = 144$. 

\item $\pmvM:\contract{LP}.\txcode{liquidate}(\pmvA)$.
The effect of this liquidation is that $\pmvM$ repays the full debt $60:\tokT[b]$ of $\pmvA$ and receives her full collateral $200:\tokT[c]$.

\item $\pmvM:\contract[2]{AMM}.\txcode{swap}(\ETH,50,\tokT[c],0)$.
The adversary restores the AMM state and the exchange rate $\ETH$/$\tokT[c]$ as they were in $\sysS$.

\end{enumerate}
The state after the attack is the following:
\begin{align*}
\sysSi & = 
\walpmv{\pmvA}{\waltok{0}{\tokT[c]},\waltok{60}{\tokT[b]}} \mid
\walpmv{\pmvM}{\waltok{300}{\tokT[c]},\waltok{0}{\tokT[b]}} \mid
\walpmv{\contract[1]{AMM}}{\waltok{50}{\tokT[b]},\waltok{100}{\ETH}} \mid
\walpmv{\contract[2]{AMM}}{\waltok{100}{\tokT[c]},\waltok{100}{\ETH}} \mid
\walpmv{\contract{SwapRouter}}{\cdots} 
\\
\cstDi & = \walpmv{\contract{LP}}{\waltok{100}{\tokT[b]},\waltok{0}{\tokT[c]},\code{pos}=\setenum{\pmvA \mapsto (0,0,0)},\ldots}
\end{align*}
Summing up, the loss of the $\contract{LP}$ due to this attack sequence is:
\begin{align*}
-\gain{\setenum{\contract{LP}}}{\sysS}{\TxTS}
& = \wealth{\setenum{\contract{LP}}}{\sysS} - \wealth{\setenum{\contract{LP}}}{\sysSi} 
\\
& = (40 \cdot \price{\tokT[b]} + 200 \cdot \price{\tokT[c]}) - (100 \cdot \price{\tokT[b]} + 0 \cdot \price{\tokT[c]}) 
\\
& = 280 - 200 > 0
\end{align*}
This implies that the unrestricted MEV is strictly positive, and consequently it is different from the restricted MEV --- which is zero. 
Therefore, $\negnonint{\sysS}{\cstD}$.

MEV non-interference can be recovered by replacing the $\contract{AMM}$ with a price oracle that cannot be manipulated by the adversary.  
Additionally, we must also fix the instance of $\interface{IInterestRate}$ to prevent the adversary from influencing the behaviour of $\contract{LP}$.
For instance, choosing $\contract{PriceOracle}$ from~\Cref{lst:price-oracle} and $\contract{LinearInterestRate}$ from~\Cref{lst:interest-rate} makes the lending pool observationally invariant \wrt adversarial moves, and so~\Cref{th:rich-nonint:sufficient-conditions}\eqref{th:rich-nonint:sufficient-conditions:adv-inert} gives non-interference \wrt wealthy adversaries.
Establishing MEV non-interference \wrt ``poor'' adversaries through~\Cref{th:rich-nonint:sufficient-conditions}\eqref{th:rich-nonint:sufficient-conditions:adv-inert} would additionally require the absence of token flows from the system to the $\contract{LP}$, which is not guaranteed in our composition.



\subsection{Lending Pool arbitrage}
\label{sec:defi:lp-arbitrage}

\begin{figure}[!tb]
\begin{lstlisting}[language=solidity
    ,frame=single
    ,caption={Arbitrage with a Lending Pool.}
    ,label={lst:lp-arbitrage}
    ]
contract LPArbitrage {
  ILoaner   immutable lp;  // borrow t0 (with collateral tc)
  IExchange immutable ex0; // swap t0 -> t1
  IExchange immutable ex1; // swap t1 -> t0

  constructor(ILoaner lp_,IExchange ex0_,IExchange ex1_,token tc_,uint n) {
    lp=lp_; ex0=ex0_; ex1=ex1_; tc=tc_;
    tc.receive(n); tc.send(lp,n); lp.deposit(n); // collateral
  }
  
  function arbitrage(uint x, token t0, token t1) public {
    require (t0!=tc && t1!=tc);
    lp.borrow(x);                  // borrow x:t0
    t0.send(ex0,x);
    uint y1 = ex0.swap(t0,x,t1,0); // sell t0, buy t1
    t1.send(ex1,y1);
    ex1.swap(t1,y1,t0,0);          // sell t1, buy t0
    t0.send(lp,x);
    lp.repay(x);                   // repay x:t0
    uint d=t0.balanceOf(this);
    require (d>0);                 // sender's gain must be positive
    t0.send(msg.sender,d); 
  }
}
\end{lstlisting}
\end{figure}

Arbitrage in DeFi refers to a transaction (or sequence of transactions) that yields a profit by exploiting price discrepancies for the same token pair across different exchanges.
This situation typically arises when two exchanges quote different prices for the same token pair, due for instance to independent liquidity pools, delayed oracle updates, or heterogeneous pricing mechanisms.
To execute the arbitrage, the arbitrageur must first obtain a sufficient amount of liquidity to perform the initial swap.
This liquidity can be sourced in different ways: in this section, we illustrate arbitrage funded through a lending pool, while in the next section we consider the use of a flash loan.

The $\contract{LPArbitrage}$ contract in~\Cref{lst:lp-arbitrage} composes two exchanges and a Lending Pool to provide a zero-risk arbitrage.
To this purpose, the function $\txcode{arbitrage}(x)$ performs atomically the following sequence of transactions:
\begin{enumerate}

\item borrow $x:\tokT[0]$;

\item swap $x:\tokT[0]$ for a certain amount $x_1:\tokT[1]$;

\item swap $x_1:\tokT[1]$ for a certain amount $y:\tokT[0]$;

\item repay the loan $x:\tokT[0]$;

\item transfer the gained units $d = y - x > 0$ of $\tokT[0]$ to the caller.
  
\end{enumerate}

Since the balance of $\contract{LPArbitrage}$ is zero in any blockchain state, there is nothing to extract, and so its MEV is zero.
Therefore, by~\Cref{th:nonint:sufficient-conditions}\eqref{th:nonint:sufficient-conditions:zero-mev}, it enjoys MEV non-interference against any blockchain state --- regardless of the implementations of the $\interface{ILoaner}$ and $\interface{IExchange}$ interfaces.
MEV non-interference also holds against wealthy adversaries, by~\Cref{th:rich-nonint:sufficient-conditions}\eqref{th:rich-nonint:sufficient-conditions:zero-mev}.


\subsection{Flash Loan arbitrage}
\label{sec:defi:flashloan-arbitrage}

\begin{figure}[tb!]
\begin{lstlisting}[language=solidity
  ,frame=single
  ,caption={A Flash Loan contract.}
  ,label={lst:flashloan}
]
interface IFlashLoaner {
  function flashLoan(uint amount, token t) external;
}

interface IFlashBorrower {
  function exec(uint amount) external;
}

contract FlashLP is LP, IFlashLoaner {
  uint public immutable fee;

  constructor(IPriceOracle o,IInterestRate i,token tB_,token tC_,uint f) LP(o,i,tB_,tC_) { fee = f; }

  function flashLoan(uint amt, token t) public {
    require (t==tB);
    uint oldBal = t.balanceOf(address(this));
    t.send(msg.sender,amt);
    IFlashBorrower(msg.sender).exec(amt);
    require (t.balanceOf(address(this)) >= oldBal+fee);
  }
}
\end{lstlisting}
\end{figure}

\begin{figure}[tb!]
\begin{lstlisting}[language=solidity
  ,frame=single
  ,caption={Arbitrage with a Flash Loan.}
  ,label={lst:flashloan-arbitrage}
]
contract FlashLoanArbitrage is IFlashBorrower { 
  IFlashLoaner immutable lp; // to borrow t0 (without collateral)
  IExchange immutable ex0;   // to swap t0 -> t1
  IExchange immutable ex1;   // to swap t1 -> t0
  token immutable t0; token immutable t1;
  
  constructor(IFlashLoaner lp_,IExchange ex0_,IExchange ex1_) { 
    lp=lp_; ex0=ex0_; ex1=ex1_;
  }

  function arbitrage(uint x, token t0_, token t1_) public {
    t0=t0_; t1=t1_; lp.flashLoan(x,t0); 
  }

  function exec(uint amt) public {
    t0.receive(amt);                  // from lp
    t0.send(ex0,x);
    uint x1 = ex0.swap(t0,x,t1,0);    // sell t0, buy t1
    t1.send(ex1,x1);
    ex1.swap(t1,x1,t0,0);             // sell t1, buy t0
    t0.send(msg.sender,amt+lp.fee()); // repay flash loan
    uint d=t0.balanceOf(this)
    require (d>0);                    // sender's gain must be positive
    t0.send(tx.origin,d);
  }  
}
\end{lstlisting}
\end{figure}

The contract $\contract{LParbitrage}$ seen before requires the arbitrageur to have a collateral in $\contract{LP}$ in order to perform the $\txcode{borrow}$ action.
Current lending protocols, such as \eg Aave~\cite{aaveflashloan}, also feature a
$\txcode{flashLoan}$ method that allows users to borrow tokens without providing a collateral, at the condition that the loan (plus a fee) is repaid in the \emph{same} transaction.

We show in~\Cref{lst:flashloan} a lending protocol exposing flash loans, and in~\Cref{lst:flashloan-arbitrage} an arbitrage contract relying on a flash loan.
Note that the dependencies between $\txcode{FlashLoanArbitrage}$ and $\txcode{FlashLoan}$ are circular, since the method $\txcode{arbitrage}$ calls $\txcode{flashLoan}$, which calls back $\txcode{exec}$.
This breaks well-formedness of blockchain states (see~\Cref{sec:blockchain}). However, well-formedness is \emph{not} required by~\Cref{th:nonint:sufficient-conditions,th:rich-nonint:sufficient-conditions}, which establish, in their first item, MEV non-interference whenever the newly deployed contracts have zero MEV.
This is the case for~$\contract{FlashLoanArbitrage}$: regardless of the dependencies, the contract balance is always zero, so there is nothing to extract.



\section{Related work}
\label{sec:related}

A primary term of comparison for our MEV non-interference is the $\epsilon$-composability notion of~\cite{Babel23clockwork}:
\begin{equation*}
  \text{$\cstD$ is $\epsilon$-composable in $\sysS$}
  \quad \text{iff} \quad
  \mev{}{\sysS \mid \cstD}{} \leq (1 + \epsilon) \ \mev{}{\sysS}{}
\end{equation*}

The key conceptual difference between the two notions is that  $\epsilon$-composability is a global property of the \emph{entire} blockchain state, whereas MEV non-interference has a more local nature, as it only requires inspecting the newly deployed contracts and their dependencies.
Indeed, the ``only if'' direction of the implication in \Cref{th:nonint:state-narrowing,th:rich-nonint:state-narrowing} states that, to establish the MEV non-interference of $\cstD$ in state $\cstS \mid \cstC$, we can neglect the part of the state $\cstC$ of the non-dependencies of $\cstD$, under suitable conditions (\eg, $\deps{\cstD}$ is sender-agnostic).

Technically, $\epsilon$-composability and MEV non-interference (in both ``rich'' and ``poor'' versions) are incomparable notions.
To show that MEV non-interference does not imply \mbox{$\epsilon$-composability}, assume that the newly deployed contracts $\cstD$ contains exactly the $\contract{Airdrop}$ contract from~\Cref{ex:nonint:airdrop}, with a wallet of $n:\tokT$.
In \mbox{$\sysS \mid \cstD$}, the adversary can extract
$n:\tokT$ from $\contract{Airdrop}$, and possibly use these tokens
to extract more MEV from~$\sysS$.
Therefore, we have
$\mev{}{\sysS \mid \cstD}{} \geq n \cdot \price{\tokT} + \mev{}{\sysS}{}$,
which implies that, if $n$ is large enough,
the airdrop is \emph{not} $\epsilon$-composable with $\sysS$.
\Cref{ex:nonint:airdrop} shows, instead, that MEV non-interference holds.

\Cref{ex:babel-composability-not-implies-mev-nonintererence} below
shows that the converse implication fails as well.

\begin{example}
  \label{ex:babel-composability-not-implies-mev-nonintererence}
  We craft $0$-composable contracts which violate MEV non-interference.
  For this, we define
  $\cstC$ and $\cstD$ such that both contracts expose
  identical and \emph{mutually exclusive} MEV:
  one can extract MEV from either contract, but not from both. 
  The choice of extracting MEV from $\cstC$ or from $\cstD$ 
  is done by calling a suitable method of $\cstC$.
  Only after the choice is made the MEV is exposed, 
  and the MEV from the other contract is permanently disabled. 
  %
  Since the two MEVs are mutually exclusive, we obtain 
  $\mev{}{\cdots \mid \cstC}{} = \mev{}{\cdots \mid \cstC \mid \cstD}{}$, hence \mbox{$0$-composability}.
  Non-interference fails because extracting MEV
  from $\cstD$ requires calling a method of $\cstC$.

  More formally, consider the contracts in \Cref{lst:babel-notimply-nonint}, and let
  \(
  \cstC = 
  \walpmv{\contract{C1}}{\waltok{1}{\tokT},\code{n}=0}
  \),
  \(
  \sysS = \walu{\pmvM}{0}{\tokT} \mid \cstC
  \), and 
  \(
  \cstD = \walu{\contract{C2}}{1}{\tokT}
  \).
  To study $\epsilon$-composability,
  we compare the (global) MEV of $\sysS$ and $\sysS \mid \cstD$.
  We have that
  $\mev{}{\sysS}{} = 1 \cdot \price{\tokT}$
  by the sequence $\pmvM:\contract{C1}.\txcode{f1}()$,
  and $\mev{}{\sysS \mid \cstD}{} = 1 \cdot \price{\tokT}$
  by the sequence \mbox{$\pmvM:\contract{C1}.\txcode{f2}() \ \pmvM:\contract{C2}.\txcode{g}()$}
  (or, alternatively, by the sequence $\pmvM:\contract{C1}.\txcode{f1}()$, which provides the same MEV).
  Therefore, the two contracts are $0$-composable.
  
  For MEV non-interference, we have that
  $\lmev{}{\sysS \mid \cstD}{\setenum{\contract{C2}}} = 1 \cdot \price{\tokT}$
  by the sequence
  \mbox{$\pmvM:\contract{C1}.\txcode{f2}() \ \pmvM:\contract{C2}.\txcode{g}()$},
  while
  $\lmev{\setenum{\contract{C2}}}{\sysS \mid \cstD}{\setenum{\contract{C2}}} = 0$,
  since the only callable method, \ie $\txcode{g}()$, always fails.
  Therefore, $\negnonint{\sysS}{}{\cstD}$,
  \ie $\cstD$ is not composable with $\sysS$ according to our notion.
  \hfill\qedex
\end{example}

Another difference between MEV non-interference and $\epsilon$-composability is highlighted by the first item of~\Cref{th:nonint:sufficient-conditions,th:rich-nonint:sufficient-conditions}, which state that if the newly deployed contracts have zero MEV, then MEV non-interference holds.
This is not, instead,  a sufficient condition for $\epsilon$-composability: a contract may not be $0$-composable in $\sysS$
when $\Adv$ does not have the tokens needed to extract MEV
in $\sysS$, but becomes able to do so
after exchanging tokens between $\sysS$ and $\cstD$, provided that
this preserves $\cstD$'s wealth
(see~\Cref{cex:th:nonint:sufficient-conditions:zero-mev:babel}).


\begin{figure}[t]
  \begin{lstlisting}[language=solidity
  % ,morekeywords={f1,f2,f3,g},classoffset=4,morekeywords={a,b,A,Oracle},keywordstyle=\pmvColor,classoffset=5,morekeywords={t,T,T2},keywordstyle=\tokColor,classoffset=6,morekeywords={C1,C2,C3},keywordstyle=\cmvColor,frame=single,
  ,caption={Babel \emph{et al.}' composability does not imply MEV non-interference.}
  ,label={lst:babel-notimply-nonint}
  ]
contract C1 {
  uint public n;
  constructor() { T.receive(1); n=0; }
  function f1() { require (n==0); n=1; T.send(msg.sender,1); }
  function f2() { require (n==0); n=2; }
}
contract C2 {
  constructor() { T.receive(1); }    
  function g() { require C1.n()==2; T.send(msg.sender,1); }
}
\end{lstlisting}
\end{figure}

\begin{example}
  \label{cex:th:rich-lmev:future-not-affect-past:non-rich}
  \label{cex:th:nonint:sufficient-conditions:zero-mev:babel}
  Using the states $\sysS$ and $\cstD$  
  in~\Cref{cex:lem:mev:state-narrowing}, 
  we show that 
  a contract with zero MEV may not be $0$-composable with the context.
  Note that any action on $\contract{Exchange2}$ does not alter its wealth,
  since the total number of tokens held in that contract is preserved (because their internal prices given by $\contract{PriceOracle2}$ are equal), and they have the same (external) price.
  Therefore, $\lmev{}{\sysS \mid \cstD}{\setenum{\contract{Exchange2}}} = 0$,
  and so condition~\ref{th:nonint:sufficient-conditions:zero-mev} of
  \Cref{th:nonint:sufficient-conditions} ensures that $\sysS$ is
  MEV non-interfering with $\cstD$.
  
  Instead, $\cstD$ is \emph{not} $0$-composable with $\sysS$:
  indeed, $\mev{}{\sysS}{} = 0 \neq \mev{}{\sysS \mid \cstD}{} = 1$,
  since $\pmvM$ can first swap her $1:\tokT[2]$ for $1:\tokT[1]$
  using $\contract{Exchange2}$, and then swap $1:\tokT[1]$ for $2:\tokT[2]$
  using $\contract{Exchange1}$, with an overall gain of $1:\tokT[2]$.
  This is coherent with the different interpretations of composability:
  for~\cite{Babel23clockwork}, adding a new contract must preserve the global MEV,
  while our notion requires to preserve the MEV of the victim contracts.
  \hfill\qedex
\end{example}



The work~\cite{Guesmi24dlt} applies classical program non-interference notions to the setting of smart contracts. 
There, smart contracts are specified in a general-purpose concurrent imperative language.
Unlike in our approach, where token holdings and exchanges are modelled implicitly in the contract semantics, in~\cite{Guesmi24dlt}
they must be explicitly programmed as assignments to suitable variables.
Therefore, measuring the MEV of a contract is reduced to estimating the change of such variables. 
Similarly, establishing non-interference amounts to detecting whether some high-level variables (\eg, the rate of an exchange contract) can affect low-level variables (\eg, the wallet of a contract using that exchange). 
Therefore, this approach requires the programmer to identify the high-level and the low-level variables, and to manually downgrade some variables from high to low to explicitly admit some flows (\eg, the oracle owner can update the exchange rate). 
Our approach, instead, does not require the programmer to annotate the contract code: roughly, the ``low'' variable is the MEV extractable from the victim contracts, and the ``high'' variables are the methods callable from the adversary. 

This observation reveals a key difference between our approach and~\cite{Guesmi24dlt}: MEV non-interference may hold also when the adversary can affect the victim contracts' wallet, provided that their loss is preserved, while in~\cite{Guesmi24dlt} any flow from the variables touched by the adversary to the wallet is considered harmful. 
For example, consider a variant of the $\contract{Airdrop}$ contract that allows its users to withdraw either a token $\tokT[1]$ or another token $\tokT[2]$, where both tokens have the same external price. This $\contract{Airdrop}$ contract can only be used once, and internally chooses which token to send by querying another contract $\contract{PriceOracle}$. 
According to our notion of MEV non-interference, there is no interference between $\contract{PriceOracle}$ and the $\contract{Airdrop}$, since the extracted token will have the same price, hence generating the same MEV.
Instead, the non-interference notion in~\cite{Guesmi24dlt} can observe that 
$\contract{PriceOracle}$ can interfere with $\contract{Airdrop}$, affecting its behavior.

More than with MEV non-interference, the notion of non-interference in~\cite{Guesmi24dlt} seems related to our observational invariance
(\Cref{sec:mev:observable-invariance}).
Informally, this notion requires that no action performed by the adversary on $\sysS$ can influence what the contracts in $\cstD$ are able to observe from $\sysS$. 
Conceptually, this resembles a classic non-interference property, where high-level variables (\ie, the methods of the contracts in $\sysS$) are prevented from affecting low-level variables (\ie, the observations that $\cstD$ can make on $\sysS$).  


\newcommand{\qnonint}[2]{\ifempty{#1}{\mathcal{I}}{\mathcal{I}({#1} \rightsquigarrow {#2})}}

The work~\cite{PB25wtsc} introduces a \emph{quantitative} versions MEV non-interference, which enables to measure the effect of MEV attacks on the victim contract.
More specifically, given a set of victim contracts $\cstD$ and a system $\sysS$,
the value $\qnonint{\sysS}{\cstD} \in [0,1]$ measures the increase of loss (proportionally to $\cstD$'s wealth) that adversaries can inflict to $\cstD$ by manipulating the context $\sysS$.
In particular, a value $\qnonint{\sysS}{\cstD} = 0$ means that no such increase is possible: this is coherent with our notion, \ie $\qnonint{\sysS}{\cstD} = 0$ holds iff $\nonint{\sysS}{\cstD}$.
Instead, a larger (positive) value of $\qnonint{\sysS}{\cstD}$ means that a larger fraction of $\cstD$'s wealth can be extracted by the adversary.
The quantitative nature of this notion makes it possible to have more fine-grained results about the effect of adding new contracts to the system.
In particular, $\qnonint{\sysS}{\cstD}$ increases when $\sysS$ is extended with contracts that are not in the dependencies of $\cstD$, \ie $\qnonint{\sysS}{\cstD} \leq \qnonint{\sysS \mid \cstC}{\cstD}$.
By contrast, our state narrowing results (\Cref{th:nonint:state-narrowing,th:rich-nonint:state-narrowing}) can only elaborate on conditions about $\cstC$ that preserve MEV non-interference.
A potential advantage of a quantitative notion \emph{vs.} a qualitative one is to understand the effect of different system parameters on the MEV increase.
The work \cite{PB25wtsc} exemplifies this by measuring the degree of interference in compositions akin to our $\contract{AMM}$/$\contract{Option}$ and $\contract{AMM}$/$\contract{LP}$ use cases in~\Cref{sec:defi}.
However, establishing non-trivial interference bounds in complex contract compositions seems quite more challenging than establishing our MEV non-interference.


A preliminary version of this work appeared in~\cite{BMZ23defi}. The present article substantially extends that work, both in its theoretical development and in its analysis of DeFi compositions.
On the theoretical side, we refine the conditions under which state narrowings/widenings preserve MEV and MEV non-interference. 
In particular, while~\cite{BMZ23defi} established local reasoning principles only for unbounded-wealth adversaries, the current work extends these results to bounded-wealth adversaries. 
This yields preservation theorems for both $\lmev{}{}{}$ and $\nonint{}{}{}$ under suitable conditions (\Cref{th:mev:frontrunning-resistance,th:nonint:state-narrowing}, respectively), thereby providing a unified treatment of the two adversarial models.
On the practical side, the current version broadens the analysis of DeFi compositions. In~\cite{BMZ23defi}, the swap proxy, swap router, and arbitrage protocols were studied only in conjunction with AMM-based exchanges. Here, we generalise the analysis to arbitrary exchanges implementing the required interfaces (\Cref{tab:defi-compositions}).
In addition, the present version streamlines and generalises several technical developments, and provides complete proofs for all results.

\section{Conclusions and discussion}
\label{sec:conclusions}

We have proposed MEV non-interference, a new security notion for DeFi composition which ensures that adversaries cannot inflict economic harm on compound contracts by exploiting their dependencies.
In particular, we have studied sufficient conditions for MEV non-interference (\Cref{th:nonint:sufficient-conditions,th:rich-nonint:sufficient-conditions}), and rules to enable its modular verification (\Cref{th:nonint:state-narrowing,th:rich-nonint:state-narrowing}).
We have shown that these rules allow to correctly classify
the composability of several common DeFi protocols (\Cref{tab:defi-compositions}).

\paragraph{Alternative definitions}

A relevant question is whether simpler notions of composability
would achieve the same effect as our~\Cref{def:non-interference}.
For instance, one might be tempted to regard two contracts
$\cmvC$ and $\cmvCi$ composable whenever the (global) MEV of their
composition is equal to the sum of the two individual MEVs,
thus obtaining a property which could be written along the line of
\begin{equation}
\label{eq:wrong-composability:1}
\lmev{}{\walpmv{\contract{C}}{\cdot} \mid \walpmv{\contract{C}'}{\cdot}}{}
= \lmev{}{\walpmv{\contract{C}}{\cdot}}{}
+ \lmev{}{\walpmv{\contract{C}'}{\cdot}}{}
\end{equation}
While temptingly simple, this notion has several issues.
First, it does not consider the dependencies of the two contracts,
which must be part of the blockchain state.
\Eg, when both contracts depend on a third contract $\cmvD$, we could
amend~\eqref{eq:wrong-composability:1} as:
\begin{equation}
\label{eq:wrong-composability:2}
\lmev{}{\walpmv{\contract{D}}{\cdot} \mid \walpmv{\contract{C}}{\cdot} \mid \walpmv{\contract{C}'}{\cdot}}{}
= \lmev{}{\walpmv{\contract{D}}{\cdot} \mid \walpmv{\contract{C}}{\cdot}}{}
+ \lmev{}{\walpmv{\contract{D}}{\cdot} \mid \walpmv{\contract{C}'}{\cdot}}{}
\end{equation}
However, this equation is almost always false, since in the RHS 
the MEV extractable from $\cmvD$ is counted \emph{twice}.
Even when $\cmvD$ has no MEV, it can still act as a shared state between 
$\cmvC$ and $\cmvCi$, \eg making it possible to extract MEV from only one of them (but not both). 
This, again, can be used to falsify~\eqref{eq:wrong-composability:2}. 
%
Furthermore, \eqref{eq:wrong-composability:2} does
not mention the adversary wallet.
Using the same wallet in each $\lmev{}{\cdot}{}$ would duplicate the adversary wealth in the RHS, leading to similar double-counting issues as those discussed previously for~$\cmvD$.

Another alternative definition of MEV could be the following:
\begin{equation}
    \label{eq:lmev:alt}
    \lmev{\CmvD}{\sysS}{\CmvC}
    = \max \setcomp
    {-\gain{\CmvC}{\sysS}{\TxTS} - \gain{\PmvU \setminus \Adv}{\sysS}{\TxTS}}
    {\TxTS \in \mall{\CmvD}{\Adv}^*}
\end{equation}

Compared to~\eqref{eq:lmev}, the MEV in~\eqref{eq:lmev:alt} subtracts, from the loss of the victim contracts $\CmvC$, the gain of the addresses $\PmvU \setminus \Adv$ not controlled by the adversary.
The intuition is that if a transaction fired by the adversary has the effect of moving tokens from a victim contract to a ``legit'' user (\eg, one that should have received the tokens anyway), then it should not count as MEV.
The problem with this alternative definition is that ``legit'' addresses seem hard to characterise: for example, if the adversary manages to move tokens from the victim contracts to a sink contract from where no tokens can be extracted, or to a randomly chosen address, then it would be make more sense to count these tokens as a loss, unlike~\eqref{eq:lmev:alt}.  


\paragraph{Pitfalls}

MEV non-interference enjoys some interesting structural properties: notably, \Cref{th:nonint:state-narrowing,th:rich-nonint:state-narrowing} show that, under certain conditions, non-interference is preserved when extending or restricting the state. 
At the same time, a number of seemingly plausible structural properties do not hold in general.
\Cref{tab:pitfalls} summarises several such non-properties of $\richnonint{}{}$, which --- by~\Cref{lem:rich-nonint-nonint} --- also fail to hold for $\nonint{}{}$.

The first two sections show that, in general, adding or removing contracts from $\cstD$ does not preserve MEV non-interference.
While it is intuitively obvious that adding contracts breaks non-interference, the fact that removing them also breaks it is somehow surprising. 
The last section of \Cref{tab:pitfalls} is even more surprising:
proving that $\cstC$ does not interfere with both $\cstD[1]$ and $\cstD[2]$ is not enough to guarantee that $\cstC$ does not interfere with the composition $\cstD[1] \mid \cstD[2]$.
Since the counterexamples needed to refute these conjectures are quite technical, we relegate them to~\ref{sec:proofs:cex}.

\begin{table}[t]
  \caption{Some non-properties of $\richnonintrel$ (the hypothesis does not imply the thesis).}
  \label{tab:pitfalls}
  \setlength{\tabcolsep}{7pt} 
  \renewcommand{\arraystretch}{1} 
  \centering
  \begin{tabular}{|c|c|c|}
    \hline
    \textbf{Hypothesis} & \textbf{Thesis} & \textbf{Counterexample} 
    \\
    \hline
    \multirow{2}{*}{$\richnonint{\cstC}{\cstD}$}
    & $\richnonint{\cstC}{\cstD \mid \cstDi}$
    & \Cref{ex:rich-nonint:mid-R}
    \\
    & $\richnonint{\cstC}{\cstDi \mid \cstD}$
    & \Cref{ex:rich-nonint:mid-R}
    \\
    \hline
    $\richnonint{\cstC}{\cstD \mid \cstDi}$
    & \multirow{2}{*}{$\richnonint{\cstC}{\cstD}$}
    & \Cref{ex:rich-nonint:mid-L}
    \\
    $\richnonint{\cstC}{\cstDi \mid \cstD}$
    &
    & \Cref{ex:rich-nonint:mid-L}
    \\
    \hline
    $\richnonint{\cstC}{\cstD[1]}$ $\,\land\,$
    $\richnonint{\cstC}{\cstD[2]}$
    & \multirow{2}{*}{$\richnonint{\cstC}{\cstD[1] \mid \cstD[2]}$}
    & \Cref{ex:rich-nonint:union}
    \\
    $\richnonint{\cstC}{\cstD[1]}$ $\,\land\,$
    $\richnonint{\cstC \mid \cstD[1]}{\cstD[2]}$
    & 
    & \Cref{ex:rich-nonint:union}
    \\    
    \hline
  \end{tabular}
\end{table}



\paragraph{Adversary model}

Our definition of MEV is parameterised by a set of addresses $\Adv$ controlled by the adversary.
In principle, different choices of $\Adv$ may yield different MEV values and, consequently, different non-interference classifications for the same contract composition.
In practice, however, any sufficiently large set of addresses that does not already occur in the contracts under analysis (\ie, neither in their  storage nor as hard-coded constants in their code) would yield the same MEV.
It would therefore be possible to reformulate our notions to as to eliminate any dependence on the particular choice for $\Adv$, albeit at the cost of an increased complexity of the definitions and statements.

A similar approach has been adopted in~\cite{BZ25fc} in the context of the classical interpretation of MEV as maximal adversarial gain.
There, the authors introduce the notion of \emph{universal MEV}, defined as the maximum gain that can be extracted over all cofinite sets of addresses $\Adv$, after suitably redistributing tokens among the wallets controlled by $\Adv$.
Besides abstracting from the actual identity of the adversary (\ie, from the choice of the set $\Adv$), this notion also abstracts from their wealth, in this second respect being conceptually similar to our $\rlmev{}{}{}$ (\Cref{def:rich-mev}). 
The two notions differ, however, in the assumptions they make about the adversary's available tokens.
In universal MEV, the adversary's capital is ultimately bounded by the total amount of tokens already occurring in the blockchain state (since the definition just redistributes that to $\Adv$). 
In contrast, our definition of $\rlmev{}{}{}$ assumes that the adversary has access to whatever capital is required to execute the attack.


Another key difference between our notion of MEV and that of~\cite{BZ25fc} is that ours implicitly assumes an empty mempool, while the definitions in~\cite{BZ25fc} are parameterised by the mempool.
In certain scenarios, access to pending transactions in the mempool increases the adversary's opportunities for MEV extraction. 
A notable example is the \emph{sandwich attacks} on AMMs~\cite{Zhou21high}, where the adversary observes a user's trade in the mempool and strategically places transactions before and after it, causing the user's trade to execute under less favourable conditions, while generating a profit for the adversary.
Our framework abstracts away from the mempool because our primary goal is to to study secure composability, namely the additional MEV opportunities introduced by newly deployed contracts rather than those arising from pending transactions. 
While incorporating the mempool into our definitions would be consistent with the principle of analysing security notions under increasingly stronger adversaries (as we did when considering wealthy adversaries), doing so would considerably increase the complexity of an already technically demanding theory, while providing limited additional insight into the composability issues that are the focus of this work.


\paragraph{Token prices}

Our definition of the price of a token $\tokT$ as $\price{\tokT}$ treats token prices as exogenous parameters, in that such prices do not depend on the blockchain state.
As a consequence, our composability notions are parametric in the chosen price assignment, and a composition classified as MEV non-interfering under one assignment need not remain so under another.

One possible way to eliminate this parameter would be to internalize prices in the blockchain state, for instance by designating certain contracts as price oracles (\eg, contracts exposing the interface $\contract{IPriceOracle}$). 
This would make prices endogenous, but it would also make the resulting security notions depend more heavily on the particular oracle contracts included in the state (an on their transitions). 
In this sense, the dependency on external prices would not really disappear, but would instead be shifted to a dependency on the oracle mechanisms used to determine them.

A more robust alternative is to rely on stricter, price-independent notions of secure composability that avoid collapsing heterogeneous token movements into a single MEV scalar. 
Our observational invariance provides one such notion: it is independent of token prices by definition~(\Cref{def:observational-invariance}), and it implies MEV non-interference under the sufficient conditions established in~\Cref{th:nonint:sufficient-conditions,th:rich-nonint:sufficient-conditions}. 
Thus, observational invariance can be seen as a conservative, price-agnostic criterion for secure composability, particularly useful when no canonical or stable price assignment is available.

\paragraph{Limitations}

Some of our results rely on the assumption that dependencies are backward-closed, \ie a contract cannot invoke a contract deployed after itself.
For example, the monotonicity result of~\Cref{lem:mev:basic}\eqref{lem:mev:basic:monotonicity}, which plays a key role in establishing the front-running resistance of $\richnonint{}{}$, fails without this assumption (see~\Cref{cex:mev:basic:monotonicity}).
Backward-closure is a simplification with respect to Ethereum, where  dependencies may be dynamic and even cyclic.
In particular, contracts can interact through reentrant calls, enabling attacks that exploit existing contracts to extract additional MEV, as exemplified by the DAO attack~\cite{DAO}.
A possible research direction is to relax these assumptions by localizing them to the contracts under analysis.
Rather than requiring backward-closure throughout the entire blockchain state, one could investigate whether our results continue to hold under assumptions that only constrain the newly deployed contracts $\cstD$ and their dependency graph. 
This would leave the surrounding state $\sysS$ unconstrained, allowing arbitrary pre-existing dependencies and interactions among already deployed contracts. 
With this relaxation, the rest of the blockchain state would not be constrained, making our results applicable to a wider class of contracts.

Another limitation is that our notion of MEV measures the loss of a contract as the value of the tokens that adversaries can remove from it.
In practice, adversaries could harm contracts also by \emph{freezing}
tokens without actually extracting them from the contract, as in the infamous Parity Wallet attack~\cite{parity17nov}.
Refining local MEV to take these attacks into account could be done
by adapting the notion of liquidity~\cite{BLMZ22lmcs}.

\paragraph*{Acknowledgments}

This work was partially supported by project SERICS (PE00000014)
under the MUR National Recovery and Resilience Plan funded by the
European Union -- NextGenerationEU, and by PRIN 2022 PNRR project DeLiCE (F53D23009130001).

\bibliography{main}

\clearpage
\appendix
\section{Supplementary material for~\Cref{sec:mev}} 
\label{sec:proofs:mev}

In this and the following appendices we provide detailed proofs for all your statements.
These proofs are presented in the order in which the statements appear in the paper, even though this order does not always reflect their logical dependencies.
To clarify the relationship among our statements, \Cref{fig:graph-proofs} displays a graph of the dependencies: an arrow $a \rightarrow b$ means that the proof of statement $a$ depends on statement $b$.
Note that this graph is acyclic.

\begin{figure}[t]
  \centering
  \scalebox{0.675}{
    \Large
    \begin{tikzpicture}[>=latex,line join=bevel,]
  \pgfsetlinewidth{1bp}
\begin{scope}
  \pgfsetstrokecolor{black}
  \definecolor{strokecol}{rgb}{1.0,1.0,1.0};
  \pgfsetstrokecolor{strokecol}
  \definecolor{fillcol}{rgb}{1.0,1.0,1.0};
  \pgfsetfillcolor{fillcol}
  \filldraw (0.0bp,0.0bp) -- (0.0bp,324.0bp) -- (470.51bp,324.0bp) -- (470.51bp,0.0bp) -- cycle;
\end{scope}
  \pgfsetcolor{black}
  \draw [->] (321.6bp,71.697bp) .. controls (328.12bp,63.389bp) and (336.07bp,53.277bp)  .. (349.56bp,36.104bp);
  \draw [->] (298.48bp,143.7bp) .. controls (288.13bp,134.97bp) and (275.41bp,124.24bp)  .. (256.28bp,108.1bp);
  \draw [->] (316.28bp,143.7bp) .. controls (315.07bp,135.98bp) and (313.61bp,126.71bp)  .. (310.69bp,108.1bp);
  \draw [->] (404.4bp,71.697bp) .. controls (397.88bp,63.389bp) and (389.93bp,53.277bp)  .. (376.44bp,36.104bp);
  \draw [->] (181.18bp,78.753bp) .. controls (187.35bp,76.483bp) and (193.89bp,74.122bp)  .. (200.0bp,72.0bp) .. controls (243.05bp,57.052bp) and (292.93bp,41.02bp)  .. (335.78bp,27.488bp);
  \draw [->] (154.0bp,71.697bp) .. controls (154.0bp,63.983bp) and (154.0bp,54.712bp)  .. (154.0bp,36.104bp);
  \draw [->] (112.6bp,143.7bp) .. controls (119.12bp,135.39bp) and (127.07bp,125.28bp)  .. (140.56bp,108.1bp);
  \draw [->] (215.67bp,143.7bp) .. controls (218.71bp,135.81bp) and (222.38bp,126.3bp)  .. (229.4bp,108.1bp);
  \draw [->] (195.4bp,143.7bp) .. controls (188.88bp,135.39bp) and (180.93bp,125.28bp)  .. (167.44bp,108.1bp);
  \draw [->] (54.245bp,145.98bp) .. controls (72.827bp,135.74bp) and (97.638bp,122.07bp)  .. (126.8bp,105.99bp);
  \draw [->] (154.0bp,215.87bp) .. controls (154.0bp,191.67bp) and (154.0bp,147.21bp)  .. (154.0bp,108.19bp);
  \draw [->] (140.4bp,215.7bp) .. controls (133.88bp,207.39bp) and (125.93bp,197.28bp)  .. (112.44bp,180.1bp);
  \draw [->] (346.24bp,229.72bp) .. controls (375.88bp,224.6bp) and (422.22bp,211.74bp)  .. (446.0bp,180.0bp) .. controls (474.86bp,141.48bp) and (477.99bp,113.73bp)  .. (454.0bp,72.0bp) .. controls (442.33bp,51.704bp) and (419.57bp,38.451bp)  .. (390.29bp,26.659bp);
  \draw [->] (306.5bp,215.65bp) .. controls (299.34bp,205.45bp) and (290.34bp,192.2bp)  .. (283.0bp,180.0bp) .. controls (270.62bp,159.4bp) and (257.98bp,135.25bp)  .. (244.29bp,108.04bp);
  \draw [->] (319.0bp,215.7bp) .. controls (319.0bp,207.98bp) and (319.0bp,198.71bp)  .. (319.0bp,180.1bp);
  \draw [->] (399.41bp,143.89bp) .. controls (393.54bp,133.77bp) and (386.54bp,120.52bp)  .. (382.0bp,108.0bp) .. controls (374.69bp,87.858bp) and (369.79bp,64.079bp)  .. (365.22bp,36.119bp);
  \draw [->] (411.98bp,143.7bp) .. controls (412.86bp,135.98bp) and (413.92bp,126.71bp)  .. (416.05bp,108.1bp);
  \draw [->] (122.74bp,287.87bp) .. controls (118.14bp,263.67bp) and (109.69bp,219.21bp)  .. (102.27bp,180.19bp);
  \draw [->] (132.92bp,287.7bp) .. controls (136.07bp,279.81bp) and (139.88bp,270.3bp)  .. (147.16bp,252.1bp);
  \draw [->] (27.0bp,215.7bp) .. controls (27.0bp,207.98bp) and (27.0bp,198.71bp)  .. (27.0bp,180.1bp);
\begin{scope}
  \definecolor{strokecol}{rgb}{0.0,0.0,0.0};
  \pgfsetstrokecolor{strokecol}
  \draw (363.0bp,18.0bp) node {\Cref{lem:mev:basic}};
\end{scope}
\begin{scope}
  \definecolor{strokecol}{rgb}{0.0,0.0,0.0};
  \pgfsetstrokecolor{strokecol}
  \draw (236.0bp,90.0bp) node {\Cref{lem:mev:state-narrowing}};
\end{scope}
\begin{scope}
  \definecolor{strokecol}{rgb}{0.0,0.0,0.0};
  \pgfsetstrokecolor{strokecol}
  \draw (308.0bp,90.0bp) node {\Cref{th:mev:callable-narrowing:deps}};
\end{scope}
\begin{scope}
  \definecolor{strokecol}{rgb}{0.0,0.0,0.0};
  \pgfsetstrokecolor{strokecol}
  \draw (319.0bp,162.0bp) node {\Cref{th:mev:frontrunning-resistance}};
\end{scope}
\begin{scope}
  \definecolor{strokecol}{rgb}{0.0,0.0,0.0};
  \pgfsetstrokecolor{strokecol}
  \draw (418.0bp,90.0bp) node {\Cref{th:mev:callable-narrowing:observational-invariance}};
\end{scope}
\begin{scope}
  \definecolor{strokecol}{rgb}{0.0,0.0,0.0};
  \pgfsetstrokecolor{strokecol}
  \draw (154.0bp,18.0bp) node {\Cref{lem:mev:wallet}};
\end{scope}
\begin{scope}
  \definecolor{strokecol}{rgb}{0.0,0.0,0.0};
  \pgfsetstrokecolor{strokecol}
  \draw (154.0bp,90.0bp) node {\Cref{lem:mev:stability}};
\end{scope}
\begin{scope}
  \definecolor{strokecol}{rgb}{0.0,0.0,0.0};
  \pgfsetstrokecolor{strokecol}
  \draw (99.0bp,162.0bp) node {\Cref{lem:rich-mev:basic}};
\end{scope}
\begin{scope}
  \definecolor{strokecol}{rgb}{0.0,0.0,0.0};
  \pgfsetstrokecolor{strokecol}
  \draw (209.0bp,162.0bp) node {\Cref{lem:rich-mev:state-narrowing}};
\end{scope}
\begin{scope}
  \definecolor{strokecol}{rgb}{0.0,0.0,0.0};
  \pgfsetstrokecolor{strokecol}
  \draw (27.0bp,162.0bp) node {\Cref{th:rich-mev:callable-narrowing:deps}};
\end{scope}
\begin{scope}
  \definecolor{strokecol}{rgb}{0.0,0.0,0.0};
  \pgfsetstrokecolor{strokecol}
  \draw (154.0bp,234.0bp) node {\Cref{th:rich-mev:callable-narrowing:observational-invariance}};
\end{scope}
\begin{scope}
  \definecolor{strokecol}{rgb}{0.0,0.0,0.0};
  \pgfsetstrokecolor{strokecol}
  \draw (319.0bp,234.0bp) node {\Cref{th:nonint:state-narrowing}};
\end{scope}
\begin{scope}
  \definecolor{strokecol}{rgb}{0.0,0.0,0.0};
  \pgfsetstrokecolor{strokecol}
  \draw (410.0bp,162.0bp) node {\Cref{th:nonint:sufficient-conditions}};
\end{scope}
\begin{scope}
  \definecolor{strokecol}{rgb}{0.0,0.0,0.0};
  \pgfsetstrokecolor{strokecol}
  \draw (126.0bp,306.0bp) node {\Cref{th:rich-nonint:sufficient-conditions}};
\end{scope}
\begin{scope}
  \definecolor{strokecol}{rgb}{0.0,0.0,0.0};
  \pgfsetstrokecolor{strokecol}
  \draw (27.0bp,234.0bp) node {\Cref{th:rich-nonint:state-narrowing}};
\end{scope}
\end{tikzpicture}
  }
  \caption{Dependencies among the statements.}
  \label{fig:graph-proofs}
\end{figure}
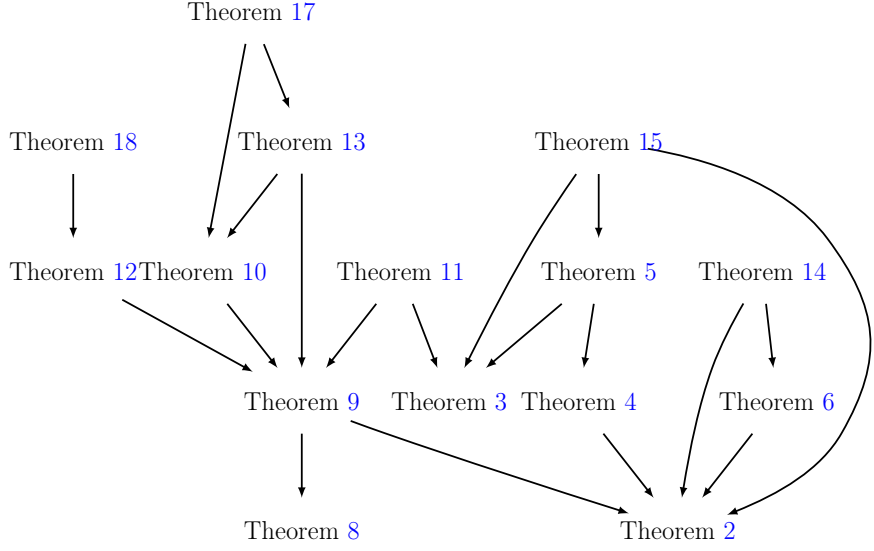

\begin{proofof}{lem:mev:mev-vs-babel}
We have to prove that
$\lmev{\CmvU}{\sysS}{\CmvU} \geq \mev{}{\sysS}{}$.
Let $\TxTS \in \mall{}{\Adv}^*$.
%
Since transactions cannot create or destroy wealth \emph{globally}, \ie on the entire universe of accounts $\AddrU$, we have that:
\begin{equation}
\label{eq:mev:mev-vs-babel}
	0 
    =
    \gain{\AddrU}{\sysS}{\TxTS}
    =  
    \gain{\CmvU}{\sysS}{\TxTS} + \gain{\PmvU \setminus \Adv}{\sysS}{\TxTS}+ \gain{\Adv}{\sysS}{\TxTS}
\end{equation}
Note that $\gain{\PmvU \setminus \Adv }{\sysS}{\TxTS} \geq 0$, 
since a transaction crafted by $\Adv$ cannot steal tokens from user accounts. 
Therefore, by~\eqref{eq:mev:mev-vs-babel} we obtain 
$\gain{\Adv}{\sysS}{\TxTS} \leq - \gain{\CmvU}{\sysS}{\TxTS}$.
Noting that $\mall{\CmvU}{\Adv} = \mall{}{\Adv}$, we conclude:
\begin{align*}
	\lmev{\CmvU}{\sysS}{\CmvU} &= \max \setcomp{ -\gain{\CmvU}{\sysS}{\TxTS} }{ \TxTS \in \mall{\CmvU}{\Adv}^* } 
	\\
	& \geq \max\setcomp{\gain{\Adv}{\sysS}{\TxTS}}{ \TxTS \in \mall{}{\Adv}^* } = \mev{}{\sysS}{}
    \tag*{\qed}
\end{align*}
\end{proofof}


\begin{proofof}{lem:mev:basic}
	For~\Cref{lem:mev:basic:zero} we need to prove the following equalities.
	\begin{itemize}
		\item[$a.$] $\lmev{\CmvD}{\sysS}{\emptyset} = 0$. Since there are no victims, the identity $\gain{\emptyset}{\sysS}{\TxTS}=0$ holds for any $\TxTS$.
		\item[$b.$] $\lmev{\emptyset}{\sysS}{\CmvC}=0$. 
		Since $\mall{\emptyset}{\Adv}$ is empty, the only $\TxTS\in \mall{\emptyset}{\Adv}^*$ is the empty sequence $\tx{\emptyseq}$, which preserve the state. 
        Then, $\gain{\CmvC}{\sysS}{\tx{\emptyseq}} = \wealth{\CmvC}{\sysS} - \wealth{\CmvC}{\sysS} = 0$.
    \end{itemize}

    \medskip\noindent
	For~\Cref{lem:mev:basic:leq-wealth}, we prove two inequalities:
	\begin{itemize}
		\item[$a.$] Let $\tx{\emptyseq} \in \mall{\CmvD}{\Adv}^*$ be the empty sequence of transactions. 
        We have that: 
		\[
			0 
            \; = \;  
            \gain{\CmvC}{\sysS}{\tx{\emptyseq}} \leq \max \setcomp  { -\gain{\CmvC}{\sysS}{\TxTS}} {\TxTS \in \mall{\CmvD}{\Adv}^*} 
            \; = \;
            \lmev{\CmvD}{\sysS}{\CmvC}
		\]

		\item[$b.$]
		The amount of tokens that $\CmvC$ can lose is bounded by the amount of tokens contained in it, and so $\gain{\CmvC}{\sysS}{\TxTS} \geq - \wealth{\CmvC}{\sysS}$. 
        This implies $-\gain{\CmvC}{\sysS}{\TxTS} \leq  \wealth{\CmvC}{\sysS}$, 
        and since this holds for all $\TxTS$ we have the thesis.
	\end{itemize}

    \medskip\noindent
	For~\Cref{lem:mev:basic:garbage} we first prove
	$\lmev{\CmvD}{\sysS}{\CmvC} = \lmev{\CmvD}{\sysS}{\CmvC \cap \cmvOfcst{\sysS}}$.
	For all $\sysS$ and $\TxTS$:
	\[
		\gain{\CmvC}{\sysS}{\TxTS} 
        \; = \;
        \sum_{\cmvC \in \CmvC} \gain{\cmvC}{\sysS}{\TxTS} 
        \; = 
        \sum_{\cmvC \in \CmvC \cap \cmvOfcst{\cstC}} \gain{\cmvC}{\sysS}{\TxTS} 
        \; = \;
        \gain{\CmvC \cap \cmvOfcst{\cstC}}{\sysS}{\TxTS}
	\]
    where the second equality holds because any transaction in $\TxTS$ targeting a contract $\cmvD\notin \cmvOfcst{\sysS}$ is invalid, causing $\gain{\cmvD}{\sysS}{\TxTS}=0$. 
	Since this holds for any $\TxTS$, the first equality of~\Cref{lem:mev:basic:garbage} is proven.
    We now prove $\lmev{\CmvD}{\sysS}{\CmvC} = \lmev{\CmvD \cap \cmvOfcst{\sysS}}{\sysS}{\CmvC}$. 
    To do so, just note that a transaction in $\mall{\CmvD}{\Adv}$ that calls a contract not appearing in $\sysS$ is invalid, and therefore will have no effect on the loss of $\CmvC$.

    \medskip\noindent
	For~\Cref{lem:mev:basic:L-leq-H}, the proof follows directly from the inclusion $\mall{\CmvD}{\Adv}\subseteq \mall{\CmvDi}{\Adv}$.

    \medskip\noindent
	For~\Cref{lem:mev:basic:monotonicity},
    Let $\TxTS \in \mall{\CmvD}{\Adv}^*$ be a sequence of transactions that maximizes the loss $-\gain{\CmvC}{\sysS}{\TxTS}$. 
    Without loss of generality, assume that $\TxTS$ is valid in $\sysS$: in particular, this implies that $\TxTS$ never triggers calls (either external or internal) to methods of contracts not in $\sysS$ 
    (otherwise, such call would revert, and the sequence would not be valid).
	Therefore, the contract state $\cstD$ is not affected by the execution of $\TxTS$, that is:  
    \[
    \sysS \xrightarrow{\TxTS} \sysSi
    \implies 
    \sysS \mid \cstD \xrightarrow{\TxTS} \sysSi \mid \cstD
    \]
	Therefore: 
	\begin{align*}
		\gain{\CmvC}{\sysS \mid \cstD}{\TxTS} 
        & = \wealth{\CmvC}{\sysSi \mid \cstD} - \wealth{\CmvC}{\sysS \mid \cstD}
        \\
		& =	\wealth{\CmvC}{\sysSi} + \wealth{\CmvC}{\cstD} - (\wealth{\CmvC}{\sysS} + \wealth{\CmvC}{\cstD})
		\\
		& = \wealth{\CmvC}{\sysSi} - \wealth{\CmvC}{\sysS} 
		\\
		& = \gain{\CmvC}{\sysS}{\TxTS}
	\end{align*}
    which shows the inequality
    $\lmev{\CmvD}{\sysS}{\CmvC} \leq \lmev{\CmvD}{\sysS \mid \cstD}{\CmvC}$.
    \qed
    
\end{proofof}

\begin{proofof}{lem:mev:state-narrowing}
    Let $\CmvD \subseteq \cmvOfcst{\sysS}$.
    We must prove that
    $\lmev{\CmvD}{\sysS \mid \cstD}{\CmvC} = \lmev{\CmvD}{\sysS}{\CmvC}$.
    The inequality $\geq$ follows from~\Cref{lem:mev:basic:monotonicity} of~\Cref{lem:mev:basic}.
    For the inequality $\leq$, 
    let $\TxTS \in \mall{\CmvD}{\Adv}^*$ be a valid sequence of transactions that maximizes $-\gain{\CmvC}{\sysS \mid \cstD}{\TxTS}$. 
    We show that $\TxTS$ causes the same loss to $\CmvC$ in $\sysS$.
    Since the transactions in $\TxTS$ target contracts in $\CmvD \subseteq \cmvOfcst{\sysS}$ and since, by the well-formedness of $\sysS$, the contracts in $\sysS$ have no dependencies in $\cstD$, then the contracts in $\cstD$ are not affected by $\TxTS$.
    Hence, executing $\TxTS$ yields a transition of the form:
    \[
    \sysS \mid \cstD 
    \; \xrightarrow{\TxTS} \;
    \sysSi \mid \cstD
    \]
    Since $\TxTS$ does not include any direct/indirect calls to the contracts in $\cstD$, then $\TxTS$ is also valid in $\sysS$, producing the same effect as in $\sysS \mid \cstD$, namely: 
    \[
    \sysS 
    \; \xrightarrow{\TxTS} \; 
    \sysSi
    \]
    To prove that the loss is preserved, observe that:
    \begin{align*}
    \gain{\CmvC}{\sysS}{\TxTS}
    & = \wealth{\CmvC}{\sysSi} - \wealth{\CmvC}{\sysS}
    \\
    & = \wealth{\CmvC}{\sysSi} + \wealth{\CmvC}{\cstD} - \wealth{\CmvC}{\sysS} - \wealth{\CmvC}{\cstD}
    \\
    & = \wealth{\CmvC}{\sysSi \mid \cstD} - \wealth{\CmvC}{\sysS \mid \cstD}
    \\
    & = \gain{\CmvC}{\sysS \mid \cstD}{\TxTS}
    \end{align*}
    This implies the thesis:
    \[
    \lmev{\CmvD}{\sysS \mid \cstD}{\CmvC} \leq \lmev{\CmvD}{\sysS}{\CmvC}
    \tag*{\qed}
    \]
\end{proofof}

\begin{proofof}{th:mev:callable-narrowing:deps}
First, note that the inequality $\lmev{\strip{\CmvD}{\CmvC}}{\sysS}{\CmvC} \leq \lmev{\CmvD}{\sysS}{\CmvC}$ follows from \Cref{lem:mev:basic:L-leq-H} of \Cref{lem:mev:basic}, so we just need to show that:
\[
\lmev{\CmvD}{\sysS}{\CmvC} \leq \lmev{\strip{\CmvD}{\CmvC}}{\sysS}{\CmvC}
\]
Let $\TxTS \in \mall{\CmvD}{\Adv}^*$ be a sequence of transactions that maximizes the loss of $\CmvC$ when executed in state $\sysS$. 
We show that there exists $\TxYS \in \mall{\strip{\CmvD}{\CmvC}}{\Adv}^*$ that causes a loss to $\CmvC$ equal to the one caused by $\TxTS$, namely:
\begin{equation}
\label{eq:mev:callable-narrowing:deps:Y}
\TxYS \in \mall{\strip{\CmvD}{\CmvC}}{\Adv}^*
\qquad\qquad
\gain{\CmvC}{\sysS}{\TxYS} = \gain{\CmvC}{\sysS}{\TxTS}
\end{equation}
\Wlog we assume that all the transactions in $\TxTS$ are valid:
indeed, invalid transactions in $\TxTS$ are reverted, so they can be removed without affecting the loss.

\medskip\noindent
Note that each transaction $\txT[i] = \pmvM[i]:\contract[i,1]{D}.\txcode{f_{i,1}}(\code{args_{i,1}})$ in $\TxTS$ 
can trigger a sequence of \emph{internal} contract-to-contract method calls:
\[
\contract[i,1]{C}:\contract[i,1]{D}.\txcode{f_{i,1}}(\code{args_{i,1}})
\;\;
\contract[i,2]{C}:\contract[i,2]{D}.\txcode{f_{i,2}}(\code{args_{i,2}})
\;\; \cdots \;\;
\contract[i,k]{C}:\contract[i,k]{D}.\txcode{f_{i,k}}(\code{args_{i,k}})
\]
Let $\vec{x}$ be the sequence of all method calls (either external or internal) that are performed upon the execution of $\TxTS$ in state $\sysS$.
To construct $\TxYS$, we start by considering
the subsequence $\vec{y}$ of $\vec{x}$ containing all and only the calls of the form:

\vbox{%
\begin{enumerate}[(A)]

\item \label{item:mev:callable-narrowing:deps:Y-A}
$\pmvM[i]:\contract[i,1]{D}.\txcode{f_{i,1}}(\code{args_{i,1}})$ where $\contract[i,1]{D} \in \deps{\CmvC}$, or

\item \label{item:mev:callable-narrowing:deps:Y-B}
$\contract[i,j]{C}:\contract[i,j]{D}.\txcode{f_{i,j}}(\code{args_{i,j}})$, where $\contract[i,j]{C} \not\in \deps{\CmvC}$ and $\contract[i,j]{D} \in \deps{\CmvC}$.



\end{enumerate}}

\noindent
\textbf{Claim (1).} If $\contract[i,j]{C}:\contract[i,j]{D}.\txcode{f_{i,j}}(\code{args_{i,j}}) \in \vec{y}$, 
then $\contract[i,j]{D} \in \boundary{\CmvC}{\CmvD}$.

\medskip\noindent
\emph{Proof.} 
By hypothesis, $\contract[i,j]{D} \in \deps{\CmvC}$ 
and $\contract[i,j]{C}$ calls $\contract[i,j]{D}$.
Since the first contract called in the transaction, that is $\contract[i,1]{D}$, belongs to $\CmvD$, then
$\contract[i,j]{C} \in \deps{\CmvD}$.
By hypothesis, $\contract[i,j]{C} \not\in \deps{\CmvC}$.
Therefore, $\contract[i,j]{C} \in \deps{\CmvD} \setminus \deps{\CmvC}$.
By definition of $\boundary{\CmvC}{\CmvD}$, this concludes the proof of Claim (1).

\medskip
To describe the construction of $\TxYS$, let the meta-variables $\addr[i]{a}$ range over user and contract addresses, so to rewrite the sequence $\vec{y}$ as follows:
\begin{align*}
& \addr[1]{a}:\contract[1]{C}.\txcode{f_{1}}(\code{args_{1}})
\;\;
\addr[2]{a}:\contract[2]{C}.\txcode{f_{2}}(\code{args_{2}})
\;\; \cdots \;\;
\addr[n]{a}:\contract[n]{C}.\txcode{f_{n}}(\code{args_{n}})
\cdots
\end{align*}

To construct $\TxYS$, we translate each call $\addr[i]{a}:\contract[i]{C}.\txcode{f_{i}}(\code{args_{i}})$ in $\vec{y}$ into a transaction as follows, depending on whether the call is due to conditions~\ref{item:mev:callable-narrowing:deps:Y-A} or~\ref{item:mev:callable-narrowing:deps:Y-B}:
\begin{itemize}

\item [\ref{item:mev:callable-narrowing:deps:Y-A}] 
In this case, $\addr[i]{a} = \pmvM[i]$, hence the call $\addr[i]{a}:\contract[i]{C}.\txcode{f_{i}}(\code{args_{i}})$ is already a transaction, which we preserve in $\TxYS$.

\item[\ref{item:mev:callable-narrowing:deps:Y-B}] In this case, we
replace the contract account $\addr[i]{a}$ in the call
\(
\addr[i]{a}:\contract[i]{C}.\txcode{f_{i}}(\code{args_{i}})
\)
in $\vec{y}$
with the user account $\pmvM[i]$ that originated the corresponding call in $\vec{y}$.

\end{itemize}

\medskip\noindent
\textbf{Claim (2).} $\TxYS \in \mall{\strip{\CmvD}{\CmvC}}{\Adv}^*$
    
\medskip\noindent
\emph{Proof.} 
Consider a transaction $\txY[i]$ in $\TxYS$.
We have two cases, depending on whether $\txY[i]$ is due to conditions~\ref{item:mev:callable-narrowing:deps:Y-A} or~\ref{item:mev:callable-narrowing:deps:Y-B} given before:
\begin{enumerate}

\item[\ref{item:mev:callable-narrowing:deps:Y-A}] in this case, $\txY[i]$ is equal to some
$\txT[j] = \pmvM[j]:\contract[j,1]{D}.\txcode{f_{j,1}}(\code{args_{j,1}})$ in $\TxTS$ where $\contract[j,1]{D} \in \deps{\CmvC}$.
Since $\txT[j] \in \mall{\CmvD}{\Adv}$, then 
$\contract[j,1]{D} \in \CmvD$, and so
$\txY[i] \in \mall{\strip{\CmvD}{\CmvC}}{\Adv}$. 

\item[\ref{item:mev:callable-narrowing:deps:Y-B}] by Claim (1), the callee of $\txY[i]$ is in $\boundary{\CmvC}{\CmvD}$, which is included in $\CmvD$ by assumption~\eqref{th:mev:callable-narrowing:deps:2}, and is included in $\deps{\CmvC}$ by definition of $\boundary{\CmvC}{\CmvD}$. 
This implies that $\txY[i] \in \mall{\strip{\CmvD}{\CmvC}}{\Adv}$, since $\Adv$ is able to infer the actual arguments of the involved call by simulating the execution of $\TxTS$ in $\sysS$. 
This completes the proof of Claim (2).

\end{enumerate}

To conclude, we show that $\TxYS$ is valid and modifies the state of the contracts in $\CmvC$ in exactly the same way as $\TxTS$.
Together with Claim (2), this will prove~\eqref{eq:mev:callable-narrowing:deps:Y}.

Observe that the sequence of calls due to the execution of $\TxYS$ coincides with the subsequence of $\vec{x}$ including all and only the calls to contracts in $\deps{\CmvC}$ --- with the only difference of the change of sender due to condition~\ref{item:mev:callable-narrowing:deps:Y-B}.
We show, however, that this difference does not affect neither the validity of the transactions nor their effect on the state update.
Indeed, we have that:

\begin{enumerate}

\item Each transaction $\txY[i]$ in $\TxYS$ can be funded by the adversary.
To show this, let $\txY[i]$ be derived from $\addr[i]{a}:\contract[i]{C}.\txcode{f_{i}}(\code{args_{i}})$. We have the following cases:
\begin{itemize}

\item $\addr[i]{a} = \pmvM[i]$. 
By contradiction, assume that in the state where that transaction is executed, $\pmvM[i]$ does not have the tokens needed to fund the call to  $\contract[i]{C}$.
Since the corresponding transaction $\txT[j]$ in $\TxTS$ was valid, this means that $\pmvM[i]$ had received some of the tokens needed to fund the call from a previous transaction $\txT[h] = \pmvM[h]:\contract[h]{D}.\txcode{f_{h}}(\code{args_{h}})$ in $\TxTS$ that has no counterpart in $\TxYS$.
Note that:
\begin{itemize}

\item $\contract[h]{D} \in \CmvD \setminus \deps{\CmvC}$ because condition~\ref{item:mev:callable-narrowing:deps:Y-A} is false for $\txT[h]$; 

\item $\contract[i]{C} \in \CmvD \cap \deps{\CmvC}$ by Claim (2).

\end{itemize}   

By assumption~\eqref{th:mev:callable-narrowing:deps:3}, there is no token flow from $\CmvD \setminus \deps{\CmvC}$ to $\CmvD \cap \deps{\CmvC}$ 
in $\sysS$, which implies that no token that $\txT[h]$ has sent to $\pmvM[i]$ is needed to call $\contract[i]{C}$
--- contradiction.

\item $\addr[i]{a} = \contract[i]{D}$. 
We have that:
\begin{itemize}

\item $\contract[i]{D} \in \deps{\CmvD} \setminus \deps{\CmvC}$. 
The fact that $\contract[i]{D} \in \deps{\CmvD}$ holds because it derives from a transaction in $\mall{\CmvD}{\Adv}$, and $\contract[i]{D} \not\in \deps{\CmvC}$ follows by condition~\ref{item:mev:callable-narrowing:deps:Y-B};

\item $\contract[i]{C} \in \boundary{\CmvC}{\CmvD}$ by Claim (1).

\end{itemize}

By assumption~\eqref{th:mev:callable-narrowing:deps:4}, 
there is no direct token flow from $\deps{\CmvD} \setminus \deps{\CmvC}$ 
to $\boundary{\CmvC}{\CmvD}$ in $\sysS$.
Therefore, there is no token transfer from $\contract[i]{D}$ to $\contract[i]{C}$, and so $\txY[i]$ does not need to be funded.

\end{itemize}

\item The transactions $\txY[i]$ that are in $\TxYS$ due to condition~\ref{item:mev:callable-narrowing:deps:Y-B} have callee in $\boundary{\CmvC}{\CmvD}$ by Claim (1), and so their methods are \emph{sender-agnostic} by assumption~\eqref{th:mev:callable-narrowing:deps:1}. 
So, the fact that in the execution of $\txY[i]$ they are called directly from a user address, while in the execution of $\txT[i]$ they are called from a contract address, does not affect the execution of these calls.
Note that the adversary could receive tokens from these calls, namely those tokens which are transferred to the sender, but from the previous item such tokens are irrelevant to the validity of $\TxYS$.

\item The transaction $\pmvM[i]:\contract[i,j]{C}.\txcode{f_{i,j}}(\code{args_{i,j}})$ that are in $\TxYS$ due to condition~\ref{item:mev:callable-narrowing:deps:Y-B} may transfer tokens to $\pmvM[i]$,
while the corresponding calls $\contract[i,j-1]{C}:\contract[i,j]{C}.\txcode{f_{i,j}}(\code{args_{i,j}})$ due to the execution of $\TxTS$ may send these tokens to the sender $\contract[i,j-1]{C}$, thus affecting its gain. 
This difference however does not affect the loss of $\CmvC$, since $\contract[i,j-1]{C}$ is not in $\deps{\CmvC}$ by condition~\ref{item:mev:callable-narrowing:deps:Y-B}.

\end{enumerate}

For the reasons above, $\TxYS$ is valid, and it updated the state of contracts $\deps{\CmvC}$ in exactly the same way as $\TxTS$ and $\TxYS$ --- and, in particular, $\TxTS$ and $\TxYS$ cause exactly the same losses to the contracts in $\CmvC$, \ie $\gain{\CmvC}{\sysS}{\TxYS} = \gain{\CmvC}{\sysS}{\TxTS}$. 
This proves~\eqref{eq:mev:callable-narrowing:deps:Y}, from which we obtain the thesis.
\qed
\end{proofof}

\begin{example}
\label{cex:mev:callable-narrowing:sender-agnostic}
To illustrate the need of the sender-agnosticism assumption~\eqref{th:mev:callable-narrowing:deps:1} in \Cref{th:mev:callable-narrowing:deps}, consider the contracts: 

\begin{lstlisting}[
,language=solidity
%,basicstyle=\fontseries{m}\normalsize\ttfamily\lst@ifdisplaystyle\footnotesize\fi,
%,morekeywords={f,g,f0,f1,g0,g1},classoffset=4
%,morekeywords={a,A,M},keywordstyle=\pmvColor
%,classoffset=5,morekeywords={t,T,T0,T1,T2,ETH}
%,keywordstyle=\tokColor,classoffset=6
%,morekeywords={C0,C1,C2,C3},keywordstyle=\cmvColor
%,caption={Illustration of~\Cref{th:mev:callable-narrowing:deps}}
%,label={lst:cex:lmev:contract-stripping}
]
contract C0 { function f() { require (msg.sender==C1); T.send(M,1); }}
contract C1 { function f() { C0.f(); } }
\end{lstlisting}

\noindent
Let 
$\Adv = \setenum{\pmvM}$, 
$\CmvD= \setenum{\contract{C0},\contract{C1}}$, 
$\CmvC = \setenum{\contract{C0}}$, and let:
\[
\sysS 
\; = \;
\walpmv{\pmvM}{0:\tokT} \mid
\walpmv{\contract{C0}}{\waltok{1}{\tokT}} \mid 
\walpmv{\contract{C1}}{\waltok{0}{\tokT}}
\]
Let $\TxTS = \pmvM:\contract{C1}.\txcode{f()} \in \mall{\CmvD}{\pmvM}$.
By executing $\TxTS$ in $\sysS$, we have that:
\begin{align*}
    \sysS
    & \xrightarrow{\pmvM:\contract{C1}.\txcode{f()}}
    \walpmv{\pmvM}{1:\tokT} \mid
    \walpmv{\contract{C0}}{\waltok{0}{\tokT}} \mid 
    \walpmv{\contract{C1}}{\waltok{0}{\tokT}}
\end{align*}
Since there are no tokens left in $\CmvC$, $\TxTS$ clearly maximises the loss of $\CmvC$, hence
\(
\lmev{\CmvD}{\sysS}{\CmvC} = 1 \cdot \price{\tokT}
\).
Since $\contract{C0} \in \boundary{\CmvC}{\CmvD}$ is not sender-agnostic,
narrowing the set of callable contracts to
$\CmvD \cap \deps{\CmvC} = \setenum{\contract{C0}}$ is not guaranteed to preserve the MEV.
Indeed, $\lmev{\setenum{\contract{C0}}}{\sysS}{\CmvC} = 0$,
since the adversary can only call $\contract{C0}$, but the transaction would revert since the \lstinline[language=solidity]{require} condition in $\contract{C0}$ is violated.
\hfill\qedex
\end{example}

\begin{example}
\label{cex:mev:callable-narrowing:boundary}
To illustrate the need of the assumption \eqref{th:mev:callable-narrowing:deps:2} in \Cref{th:mev:callable-narrowing:deps}, consider the contracts: 

\begin{lstlisting}[
,language=solidity
%,basicstyle=\fontseries{m}\normalsize\ttfamily\lst@ifdisplaystyle\footnotesize\fi,
%,morekeywords={f,g,f0,f1,g0,g1},classoffset=4
%,morekeywords={a,A,M},keywordstyle=\pmvColor
%,classoffset=5,morekeywords={t,T,T0,T1,T2,ETH}
%,keywordstyle=\tokColor,classoffset=6
%,morekeywords={C0,C1,C2,C3},keywordstyle=\cmvColor
%,caption={Illustration of~\Cref{th:mev:callable-narrowing:deps}}
%,label={lst:cex:lmev:contract-stripping}
]
contract C0 { function f() { T.send(M,1); } }
contract C1 { function f() { C0.f(); } }
\end{lstlisting}

\noindent
Let 
$\Adv = \setenum{\pmvM}$, 
$\CmvD= \setenum{\contract{C1}}$, 
$\CmvC = \setenum{\contract{C0}}$, and let:
\[
\sysS = 
\walpmv{\pmvM}{0:\tokT} \mid
\walpmv{\contract{C0}}{\waltok{1}{\tokT}} \mid 
\walpmv{\contract{C1}}{\waltok{0}{\tokT}}
\]
Let $\TxTS = \pmvM:\contract{C1}.\txcode{f()} \in \mall{\CmvD}{\pmvM}$.
By executing $\TxTS$ in $\sysS$, we have that:
\begin{align*}
    \sysS
    & \xrightarrow{\pmvM:\contract{C1}.\txcode{f()}}
    \walpmv{\pmvM}{1:\tokT} \mid
    \walpmv{\contract{C0}}{\waltok{0}{\tokT}} \mid 
    \walpmv{\contract{C1}}{\waltok{0}{\tokT}}
\end{align*}
Since there are no tokens left in $\CmvC$, $\TxTS$ clearly maximises the loss of $\CmvC$, hence
\(
\lmev{\CmvD}{\sysS}{\CmvC} = 1 \cdot \price{\tokT}
\).
Note that conditions \eqref{th:mev:callable-narrowing:deps:1},  \eqref{th:mev:callable-narrowing:deps:3} and~\eqref{th:mev:callable-narrowing:deps:4} are trivially satisfied, 
while $\boundary{\CmvC}{\CmvD} = \setenum{\contract{C0}} \not\subseteq \CmvD$, thus violating condition \eqref{th:mev:callable-narrowing:deps:2}. 
Therefore, narrowing the set of callable contracts to
$\CmvD \cap \deps{\CmvC} = \emptyset$ does not preserve the MEV: indeed, such MEV is zero, since there are no callable contracts.
\hfill\qedex
\end{example}

\begin{example}
\label{cex:mev:callable-narrowing:no-token-flows}
To illustrate the need of the token flows assumptions in \Cref{th:mev:callable-narrowing:deps}, consider the following contracts: 
\begin{lstlisting}[
,language=solidity
%,basicstyle=\fontseries{m}\normalsize\ttfamily\lst@ifdisplaystyle\footnotesize\fi,
%,morekeywords={f,g,f0,f1,g0,g1,receive},classoffset=4
%,morekeywords={a,A,M},keywordstyle=\pmvColor
%,classoffset=5,morekeywords={t,T,T0,T1,T2,ETH}
%,keywordstyle=\tokColor,classoffset=6
%,morekeywords={C0,C1,C2,C3},keywordstyle=\cmvColor
%,caption={Necessity of token flows assumptions in~\Cref{th:mev:callable-narrowing:deps}.}
%,label={lst:mev:callable-narrowing:deps:token-flows}
]
contract C0 { function f(){ T.receive(2); T.send(M,4); } }
contract C1 { function f(){ T.receive(1); T.send(C0,2); C0.f(); } }
contract C2 { function f(){ T.send(C1,1); C1.f(); } }
\end{lstlisting}
%
Let 
$\Adv = \setenum{\pmvM}$,
$\CmvD= \setenum{\contract{C0},\contract{C1},\contract{C2}}$, 
$\CmvC = \setenum{\contract{C0},\contract{C1}}$, and let:
\[
\sysS 
\; = \; 
\walpmv{\pmvM}{0:\tokT} \mid
\walpmv{\contract{C0}}{\waltok{2}{\tokT}} \mid 
\walpmv{\contract{C1}}{\waltok{1}{\tokT}} \mid 
\walpmv{\contract{C2}}{\waltok{1}{\tokT}}
\]
Let $\TxTS = \pmvM:\contract{C2}.\txcode{f()} \in \mall{\CmvD}{\pmvM}$.
By executing $\TxTS$ in $\sysS$, we obtain:
\begin{align*}
    \sysS
    & \xrightarrow{\pmvM:\contract{C3}.\txcode{f()}}
    \walpmv{\pmvM}{4:\tokT} \mid
    \walpmv{\contract{C0}}{\waltok{0}{\tokT}} \mid 
    \walpmv{\contract{C1}}{\waltok{0}{\tokT}} \mid 
    \walpmv{\contract{C2}}{\waltok{0}{\tokT}}
\end{align*}
Since there are no tokens left in $\CmvC$, $\TxTS$ clearly maximises the loss of $\CmvC$.
Hence, 
\(
\lmev{\CmvD}{\sysS}{\CmvC} = 3 \cdot \price{\tokT}
\).
Similarly to~\Cref{ex:mev:callable-narrowing:deps}, we have that $\boundary{\CmvC}{\CmvD} = \setenum{\contract{C1}}$, which satisfies conditions \eqref{th:mev:callable-narrowing:deps:1} and \eqref{th:mev:callable-narrowing:deps:2} of~\Cref{th:mev:callable-narrowing:deps}.
Instead, condition \eqref{th:mev:callable-narrowing:deps:4} is violated, since there is a direct token flow of $1:\tokT$ from
$\contract{C2} \in \deps{\CmvD} \setminus \deps{\CmvC}$ to
$\contract{C1} \in \boundary{\CmvC}{\CmvD}$.
Therefore, narrowing the set of callable contracts to
$\CmvD \cap \deps{\CmvC} = \setenum{\contract{C0},\contract{C1}}$
is not guaranteed to preserve the MEV.
Indeed, $\pmvM$ has no tokens in $\sysS$: hence, 
calling $\contract{C0}$ fails, since $\contract{C0}$ has not the required $4:\tokT$ to transfer;
similarly, calling $\contract{C1}$ fails, since $\contract{C1}$ does not have the $2:\tokT$ required to call $\contract{C0}$. 
Therefore, $\lmev{\setenum{\contract{C0},\contract{C1}}}{\sysS}{\CmvC} = 0 < \lmev{\CmvD}{\sysS}{\CmvC}$.
\hfill\qedex
\end{example}

\section{Supplementary material for~\Cref{sec:rich-mev}} 
\label{sec:proofs:rich-mev}

\begin{proofof}{lem:mev:wallet}
	For~\Cref{lem:mev:wallet:zero-user}, let $\Adv \cap \dom{\WmvA} = \emptyset$.
    Since 
    \begin{inlinelist}
    \item $\Adv$ cannot spend tokens from the wallets in $\WmvA$, and
    \item the presence or absence of a user account is irrelevant to the effect transactions originated from other accounts  
    \end{inlinelist},
    then any sequence of transactions $\TxTS \in\mall{\CmvD}{\Adv}^*$ is valid in $\sysS$ if and only if it is valid in $\sysS \mid \WmvA$.
	By wallet monotonicity (\Cref{lem:wallet-monotonicity}), if $\TxTS$ is valid (in either state), then
	\[
		\sysS \xrightarrow{\TxTS}  \WmvAi \mid \cstCi  
		\implies
		\sysS \mid \WmvA \xrightarrow{\TxTS}  (\WmvAi + \WmvA) \mid \cstCi
	\] 
	From this, we obtain that $\TxTS$ causes the same loss to $\CmvC$ in both $\sysS$ and $\sysS \mid \WmvA$:
	\begin{align*}
	\gain{\CmvC}{\sysS}{ \TxTS} 
    & = \wealth{\CmvC}{\WmvAi \mid \cstCi } -  \wealth{\CmvC}{\sysS} 
    \; = \;
    \wealth{\CmvC}{\cstCi } -  \wealth{\CmvC}{\sysS}
	\\
	& = \wealth{\CmvC}{\WmvAi+ \WmvA \mid \cstCi } - \wealth{\CmvC}{\sysS \mid \WmvA}
    \; = \; 
    \gain{\CmvC}{\sysS \mid \WmvA}{ \TxTS}
	\end{align*}
	Therefore, we have the thesis $\lmev{\CmvD}{\sysS}{\CmvC} =\lmev{\CmvD}{\sysS \mid \WmvA}{\CmvC}$.

    \medskip
	For~\Cref{lem:mev:wallet:monotonicity}, let $\TxTS \in \mall{\CmvD}{\Adv}^*$ maximize the loss of ${\CmvC}$. 
    Without loss of generality, assume that $\TxTS$ is valid in $\sysS$.
	By wallet monotonicity, $\TxTS$ is also valid in $\sysSi$ and it extracts the same amount of tokens from $\CmvC$, giving us a lower bound for $\lmev{\CmvD}{\sysSi}{\CmvC}$.\qed
\end{proofof}

\begin{proofof}{lem:mev:stability}
    Fix $\CmvC, \CmvD$ and $\cstC$, and let 
    $f(\wmvA) = \lmev{\CmvD}{\wmvA \mid \cstC }{\CmvC}$.
    By \Cref{lem:mev:basic:leq-wealth} of \Cref{lem:mev:basic}, $f(w)$ is bounded from above by $\wealth{\CmvC}{\wmvA \mid \cstC}$.
	By \Cref{def:wealth}, the wealth $\wealth{\CmvC}{\wmvA \mid \cstC}$ only depends on the contract state $\cstC$, so it is equal to $\wealth{\CmvC}{\cstC}$.
	Hence, $f(\wmvA)$ is an integer value bounded from above, so there exists a wallet $\WmvA$ for which the maximum is reached.  
	By \Cref{lem:mev:wallet:monotonicity} of \Cref{lem:mev:wallet}, the same MEV is also achieved with a wallet extension $\WmvAi$ of $\WmvA$. 
    \qed
\end{proofof}

\begin{proofof}{lem:rich-mev:basic}
	\Cref{lem:rich-mev:basic:zero,lem:rich-mev:basic:L-leq-H,lem:rich-mev:basic:monotonicity,lem:rich-mev:basic:garbage} have an analogous statement in \Cref{lem:mev:basic}, which holds for any wallet state. Due to the stability property (\Cref{lem:mev:stability}) and the definition of $\rlmev{}{}{}$ (\Cref{def:rich-mev}), the wealthy-adversary versions of the statements must also hold. 
	The same reasoning can be carried out for \Cref{lem:rich-mev:basic:leq-wealth} and its analogous in \Cref{lem:mev:basic}, since $\wealth{\CmvC}{\cstC}  = \wealth{\CmvC}{\WmvA \mid \cstC}$ for any $\WmvA$.
	\qed
\end{proofof}

\section{Supplementary material for~\Cref{sec:nonint}} 
\label{sec:proofs:nonint}

\begin{proofof}{lem:rich-nonint-nonint}
From \Cref{def:rich-mev} and \Cref{lem:mev:stability}, there exist wallets $\WmvA[1]$ and $\WmvA[2]$ such that:
\begin{align}
\label{eq:proof:lem:rich-nonint-nonint-1}
& \forall \WmvA \geq_{\$} \WmvA[1] \; : \; \rlmev{\cmvOfcst{\cstD}}{\cstC \mid \cstD}{\cmvOfcst{\cstD}} = \lmev{\cmvOfcst{\cstD}}{\WmvA \mid \cstC \mid \cstD}{\cmvOfcst{\cstD}}
\\
\label{eq:proof:lem:rich-nonint-nonint-2}
& \forall \WmvA \geq_{\$} \WmvA[2] \; : \; \rlmev{}{\cstC \mid \cstD}{\cmvOfcst{\cstD}} = \lmev{}{\WmvA \mid \cstC \mid \cstD}{\cmvOfcst{\cstD}}
\end{align}
Let $\WmvA[0] = \sup \setenum{\WmvA[1] , \WmvA[2]}$. 
Then, for all $\WmvA \geq_{\$} \WmvA[0]$:
\begin{align*}
& \richnonint{\cstC}{\cstD}
\\
\iff &
\rlmev{\cmvOfcst{\cstD}}{\cstC \mid \cstD}{\cmvOfcst{\cstD}} = \rlmev{}{\cstC \mid \cstD}{\cmvOfcst{\cstD}}
&& \text{by def. $\richnonint{}{}$}
\\
\iff 
& \lmev{\cmvOfcst{\cstD}}{\WmvA \mid \cstC \mid \cstD}{\cmvOfcst{\cstD}} = \lmev{}{\WmvA \mid \cstC \mid \cstD}{\cmvOfcst{\cstD}}
&& \text{by \eqref{eq:proof:lem:rich-nonint-nonint-1} and \eqref{eq:proof:lem:rich-nonint-nonint-2}}
\\
\iff
& \nonint{\WmvA \mid \cstC}{\cstD}
&& \text{by def. $\nonint{}{}$}
\end{align*}
\qed
\end{proofof}

\section{Supplementary material for~\Cref{sec:conclusions}} 
\label{sec:proofs:cex}

\begin{example}
	\label{ex:rich-nonint:mid-R}
	$\richnonint{\cstC}{\cstD}$ does not imply $\richnonint{\cstC}{\cstD \mid \cstDi}$.
	This can be simply seen when $\cstD = \emptyset$ (which is trivially non-interfering with any $\cstC$), and  $\cstDi$ is any contract that interferes with $\cstC$.
	For instance, let
	\[
		\cstC = \walpmv{\contract{X}}{\waltok{0}{\tokT}, \code{x}=0} \quad \cstD = \emptyset \quad \cstDi = \walpmv{\contract{C}}{\waltok{1}{\tokT}}
	\]
	where $\contract{X}$ and $\contract{C}$ are defined in \Cref{fig:rich-nonint:monotonicity}.
	We have that 
	\begin{align*}
		\rlmev{\cmvOfcst{\cstD}}{\cstC \mid \cstD }{\cmvOfcst{\cstD}} &= \rlmev{\emptyset}{\cstC \mid \cstD }{\emptyset} = 0
		\\
		\rlmev{}{\cstC \mid \cstD }{\cmvOfcst{\cstD}} &=  \rlmev{}{\cstC \mid \cstD }{\emptyset} = 0	\end{align*}
	so $\richnonint{\cstC}{\cstD}$. 
	To study the non-interference between $\cstC$ and $\cstD \mid \cstDi$, we note that the \emph{restricted} local MEV of $\contract{C}$ is 0: indeed, if the adversary cannot call methods of $\contract{X}$, then they have no way to change the value of the variable $\code{x}$, and the method $\contract{C}.\txcode{f}()$ will never be enabled. On the other hand, if the adversary has access to $\contract{X}$, she can perform the following sequence of transactions: $\pmvM: \contract{X}.\txcode{set}(1)\ \pmvM: \contract{C}.\txcode{f}()$, extracting $\waltok{1}{\tokT}$ from $\contract{C}$. 
	Therefore, $\negrichnonint{\cstC}{\cstD \mid \cstDi}$.
	Note that since the methods of $\cstDi$ do not depend on $\cstD$, we can change their order while preserving the well formedness of the state, so $\negrichnonint{\cstC}{\cstDi \mid \cstD}$. 
\end{example}

\begin{example}
	\label{ex:rich-nonint:mid-L}
	$\richnonint{\cstC}{\cstD \mid \cstDi}$ does not imply  $\richnonint{\cstC}{\cstD}$.
	For instance, let
	\[
		\cstC = \walpmv{\contract{X}}{\waltok{0}{\tokT}, \code{x}=0} \quad \cstD = \walpmv{\contract{C}}{\waltok{1}{\tokT}} \quad \cstDi = \walpmv{\contract{ForwardX}}{\waltok{0}{\tokT}} 
	\]
	using the contracts in \Cref{fig:rich-nonint:monotonicity}.
    From~\Cref{ex:rich-nonint:mid-R} we know that $\negrichnonint{\cstC}{\cstD}$. However, $\richnonint{\cstC}{\cstD \mid \cstDi}$. Indeed, it does not matter if the adversary has direct access to $\contract{X}$, since it can always call $\contract{ForwardX}$ to maximize the loss of $\cstD \mid \cstDi$. 
	More specifically, the sequence of transactions $\pmvM : \contract{ForwardX}.\txcode{set\_x(1)} \ \pmvM: \cmvC.\txcode{f}()$ causes a loss of $\waltok{1}{\tokT}$ in $\cmvC$, so
	\[
		\rlmev{\cmvOfcst{(\cstD \mid \cstDi)}}{\cstC \mid \cstD \mid \cstDi }{\cmvOfcst{(\cstD \mid \cstDi)}} = 1.
	\]
	The unrestricted MEV is also equal to 1 (it cannot be greater, since $\wealth{}{\cstD \mid \cstDi} = 1$), so we have 
	\[
		\rlmev{\cmvOfcst{(\cstD \mid \cstDi)}}{\cstC \mid \cstD \mid \cstDi }{\cmvOfcst{(\cstD \mid \cstDi)}} = 1 = \rlmev{}{\cstC \mid \cstD \mid \cstDi }{\cmvOfcst{(\cstD \mid \cstDi)}}
	\]
	and $\richnonint{\cstC}{ \cstD \mid \cstDi}$.
	Since there are no calls from $\cstDi$ to $\cstD$, we can also change their order while preserving the well formedness of the state, and conclude that
	$\richnonint{\cstC}{ \cstDi \mid \cstD}$.
\end{example}

\begin{figure}[t]
\begin{lstlisting}[language=solidity
    ,caption={Contracts for \Cref{ex:rich-nonint:mid-R,ex:rich-nonint:mid-L}.}
    ,label={fig:rich-nonint:monotonicity}]
contract X { 
  uint public x;
  constructor { x=0; }
  function set(uint y) public { x=y; }
}

contract C { 
  function f() { require (X.x()==1); T.send(msg.sender,1); }
}

contract ForwardX {
  function get_x()  { return X.x(); }
  function set_x(y) { X.set(y); }
}
\end{lstlisting}
\end{figure}

\begin{figure}[t]
  \begin{lstlisting}[language=solidity
  % ,morekeywords={get,set,drop2,drop3},classoffset=4,morekeywords={a,A,Oracle},keywordstyle=\pmvColor,classoffset=5,morekeywords={ETH},keywordstyle=\tokColor,classoffset=6,morekeywords={Var,Drop},keywordstyle=\cmvColor,frame=single
  ,caption={Contracts for~\Cref{ex:rich-nonint:union}.}
  ,label={lst:rich-nonint:union}
]
contract Var {
  uint public x;
  function set(uint y) public { if (x==0) x=y; }
}

contract Drop {
  uint b=0;   
  function drop2() { require (b==0 && Var.x()==1); b=1; T.send(msg.sender,2); }
  function drop3() { require (b==0 && Var.x()==0); b=1; Var.set(2); T.send(msg.sender,3); }
}
  \end{lstlisting}
\end{figure}

\begin{example}
  \label{ex:rich-nonint:union}
  $\richnonint{\cstC}{}{\cstD[1]}$ and
  $\richnonint{\cstC}{}{\cstD[2]}$
  does not imply
  $\richnonint{\cstC}{}{\cstD[1] \mid \cstD[2]}$.
  Let:
  \[
  \cstC =
  \walpmv{\contract{Var}}{\waltok{0}{\tokT},\code{x}=0}
  \quad
  \cstD[1] =
  \walpmv{\contract{Drop1}}{\waltok{3}{\tokT},\code{b}=0}
  \quad
  \cstD[2] =
  \walpmv{\contract{Drop2}}{\waltok{3}{\tokT},\code{b}=0}
  \]
  using the contracts in~\Cref{lst:rich-nonint:union}, where $\contract{Drop}$ is deployed twice in $\cstC$.
  The wealthy adversaries' MEV of both instances
  $\contract{Drop1}$ and $\contract{Drop2}$ in $\cstC$ is $3$,
  obtained by calling $\txcode{drop3}$ exactly once.
  This holds both for the unrestricted and restricted MEV,
  hence \mbox{$\richnonint{\cstC}{}{\cstD[1]}$} and \mbox{$\richnonint{\cstC}{}{\cstD[2]}$}.
  
  To study the composition of $\cstC$ with $\cstD[1] \mid \cstD[2]$, note that the unrestricted MEV of $\setenum{\contract{Drop1},\contract{Drop2}}$ is $4$, which can be obtained by the sequence:
  \[
  \pmvM:\contract{Var}.\txcode{set}(1) \;\;
  \pmvM:\contract{Drop1}.\txcode{drop2}() \;\;
  \pmvM:\contract{Drop2}.\txcode{drop2}()
  \]
  Instead the \emph{restricted} MEV of
  $\setenum{\contract{Drop1},\contract{Drop2}}$ is $3$,
  since calling $\contract{Var}$ is forbidden, 
  and $\txcode{drop3}$ can only be called on one of the instances.
  Then, $\negrichnonint{\cstC}{\cstD[1] \mid \cstD[2]}$.
  \hfill\qedex
\end{example}


\end{document}